\documentclass[11pt]{article}

\usepackage{amsfonts}
\usepackage{amssymb}
\usepackage{amstext}
\usepackage{amsmath}
\usepackage{enumerate}
\usepackage{algorithm}
\usepackage{algorithmic}
\usepackage[normalem]{ulem}
\usepackage{cancel}

\usepackage{multirow, makecell}
\usepackage{tablefootnote}

\usepackage{amsthm}

\usepackage{thmtools}
\usepackage{thm-restate}

\usepackage[pagebackref,colorlinks=true,pdfpagemode=UseNone,citecolor=OliveGreen,linkcolor=BrickRed,urlcolor=BrickRed,pdfstartview=FitH]{hyperref}
\usepackage{cleveref}

\usepackage[usenames,dvipsnames,svgnames,table]{xcolor}

\usepackage{graphicx}
\usepackage{subcaption}

\usepackage{framed}
\usepackage{enumitem}


\newtheorem{theorem}{Theorem}[section]

\newtheorem{lemma}[theorem]{Lemma}

\newtheorem{corollary}[theorem]{Corollary}

\newtheorem{proposition}[theorem]{Proposition}
\newtheorem{claim}[theorem]{Claim}
\newtheorem*{claim*}{Claim}

\newtheorem*{problem*}{Problem}

\newtheorem{remark}[theorem]{Remark}
\newtheorem*{remark*}{Remark}

\newtheorem{observation}[theorem]{Observation}
\newtheorem*{observation*}{Observation}
\newtheorem{example}[theorem]{Example}

\numberwithin{equation}{section}
\numberwithin{table}{section}

\renewcommand{\tilde}{\widetilde}

\newcommand{\R}{\ensuremath{\mathbb R}}
\newcommand{\mZ}{\ensuremath{\mathbb Z}}

\newcommand{\E}[1]{{\mathbb{E}}\left[#1\right]}

\newcommand{\junk}[1]{}

\newcommand{\norm}[1]{\left\lVert#1\right\rVert}

\newenvironment{proofof}[1]{{\medbreak\noindent \em Proof of #1.  }}{\hfill\qed\medbreak}

\def\eps{{\varepsilon}}

\DeclareMathAlphabet{\mathsfit}{T1}{\sfdefault}{\mddefault}{\sldefault}
\SetMathAlphabet{\mathsfit}{bold}{T1}{\sfdefault}{\bfdefault}{\sldefault}

\def\vb{{\mathsfit b}}
\def\vc{{\mathsfit c}}
\def\ve{{\mathsfit e}}
\def\vu{{\mathsfit u}}
\def\vv{{\mathsfit v}}
\def\vw{{\mathsfit w}}
\def\vx{{\mathsfit x}}
\def\vy{{\mathsfit y}}
\def\vz{{\mathsfit z}}

\def\mA{{\mathsfit A}}
\def\mB{{\mathsfit B}}
\def\mC{{\mathsfit C}}
\def\mE{{\mathsfit E}}
\def\mF{{\mathsfit F}}
\def\mI{{\mathsfit I}}

\def\mL{{\mathsfit L}}
\def\mM{{\mathsfit M}}
\def\mP{{\mathsfit P}}
\def\mU{{\mathsfit U}}
\def\mV{{\mathsfit V}}
\def\mX{{\mathsfit X}}
\def\mY{{\mathsfit Y}}
\def\mZ{{\mathsfit Z}}

\def\diag{\operatorname{diag}} % diagonal matrix
\def\Reff{\text{Reff}}

\def\tr{\operatorname{tr}}

\global\long\def\E{\mathbb{E}}
\global\long\def\R{\mathbb{R}}

\newcommand{\inner}[2]{\langle #1, #2 \rangle} % inner product
\DeclareMathOperator{\argmax}{argmax}
\DeclareMathOperator{\argmin}{argmin}

\renewcommand{\preceq}{\preccurlyeq}
\renewcommand{\succeq}{\succcurlyeq}

\title{A Local Search Framework for Experimental Design}

\author{Lap Chi Lau\footnote{School of Computer Science, University of Waterloo. Supported by NSERC Discovery Grant 2950-120715. Email: \href{mailto:lapchi@uwaterloo.ca}{lapchi@uwaterloo.ca}},~~~~~
Hong Zhou\footnote{School of Computer Science, University of Waterloo. Supported by NSERC Discovery Grant 2950-120715. Email: \href{mailto:h76zhou@uwaterloo.ca}{h76zhou@uwaterloo.ca}}}

\date{}

\begin{document}

\begin{titlepage}
\def\thepage{}
\thispagestyle{empty}

\maketitle

\begin{abstract}
We present a local search framework to design and analyze both combinatorial algorithms and rounding algorithms for experimental design problems.
This framework provides a unifying approach to match and improve all known results in D/A/E-design and to obtain new results in previously unknown settings.  
\begin{itemize}
\item
For combinatorial algorithms, we provide a new analysis of the classical Fedorov's exchange method.
We prove that this simple local search algorithm works well as long as there exists an almost optimal solution with good condition number.
Moreover, we design a new combinatorial local search algorithm for E-design using the regret minimization framework.
\item
For rounding algorithms, we provide a unified randomized exchange algorithm to match and improve previous results for D/A/E-design. 
Furthermore, the algorithm works in the more general setting to approximately satisfy multiple knapsack constraints, which can be used for weighted experimental design and for incorporating fairness constraints into experimental design.
\end{itemize}
\end{abstract}

\end{titlepage}

\thispagestyle{empty}

%\tableofcontents

\section{Introduction} \label{s:intro}

In experimental design problems,
we are given vectors $\vv_1,\ldots,\vv_n \in \R^d$ and a parameter $b \geq d$,
and the goal is to choose a (multi-)subset $S$ of $b$ vectors so that $\sum_{i \in S} \vv_i \vv_i^\top$ optimizes some objective function.
The most popular and well-studied objective functions are:
\vspace{-2mm}
\begin{itemize}
\setlength\itemsep{-1pt}
\item D-design: Maximizing $\left(\det\left( \sum_{i \in S} \vv_i \vv_i^\top \right)\right)^{\frac{1}{d}}$.
\item A-design: Minimizing $\tr\left( \left(\sum_{i \in S} \vv_i \vv_i^\top \right)^{-1} \right)$.
\item E-design: Maximizing $\lambda_{\min}\left(\sum_{i \in S} \vv_i \vv_i^\top\right)$.
\end{itemize}
\vspace{-2mm}
Two settings are studied in the literature.
One is the ``with repetition'' setting where each vector is allowed to be chosen multiple times,
and the other is the ``without repetition'' setting where each vector is allowed to be chosen at most once.
There is a simple reduction from the with repetition setting to the without repetition setting. 
All the results in this paper apply in the more general without repetition setting.

These problems of choosing a representative subset of vectors have a wide range of applications.
\vspace{-2mm}
\begin{itemize}
\setlength\itemsep{-1pt}
\item
Experimental design is a classical topic in statistics with extensive literature~\cite{Fed72,AD92,Puk06,GJ11},
where the goal is to choose $b$ (noisy) linear measurements from $\vv_1,\ldots,\vv_n \in \R^d$ so as to maximize the statistical efficiency of estimating an unknown vector in $\R^d$.
\item
In machine learning, they are used in active learning~\cite{Ang88}, feature selection~\cite{BM13}, and data summarization~\cite{Mir17,CKS+18}. 
\item
In numerical linear algebra, they are used in column subset selection~\cite{AB13}, sparse least square regression~\cite{BDM11}, and matrix approximation~\cite{DRVW06,DV06}.
\item 
In signal processing, they are used in sensor placement problems~\cite{JB09}, and optimal subsampling in graph signal processing~\cite{CR18,CSMK15,CSMK16}. 
\item
In network design, the problem of choosing a subgraph with at most $b$ edges to minimize the total effective resistance~\cite{GBS08,LZ20} is an A-design problem,
and the problem of choosing a subgraph with at most $b$ edges to maximize the algebraic connectivity~\cite{GB06,KMS+10,LZ20} is an E-design problem.
\end{itemize}
\vspace{-2mm}
We refer the interested reader to~\cite{WYS17,SX18,MSTX19,NST19,AZLSW20} for more discussions of these applications and further references on related work.

\subsection{Our Results} \label{ss:results}

We present both combinatorial algorithms and rounding algorithms for experimental design problems.
A main contribution in this paper is to show that these two types of algorithms can be analyzed using the same local search framework.
Using this framework, we match and improve all known results
and also obtain some new results.

\subsubsection{Combinatorial Algorithms} \label{ss:combin}

The Fedorov's exchange method~\cite{Fed72} starts with an arbitrary initial set $S_0$ of $b$ vectors, and in each step $t \geq 1$ it aims to exchange one of the vectors, $S_{t} \gets S_{t-1} - \vv_i + \vv_j$ where $\vv_i \in S_{t-1}$ and $\vv_j \notin S_{t-1}$, to improve the objective value, and stops if such an improving exchange is not possible.
The simplicity of this algorithm and its good empirical performance~\cite{CN80,MN94,NM92} make the method widely used~\cite{ADT07}.
The approximation guarantee of this method is only analyzed rigorously in a recent work~\cite{MSTX19},
and we extend their analysis in multiple directions.

For D-design, it was proved in~\cite{MSTX19} that Fedorov's exchange method gives a polynomial time approximation algorithm for all inputs in the with repetition setting, and we extend their result to the without repetition setting.

\begin{restatable}{theorem}{CombD}\label{t:D-combin}
% \begin{theorem} 
The Fedorov's exchange method is a polynomial time $\frac{b-d-1}{b}$-approximation algorithm for D-design in the without repetition setting.
In particular, this is a $(1-\eps)$-approximation algorithm whenever $b \geq d + 1+ \frac{d}{\eps}$ for any $\eps > 0$.
\end{restatable}

For A-design, it was shown in~\cite{MSTX19} that there are arbitrarily bad local optimal solutions for the Fedorov's exchange method.
Interestingly, we prove that Fedorov's exchange method works well as long as there exists an almost optimal solution with good condition number.
This provides a new insight about when the local search method works well,
and this condition may hold in practical instances.
As a corollary, this also extends the analysis of Fedorov's exchange method in~\cite{MSTX19} when all the vectors are short to the without repetition setting (see Section~\ref{ss:comb-A}).

\begin{restatable}{theorem}{CombA} \label{t:A-combin}
Let $\mX := \sum_{i=1}^n \vx(i) \cdot \vv_i \vv_i^\top$ with $\sum_{i=1}^n \vx(i) = b$ and $x_i \in [0,1]$ for $1 \leq i \leq n$ be a fractional solution to A-design. 
For any $\eps \in (0,1)$, the Fedorov's exchange method returns an integral solution $\mZ = \sum_{i=1}^n \vz(i) \cdot \vv_i \vv_i^\top$ with $\sum_{i=1}^n \vz(i) \leq b$ and $\vz(i) \in \{0,1\}$ for $1 \leq i \leq n$ such that 
\[
\tr\left(\mZ^{-1}\right) \leq (1+\eps) \cdot \tr(\mX^{-1})
\quad \text{whenever} \quad
b \geq \Omega\Big( \frac{ d + \sqrt{\tr(\mX) \tr\left(\mX^{-1}\right)} }{\eps} \Big).
\]
In particular, let $\kappa = \frac{\lambda_{\max}(\mX^*)}{\lambda_{\min}(\mX^*)}$ be the condition number of an optimal solution $\mX^*$, then the Fedorov's exchange method gives a $(1+\eps)$-approximation algorithm for A-design whenever $b \geq \Omega\Big(\frac{\left(1+\sqrt{\kappa}\right) \cdot d}{\eps}\Big)$, and the time complexity is polynomial in $n, d, \frac{1}{\eps}, \kappa$.
\end{restatable}

For E-design, there are no known combinatorial local search algorithms,
and there are examples showing that Fedorov's exchange method does not work even if there exists a well-conditioned optimal solution (see Section~\ref{ss:examples}).
Using the regret minimization framework in~\cite{AZLO15,AZLSW20}, however, we prove that a modified local search algorithm using a ``smoothed'' objective function for E-design works as long as there exists an almost optimal solution with good condition number.

\begin{restatable}{theorem}{CombE} \label{t:E-combin}
Let $\mX := \sum_{i=1}^n \vx(i) \cdot \vv_i \vv_i^\top$ with $\sum_{i=1}^n \vx(i) = b$ and $\vx(i) \in [0,1]$ for $1 \leq i \leq n$ be a fractional solution to E-design. 
For any $\eps \in (0,1)$, there is a combinatorial local search algorithm which returns an integral solution $\mZ = \sum_{i=1}^n \vz(i) \cdot \vv_i \vv_i^\top$ with $\sum_{i=1}^n \vz(i) \leq b$ and $\vz(i) \in \{0,1\}$ for $1 \leq i \leq n$ such that 
\[
\lambda_{\min}(\mZ) \geq (1-\eps) \cdot \lambda_{\min} \left( \mX \right)
\quad \text{whenever} \quad
b \geq \Omega\bigg( \frac{d}{\eps^2} \sqrt{\frac{\lambda_{\rm avg}(\mX)}{\lambda_{\min}(\mX)}} \bigg),
\]
where $\lambda_{\text{avg}}(\mX) = \frac{\tr(\mX)}{d}$ is the average eigenvalue of $\mX$.

In particular, let $\kappa = \frac{\lambda_{\max}(\mX^*)}{\lambda_{\min}(\mX^*)}$ be the condition number of an optimal solution $\mX^*$, then the combinatorial local search method gives a polynomial time $(1-\eps)$-approximation algorithm for E-design whenever $b \geq \Omega\left(\frac{d\sqrt{\kappa}}{\eps^2}\right)$,
and the time complexity is polynomial in $n, d, \frac{1}{\eps}, \kappa$.
\end{restatable}

A combinatorial ``capping'' procedure was used in~\cite{MSTX19} to reduce the A-design problem to the case when every vector is ``short'', for which Fedorov's exchange method works.
This capping procedure, however, crucially leveraged that a vector can be chosen multiple times.
We do not have a preprocessing procedure to reduce A-design and E-design in the without repetition setting to the case when Theorem~\ref{t:A-combin} and Theorem~\ref{t:E-combin} apply.
We leave it as an open problem to design a fully combinatorial algorithm for A-design and E-design in the general case.

\subsubsection{Rounding Algorithms for Convex Programming Relaxations} \label{ss:rounding}

There are natural convex programming relaxations for the D/A/E-design problems.
The best known rounding algorithms for these three problems are all quite different, i.e.~approximate positively correlated distributions for D-design~\cite{SX18}, proportional volume sampling for A-design~\cite{NST19}, and regret minimization for E-design~\cite{AZLSW20,LZ20}.
Although the one-sided spectral rounding result in~\cite{AZLSW20,LZ20} (see Section~\ref{ss:iter}) provides a general solution for a large class of experimental design problems including D/A/E-design, this only works under the stronger assumption that $b \geq \Omega\left( \frac{d}{\eps^2} \right)$ and it was unclear how to unify the best known algorithmic results.

Surprisingly, we prove that the iterative randomized rounding algorithm for E-design in~\cite{LZ20} can be modified just slightly to match and improve the previous results for D/A-design, as well as to extend them to handle multiple knapsack constraints.
To this end, we bypass the one-sided spectral rounding problem.
Instead, we perform a refined analysis for the iterative randomized rounding algorithm, in which the minimum eigenvalue of the current solution plays an unexpectedly crucial role for D/A-design as well.
This provides a unified rounding algorithm to achieve the optimal results for the natural convex programming relaxations for these experimental design problems.

In D/A/E-design with knapsack constraints, we are given vectors $\vv_1, \ldots, \vv_n \in \R^d$, knapsack constraints $\vc_1, \ldots, \vc_m \in \R^n_+$ and budgets $b_1, \ldots, b_m \geq 0$, and the goal is to find a solution $\vz \in \{0,1\}^n$ with $\inner{\vc_i}{\vz} \leq b_i$ for $1 \leq i \leq m$ to optimize the objective value.
Consider the following natural convex programming relaxations for D/A-design.
\begin{equation} \label{eq:convex}
    \begin{aligned}
        & \underset{\vx \in \R^n}{\rm maximize} & & \log \det\bigg(\sum_{i=1}^n \vx(i) \cdot \vv_i \vv_i^\top\bigg) {\rm~~~~~or~~~~~} \underset{\vx \in \R^n}{\rm minimize}~~ \tr\bigg( \bigg( \sum_{i=1}^n \vx(i) \cdot \vv_i \vv_i^\top \bigg)^{-1} \bigg) \\
        & \text{\rm subject to} & & \inner{\vc_j}{\vx} \leq b_j, \quad ~~~\text{for } 1 \leq j \leq m, \\
        & & & 0 \leq \vx(i) \leq 1, \quad ~\text{ for } 1 \leq i \leq n.
    \end{aligned} %\tag{WCP-A} 
\end{equation}

\begin{restatable}{theorem}{DARounding} \label{t:DA-rounding}

Let $\vx \in [0,1]^n$ be an optimal fractional solution to D/A-design with knapsack constraints. For any $\eps \leq \frac{1}{200}$, if each knapsack constraint budget satisfies 
$b_j \geq \frac{2d \norm{\vc_j}_{\infty}}{\eps}$,
then there is a randomized exchange algorithm which returns in polynomial time an integral solution $\mZ = \sum_{i=1}^n \vz(i) \cdot \vv_i \vv_i^\top$ with 
$\vz(i) \in \{0,1\}$ for $1 \leq i \leq n$ such that
\[
\det\bigg(\sum_{i=1}^n \vz(i) \cdot \vv_i \vv_i^\top\bigg)^{\frac{1}{d}} \geq \big(1-O(\eps)\big) \cdot \det\bigg( \sum_{i=1}^n \vx(i) \cdot \vv_i \vv_i^\top\bigg)^{\frac{1}{d}}
\text{\rm~for~D-design,}
\]
\[
\tr\bigg(\bigg( \sum_{i=1}^n \vz(i) \cdot \vv_i \vv_i^\top\bigg)^{-1}\bigg) \leq (1+\eps) \cdot \tr\bigg(\bigg( \sum_{i=1}^n \vx(i) \cdot \vv_i \vv_i^\top\bigg)^{-1}\bigg)
\text{\rm~for~A-design}
\]
with probability at least $1 - O\left(\frac{k^2}{\eps^2} \cdot e^{-\Omega(\sqrt{d})}\right)$ where $k = O(d^2 + m)$.
Furthermore, each knapsack constraint $\inner{\vc_j}{\vz} \leq b_j$ is satisfied with probability at least $1 - e^{-\Omega(\eps d)}$.
\end{restatable}

Note that D/A-design with a cardinality constraint is the special case when there is only one cost constraint ($m=1$) and $\vc = \vec{1}$. % and $b = k$.
In this special case, Theorem~\ref{t:DA-rounding} improves the previous results in~\cite{SX18,NST19} by removing the term $O\left( \frac{1}{\eps^2} \log\left(\frac{1}{\eps}\right) \right)$ from their assumption $b \geq \Omega\left( \frac{d}{\eps} + \frac{1}{\eps^2} \log\left(\frac{1}{\eps}\right) \right)$, and this achieves the optimal integrality gap result for D-design~\cite{SX18} and A-design~\cite{NST19}.
In the general case with knapsack constraints, Theorem~\ref{t:DA-rounding} improves the previous result in~\cite{LZ20}, which requires a stronger assumption that $b_j \geq \Omega\big(\frac{d\norm{\vc_j}_{\infty}}{\eps^2}\big)$ to obtain the same approximation guarantee.
The knapsack constraints can be used for weighted experimental design and for incorporating fairness constraints in experimental design,
which we will discuss in the next subsection.

\subsubsection{Some Applications}

We discuss some applications of our results in specific instances of experimental design problems.

\vspace{2mm}

{\bf Fair and Diverse Data Summarization:}
In the data summarization problem, we are given $n$ data points $\vv_1, \ldots, \vv_n \in \R^d$, and the objective is to choose a subset of $b$ data points that provides a ``fair'' and ``diverse'' summary of the data.
For diversity, the D-design objective of maximizing determinant is a popular measure used in previous work~\cite{Mir17,CKS+18}.
For fairness, the partition constraints~\cite{NS16,CKS+18} for D-design are used to partition the set $X$ of data points into $p$ disjoint groups $X_1 \cup \cdots \cup X_p$ and to ensure that $b_i$ data points are chosen in $X_i$ where $\sum_{i=1}^p b_i = b$.

We believe that Theorem~\ref{t:DA-rounding} for D-design with knapsack constraints provides an alternative solution for this problem.
The main advantage is that the knapsack constraints are more flexible in that they do not require the groups to be disjoint.
For instance, we can have knapsack constraints on arbitrary subsets $X_1, \ldots, X_p \subseteq X$ of the form $\sum_{j \in X_i} \vx(j) \leq b_i$ to ensure that at most $b_i$ data points are chosen in group $X_i$, so that we can handle constraints of overlapping groups such as race, age, gender (e.g.~at most 50\% of the chosen vectors correspond to men/women), etc.
Also, the approximation guarantee in Theorem~\ref{t:DA-rounding} is stronger than the constant factor approximation for D-design with partition constraint~\cite{NS16}, and the convex programming relaxation used in Theorem~\ref{t:DA-rounding} is simpler and easier to be solved than the more sophisticated one used in~\cite{NS16}.

\vspace{2mm}

{\bf Minimizing Total Effective Resistance:}
Ghosh, Boyd and Saberi~\cite{GBS08} studied the problem of choosing a subgraph with at most $b$ edges to minimize the total effective resistance,
and showed that this is a special case of A-design.
The proportional volume sampling algorithm by Nikolov, Singh and Tantipongpipat~\cite{NST19} achieves a $(1+\eps)$-approximation for this problem when $b \geq \Omega(\frac{n}{\eps} + \frac{1}{\eps^2} \log \frac{1}{\eps})$ where $n$ is the number of vertices in the graph.
Lau and Zhou~\cite{LZ20} considered the weighted problem of choosing a subgraph with total edge cost at most $b$ to minimize the total effective resistance,
and gave a $(1+\eps)$-approximation algorithm when $b \geq \Omega\left(\frac{n\norm{\vc}_{\infty}}{\eps^2}\right)$ where $\vc$ is the cost vector of the edges. 
Theorem~\ref{t:DA-rounding} improves these two results.

\begin{restatable}{corollary}{RoundReff} \label{c:Reff-rounding}
For any $0 < \eps < 1$, there is a polynomial time randomized $(1+\eps)$-approximation algorithm for minimizing total effective resistance in an edge weighted graph whenever $b \geq \Omega\left(\frac{n\norm{\vc}_{\infty}}{\eps}\right)$. 
\end{restatable}

{\bf Maximizing Algebraic Connectivity:}
Ghosh and Boyd~\cite{GB06} studied the problem of choosing a subgraph with total cost at most $b$ that maximizes the algebraic connectivity, {\em i.e.}~the second smallest eigenvalue of its Laplacian matrix.
Kolla, Makarychev, Saberi and Teng~\cite{KMS+10} provided the first algorithm with non-trivial approximation guarantee in the zero-one cost setting.
Lau and Zhou~\cite{LZ20} observed that this is a special case of E-design
and gave a $(1-\eps)$-approximation algorithm when $b \geq \Omega\left(\frac{n\norm{\vc}_{\infty}}{\eps^2}\right)$ where $\vc$ is the cost vector of the edges.

All previous results are based on convex programming. 
Theorem~\ref{t:E-combin} provides a combinatorial algorithm for the unweighted problem, where the goal is to choose $b$ edges to maximize the algebraic connectivity, and shows that it has a good performance as long as the optimal value is large.

\begin{restatable}{corollary}{CombLambda} \label{c:lambda2-combin}
For any $0 < \eps < 1$, there is a polynomial time combinatorial $(1-\eps)$-approximation algorithm for maximizing algebraic connectivity in an unweighted graph whenever $b \geq \Omega\left(\frac{n}{\eps^4 \lambda_2^*} \right)$, where $\lambda_2^*$ is the optimal value for the problem.
\end{restatable}

\subsection{Techniques}

We extend the randomized approach in~\cite{LZ20} to analyze both combinatorial local search algorithms and to design improved approximation algorithms for D/A-design with knapsack constraints.
The approach in~\cite{LZ20} is based on the regret minimization framework developed in~\cite{AZLSW20} for the one-sided spectral rounding problem.
In the following, we will first present the techniques used in these two previous results, and then present the new ideas in this paper.

{\bf Previous Techniques:}
In~\cite{AZLSW20}, Allen-Zhu, Li, Singh, and Wang first solved the natural convex programming relaxation for experimental design and obtained a solution $\vx \in \R^n$, and performed a linear transformation so that $\sum_{i=1}^n \vx(i) \cdot \vv_i \vv_i^\top = \mI$.
They showed that the experimental design problem is reduced to the following one-sided spectral rounding problem, where the goal to find a subset $S \subseteq [n]$ so that $\sum_{i \in S} \vv_i \vv_i^\top \succeq (1-\eps)\mI$ and $|S| \leq \sum_{i=1}^n \vx(i) = b$.
To solve this problem, they started from an arbitrary initial set $S_0$ of $b$ vectors, and in each step $t \geq 1$ they sought to find a pair $i_t \in S_{t-1}$ and $j_t \notin S_{t-1}$ so that the new solution $S_{t} \gets S_{t-1} - \vv_{i_t} + \vv_{j_t}$ improves the current solution $S_{t-1}$ in terms of a potential function related to the minimum eigenvalue.  
Using the regret minimization framework that maintains a density matrix $\mA_t$ at each step $t$ (see Section~\ref{ss:regret}), they proved that choosing
\[
i_t := \argmin_{i \in S_{t-1}} \frac{\inner{\vv_i \vv_i^\top}{\mA_t}}{1-2\alpha \inner{\vv_i \vv_i^\top}{\mA_t^{\frac12}}}
\quad {\rm and} \quad
j_t := \argmax_{j \notin S_{t-1}} \frac{\inner{\vv_j \vv_j^\top}{\mA_t}}{1+2\alpha \inner{\vv_j \vv_j^\top}{\mA_t^{\frac12}}}\
\]
would improve the current solution $S_{t-1}$ as long as $\sum_{i \in S_t} \vv_i \vv_i^\top \not\succeq (1-\eps)\mI$.

In~\cite{LZ20}, the idea is to use the following probability distributions to sample $i_t$ and $j_t$:
\begin{equation} \label{e:E-distribution}
\Pr(i_t = i) \propto \big(1-\vx(i)\big) \cdot \left(1-2\alpha \inner{\vv_i \vv_i^\top}{\mA_t^{\frac12}}\right)
\quad {\rm and} \quad
\Pr(i_j = j) \propto \vx(j) \cdot \left(1+2\alpha \inner{\vv_j \vv_j^\top}{\mA_t^{\frac12}}\right).
\end{equation}
The advantage of doing random sampling is that the knapsack constraints will be approximately preserved, while the potential function related to the minimum eigenvalue is expected to improve.
Freedman's martingale inequality and a new concentration inequality for non-martingales are used to prove that all these quantities are close to their expected values with high probability.

{\bf Analysis of Combinatorial Algorithms:}
In this paper, we use this randomized approach in~\cite{LZ20} to analyze both combinatorial algorithms and rounding algorithms.
For combinatorial local search algorithms, one difference from the previous analysis in~\cite{MSTX19} is that we compare the objective of the current integral solution to that of an optimal {\em fractional} solution.
When the objective value of the fractional solution is considerably better than that of the current integral solution, we use the fractional solution to define appropriate probability distributions similar to that in (\ref{e:E-distribution}) to sample $i_t$ and $j_t$ so that the expected objective value of $S_{t} \gets S_{t-1} - \vv_{i_t} + \vv_{j_t}$ improves, and this would imply the existence of an improving pair in Fedorov's exchange method.
One advantage of this approach is that this allows us the flexibility to compare the current integral solution to a fractional solution with smaller budget which still has its objective value close to the optimal one.

Our analysis is arguably simpler than that in~\cite{MSTX19} which uses a dual fitting method while we only do a primal analysis.
More importantly, our analysis shows that if the optimal fractional solution is well-conditioned (e.g.~$\sum_{i=1}^n \vx(i) \cdot \vv_i \vv_i = \mI_d$), then the Fedorov's exchange method indeed performs as well as the best known rounding algorithms.
This gives us a new insight that the only important step in rounding algorithms for the unweighted experimental design problems is the ability to first transform the optimal fractional solution to the identity matrix.
For E-design, simply doing Fedorov's exchange method on the objective function $\lambda_{\min}\left(\sum_{i \in S_t} \vv_i \vv_i^\top\right)$ would not work (see Section~\ref{ss:examples}), and instead we apply the Fedorov's exchange method to the potential function in the regret minimization framework, which is morally the same as the potential function $\tr\left( \left( \sum_{i \in S_t} \vv_i \vv_i^\top - l \mI_d \right)^{-1} \right)$ used by Batson, Spielman and Srivastava for spectral sparsification~\cite{BSS12}.

{\bf Analysis of Rounding Algorithms:}
For the rounding algorithm for experimental design with knapsack constraints, surprisingly we prove that a minor modification of the algorithm for E-design in~\cite{LZ20} would work for D/A-design with improved approximation guarantees!
Essentially, we just use the algorithm for E-design but only require that the solution to have minimum eigenvalue $\frac34$ rather than $1-\eps$.
Our analysis has two phases.
In the first phase, using the results in~\cite{LZ20}, we show that the randomized exchange algorithm will find a solution with minimum eigenvalue at least $\frac34$ in polynomial time with high probability whenever $b \geq \Omega\left(\frac{d}{\eps}\right)$ (rather than $b \geq \Omega\left(\frac{d}{\eps^2}\right)$ in order to achieve minimum eigenvalue at least $1-\eps$).
In the second phase, we prove that the minimum eigenvalue will maintain to be at least $\frac14$ with high probability when $\eps$ is not too tiny, and then the objective value for D-design and A-design will improve to $(1\pm\eps)$ times the optimal objective value in polynomial time with high probability.
The condition that the minimum eigenvalue is at least $\frac14$ is used crucially in multiple places for the analysis of the second phase.  
Interestingly, it is used in showing that the probability distributions in~(\ref{e:E-distribution}) for E-design are also good for improving the objective value for D-design and A-design.
Moreover, it is used in the martingale concentration argument, e.g.~to show that the martingale is bounded and to prove upper bounds on the variance of the changes.
For the martingale concentration argument, we also use the optimality conditions for convex programs to prove that the vectors with fractional value are ``short'' in order to bound the quantities involved.
Overall, the analysis for the rounding algorithm is quite involved, but it provides a unifying algorithm to achieve the optimal results for the natural convex relaxations for D/A/E-design.
Please refer to Section~\ref{s:knapsack} for a more detailed outline of the analysis.

\subsection{Previous Work} \label{ss:previous}

All three experimental design problems are NP-hard~\cite{CM09,Wel82} and also APX-hard~\cite{SEFM15,NST19,CM09}.
Despite the long history and the wide interest, 
strong approximation algorithms for these problems are only obtained recently.

{\bf D-design:}
Singh and Xie~\cite{SX18} designed an $(1-\eps)$-approximation algorithm for D-design in the with repetition setting when $b \geq \frac{2d}{\eps}$,
and in the without repetition setting when 
$b = \Omega\left(\frac{d}{\eps} + \frac{1}{\eps^2} \log \frac{1}{\eps}\right)$. 
Their algorithm is by rounding an optimal solution to a natural convex program relaxation using approximate positively correlated distributions.

Madan, Singh, Tantipongpipat and Xie~\cite{MSTX19} analyzed the Fedorov's exchange method and proved that it gives an $(1-\eps)$-approximation algorithm for D-design as long as $b \geq d + \frac{d}{\eps}$, which improves upon the above result.
However, they only provide a polynomial time implementation of the local search algorithm to achieve this guarantee in the less general with repetition setting.

{\bf A-design:}
Nikolov, Singh and Tantipongpipat~\cite{NST19} designed an $(1+\eps)$-approximation algorithm for A-design in the with repetition setting when $b \geq d + \frac{d}{\eps}$, and in the without repetition setting when $b = \Omega\left(\frac{d}{\eps} + \frac{1}{\eps^2} \log \frac{1}{\eps}\right)$.
Their algorithm is by rounding an optimal solution to a natural convex program relaxation using proportional volume sampling. Their algorithm also works for D-design with the same guarantee.

Madan, Singh, Tantipongpipat and Xie~\cite{MSTX19} also analyzed the Fedorov's exchange method for A-design, and showed that there are arbitrarily bad local optimal solutions.
On the other hand, they proved that Fedorov's exchange method works when all the input vectors are ``short'', and they designed a ``capping procedure'' to reduce the general case to the case when all vectors are short.
As a result, they obtained a combinatorial $(1+\eps)$-approximation algorithm, without solving convex programs, for A-design when $b \geq \Omega\left(\frac{d}{\eps^4}\right)$ in the with repetition setting.

{\bf E-design:}
Allen-Zhu, Li, Singh and Wang~\cite{AZLSW17c,AZLSW20} designed an $(1-\eps)$-approximation algorithm for E-design in the with and without repetition settings when $b \geq \Omega\left( \frac{d}{\eps^2} \right)$.
Their algorithm is by rounding an optimal solution to a natural convex program relaxation using the regret minimization framework, which was initially developed for the spectral sparsification problem~\cite{AZLO15}.
They formulated and solved a ``one-sided spectral rounding problem'' (see Section~\ref{ss:iter}), and showed that experimental design with any objective function satisfying some mild regularity assumptions, including D/A/E-design, can be reduced to the one-sided spectral rounding problem.
Their algorithm for one-sided spectral rounding can be viewed as a local search algorithm, and this was the starting point of the current work.

Nikolov, Singh and Tantipongpipat~\cite{NST19} showed that the assumption $b \geq \Omega\left( \frac{d}{\eps^2} \right)$ is necessary to achieve $(1-\eps)$-approximation for E-design using the natural convex program,
and Lau and Zhou~\cite{LZ20} showed that the assumption $b \geq \Omega\left( \frac{d}{\eps^2} \right)$ is necessary for the one-sided spectral rounding problem.
These suggest that the regret minimization framework may not be used to match the results for D/A-design, but we bypass the one-sided spectral rounding problem to prove Theorem~\ref{t:DA-rounding}.

{\bf Experimental design with additional constraints:}
Lau and Zhou~\cite{LZ20} considered the generalization of the experimental design problem with additional knapsack constraints as in~\eqref{eq:convex}.
In particular, it generalizes the experimental design problems to the {\em weighted} problems, where each vector $\vv_i$ has a weight $\vc(i)$ and the goal is to choose a subset $S$ of vectors with $\sum_{i \in S} \vc(i) \leq b$ to optimize the objective value.
Using a randomized iterative rounding algorithm, they obtained a $(1\pm\eps)$-approximation algorithm for weighted D/A/E-design when $b \geq \Omega\left(\frac{d\norm{\vc}_{\infty}}{\eps^2} \right)$.

Using more sophisticated convex programming relaxations, Nikolov and Singh~\cite{NS16} designed an approximation algorithm for D-design under partition constraints.
Recently, Madan, Nikolov, Singh and Tantipongpipat~\cite{MNST20} designed an approximation algorithm for D-design under general matroid constraints.

\section{Preliminaries} \label{s:prelim}

We recall some basic linear algebra in Section~\ref{ss:algebra}.
Then, we review the regret minimization framework in Section~\ref{ss:regret}, and the iterative randomized rounding algorithm for one-sided spectral rounding in Section~\ref{ss:iter}. 
Finally, we state some useful inequalities for the analysis of martingales and the analysis of the objective functions in Section~\ref{ss:martingale} and Section~\ref{ss:comb-ineq} respectively.

\subsection{Linear Algebra} \label{ss:algebra}

% numbers
We write $\R$ and $\R_+$ as the sets of real numbers and non-negative real numbers.
Throughout the paper, we use italic sans-serif font for vectors and matrices, e.g.~$\vx$, $\mA$.

% vector
All the vectors in this paper only have real entries.
Let $\R^d$ denote the $d$-dimensional Euclidean space. 
We write $\vec{1}_d$ as the $d$-dimensional all-one vector.
Given a vector $\vx \in \R^d$, 
we write $\vx(i)$ as the $i$-th entry of vector $\vx$, 
and write $\vx(S) := \sum_{i \in S} \vx(i)$ for any subset $S \subseteq [d]$. 
We denote $\norm{\vx}_2$ the $\ell_2$-norm, $\norm{\vx}_1$ the $\ell_1$-norm, and $\norm{\vx}_\infty$ the $\ell_\infty$-norm of $\vx$. 
A vector $\vv \in \R^d$ is a column vector, and its transpose is denoted by $\vv^\top$.
Given two vectors $\vx,\vy \in \R^d$, the inner product is defined as $\inner{\vx}{\vy} := \sum_{i=1}^n \vx(i) \cdot \vy(i)$.
The Cauchy-Schwarz inequality says that $\inner{\vx}{\vy} \leq \norm{\vx} \norm{\vy}$.

% matrix norms and inner product
All matrices considered in this paper are real symmetric matrices.
We denote the $d \times d$ identity matrix by $\mI_d$ or simply $\mI$ when the dimension is clear from the context.
It is a fundamental result that any $d \times d$ real symmetric matrix has $d$ real eigenvalues $\lambda_1 \leq \ldots \leq \lambda_n$ and an orthonormal basis of eigenvectors.
We write $\lambda_{\max}(\mM)$ and $\lambda_{\min}(\mM)$ as the maximum and the minimum eigenvalue of a real symmetric matrix $\mM$.
The trace of a matrix $\mM$, denoted by $\tr(\mM)$, is defined as the sum of the diagonal entries of $\mM$.
It is well-known that $\tr(\mM) = \sum_{i=1}^d \lambda_i(\mM)$ where $\lambda_i(\mM)$ denotes the $i$-th eigenvalue of $\mM$.

A matrix $\mM$ is a positive semidefinite (PSD) matrix, denoted as $\mM \succeq 0$, if $\mM$ is symmetric and all the eigenvalues are nonnegative, or equivalently, the quadratic form $\vx^\top \mM \vx \geq 0$ for any vector $\vx$. 
We use $\mA \succeq \mB$ to denote $\mA - \mB \succeq 0$ for matrices $\mA$ and $\mB$. 
We write $\mathbb{S}^d_+$ as the set of all $d$-dimensional PSD matrices.
Let $\mM \succeq 0$ be a PSD matrix with eigendecomposition $\mM = \sum_i \lambda_i \vv_i \vv_i^\top$, where $\lambda_i \geq 0$ is the $i$-th eigenvalue and $\vv_i$ is the corresponding eigenvector. 
The square root of $\mM$ is $\mM^{1/2}:= \sum_i \sqrt{\lambda_i} \vv_i \vv_i^\top$.
Given two matrices $\mA$ and $\mB$ of the same size, the Frobenius inner product of $\mA, \mB$ is denoted as $\langle \mA, \mB \rangle := \sum_{i,j} \mA(i,j) \cdot \mB(i,j) = \tr(\mA^\top \mB)$. The following are two standard facts
\[
\mA, \mB \succeq 0 \; \Longrightarrow \; \langle \mA, \mB \rangle \geq 0 \quad \text{and} \quad \mA \succeq 0, \mB \succeq \mC \succeq 0 \; \Longrightarrow \; \langle \mA, \mB \rangle \geq \langle \mA, \mC \rangle.
\]

\subsection{Regret Minimization} \label{ss:regret}

We will use the regret minimization framework developed in~\cite{AZLO15,AZLSW20} for spectral sparsification and one-sided spectral rounding.
The framework is for an online optimization setting.
In each iteration $t$, the player chooses an action matrix $\mA_t$ from the set of density matrices $\Delta^d := \{\mA \in \R^{d \times d} \mid \mA \succeq 0, \tr(\mA)=1\}$, which can be understood as a probability distribution over the set of unit vectors.
The player then observes a feedback matrix $\mF_t$ and incurs a loss of $\inner{\mA_t}{\mF_t}$.
After $\tau$ iterations, the regret of the player is defined as 
\[
R_\tau := \sum_{t=1}^\tau \inner{\mA_t}{\mF_t} - \inf_{\mB \in \Delta^d} \sum_{t=1}^\tau \inner{\mB}{\mF_t} 
= \sum_{t=1}^\tau \inner{\mA_t}{\mF_t} - \lambda_{\min}\Bigg(\sum_{t=1}^\tau \mF_t\Bigg),
\]
which is the difference between the loss of the player actions and the loss of the best fixed action $\mB$, that can be assumed to be a rank one matrix $\vv\vv^\top$.
The objective of the player is to minimize the regret.
A well-known algorithm for regret minimization is follow-the-regularized-leader which plays the action
\[\mA_t = \argmin_{\mA \in \Delta^d} \bigg\{ w(\mA) + \alpha \cdot \sum_{l=0}^{t-1} \inner{\mA}{\mF_l} \bigg\},
\]
where $w(\mA)$ is a regularization term and $\alpha$ is a parameter called the learning rate that balances the loss and the regularization. 
Here $\mF_0$ is an initial feedback which is given before the game started.
Different choices of regularization give different algorithms for regret minimization.
The choice that we will use is the $\ell_{\frac12}$-regularizer $w(\mA) = -2\tr(\mA^{\frac12})$ introduced in~\cite{AZLO15}, which plays the action
\begin{equation} \label{e:closed-form}
\mA_t = \bigg( l_t \mI + \alpha \sum_{l=0}^{t-1} \mF_l \bigg)^{-2},
\end{equation}
where $l_t$ is the unique constant that ensures $\mA_t \in \Delta^d$.
Allen-Zhu, Li, Singh and Wang~\cite{AZLSW20} proves the following regret bound for rank-two feedback matrices with $\ell_{\frac12}$ regularizer.

\begin{theorem}[Lemma 2.5 in~\cite{AZLSW20}] \label{t:regret-rank-two}
Suppose the action matrix $\mA_t \in \R^{d \times d}$ is of the form of~\eqref{e:closed-form} for some $\alpha > 0$. 
Suppose the initial feedback matrix $\mF_0 \in {\mathbb S}^d$ is a symmetric matrix, and for all $t \geq 1$ each feedback matrix $\mF_t$ is of the form $\vv_{j_t} \vv_{j_t}^\top - \vv_{i_t} \vv_{i_t}^\top$ for some $\vv_{j_t}, \vv_{i_t} \in \R^d$ such that $\alpha \inner{\vv_{i_t} \vv_{i_t}^\top}{\mA_t^{\frac12}} < \frac12$. 
then for any density matrix $U$,
\[
\sum_{t=1}^\tau \inner{\mF_t}{\mU} \geq 
\sum_{t=1}^\tau 
\Bigg( \frac{ \inner{\vv_{j_t} \vv_{j_t}^\top}{\mA_t} }{1 + 2\alpha \inner{\vv_{j_t} \vv_{j_t}^\top}{\mA_t^{\frac12}} }   
- \frac{\inner{\vv_{i_t} \vv_{i_t}^\top }{\mA_t}}{1 - 2\alpha \inner{\vv_{i_t} \vv_{i_t}^\top}{\mA_t^{\frac12}} } 
\Bigg)
- \frac{D(\mA_1,\mU)}{\alpha},
\]
where $D(\mA_1, \mU):= \tr(\mA_1^{\frac12}) - 2\tr(\mU^{\frac12}) + \inner{\mA_1^{-\frac12}}{\mU}$ is the Bregman divergence of the $\ell_{\frac12}$-regularizer.
\end{theorem}

The above theorem is used in~\cite{AZLSW20} to give a lower bound on the minimum eigenvalue of $\lambda_{\min}(\sum_{t=0}^{\tau} \mF_t)$, by bounding $D(\mA_1,\mU) \leq \alpha \inner{\mF_0}{\mU} + 2\sqrt{d}$. 
We will use the following more refined version where there is an extra $\lambda_{\min}(\mF_0)$ term.

\begin{corollary} \label{cor:lambda-min-rank-two}
Under the same assumptions as in Theorem~\ref{t:regret-rank-two},
\begin{equation*} %\label{e:regret-rank-two}
\lambda_{\min}\bigg( \sum_{t=0}^\tau \mF_t \bigg) 
\geq
\sum_{t=1}^\tau
\Bigg( \frac{ \inner{\vv_{j_t} \vv_{j_t}^\top}{\mA_t} }{1 + 2\alpha \inner{\vv_{j_t} \vv_{j_t}^\top}{\mA_t^{\frac12}} }   
- \frac{\inner{\vv_{i_t} \vv_{i_t}^\top }{\mA_t}}{1 - 2\alpha \inner{\vv_{i_t} \vv_{i_t}^\top}{\mA_t^{\frac12}} } 
\Bigg) 
- \frac{2\sqrt{d}}{\alpha} + \lambda_{\min}(\mF_0).
\end{equation*}
\end{corollary}
\begin{proof}
Let $\mU$ be a rank-one projection on the minimum eigenspace of $\sum_{t=0}^\tau \mF_t$.
We will show that $D(\mA_1, \mU) \leq \inner{\alpha \mF_0}{\mU} + 2\sqrt{d} - \alpha \lambda_{\min}(\mF_0)$. 
Note that
\[
D(\mA_1, \mU) = \tr(\mA_1^{\frac12}) - 2\tr(\mU^{\frac12}) + \inner{\mA_1^{-\frac12}}{\mU} \leq \inner{\alpha \mF_0 + l_1 \mI}{\mU} + \sqrt{d} = \inner{\alpha \mF_0}{\mU} + l_1 + \sqrt{d},
\]
where we used $\tr(\mA_1^{\frac12}) \leq \sqrt{d}$ as stated in Lemma~\ref{l:tr} and $\mU$ is also a density matrix.
So it remains to upper bound $l_1$. 
Since $\mF_0 \succcurlyeq \lambda_{\min}(\mF_0) \cdot \mI$ and $\tr(\mA_1) = 1$, 
it follows that
\[
1 = \tr(\mA_1) \leq (\alpha \lambda_{\min}(\mF_0) + l_1)^{-2} \cdot \tr(\mI) = (\alpha \lambda_{\min}(\mF_0) + l_1)^{-2} \cdot d \quad \Longrightarrow \quad l_1 \leq \sqrt{d} - \alpha \lambda_{\min}(\mF_0).
\]
The corollary follows from Theorem~\ref{t:regret-rank-two} by plugging in this $\mU$ and this upper bound on $D(\mA_1,\mU)$.
\end{proof}

One technical point used in~\cite{AZLSW17c,AZLSW20} is that the partial solution $\mZ_t := \sum_{l=0}^{t-1} \mF_l$ at time $t$ and the action matrix $\mA_t$ at time $t$ have the same eigenbasis due to (\ref{e:closed-form}).
This allows one to bound $\inner{\mZ_t}{\mA_t}$ and $\inner{\mZ_t}{\mA_t^{\frac12}}$ as follows.

\begin{lemma}[Claim 2.11 in~\cite{AZLSW20}] \label{l:cospectral}
Let $\mZ \succeq 0$ be an $d \times d$ positive semidefinite matrix
and $\mA = (\alpha \mZ + l\mI)^{-2}$ for some $\alpha>0$ where $l$ is the unique constant such that $\mA$ is a density matrix.
Then
\[
\inner{\mZ}{\mA} \leq \frac{\sqrt{d}}{\alpha} + \lambda_{\min}(\mZ) \qquad \text{and} \qquad \alpha \inner{\mZ}{\mA^{\frac12}} \leq d + \alpha \sqrt{d} \cdot \lambda_{\min}(\mZ).
\]
\end{lemma}

We will use the above lemma in the analysis of the combinatorial algorithm for E-design. 
We also use the following simple fact about the action matrix.

\begin{lemma} \label{l:tr}
For any $d \times d$ matrix $\mA \succcurlyeq 0$ satisfying $\tr(\mA) = 1$, 
$\tr\big(\mA^{\frac12}\big) \leq \sqrt{d}$.
\end{lemma}

\subsection{One-Sided Spectral Rounding and Iterative Randomized Rounding} \label{ss:iter}

The following one-sided spectral rounding result was formulated and proved by Allen-Zhu, Li, Singh and Wang~\cite{AZLSW20}.
It is the main theorem that implies a $(1\pm \eps)$-approximation algorithm for a large class of experimental design problems with one cardinality constraint, whenever the budget $b \geq \Omega\left(\frac{d}{\eps^2}\right)$.

\begin{theorem}[\cite{AZLSW20}] \label{t:swap}
Let $\vv_1, \vv_2, \ldots, \vv_n \in \R^d$, $\vx \in [0,1]^m$ and $b = \sum_{i=1}^m \vx(i)$.
Suppose $\sum_{i=1}^m \vx(i) \cdot \vv_i \vv_i^\top = \mI_d$ and $b \geq \frac{5d}{\eps^2}$ for some $\eps \in (0,\frac13]$. 
Then there is a polynomial time algorithm to return a subset $S \subseteq [m]$ with
\[|S| \leq b
\quad {\rm and} \quad 
\sum_{i \in S} v_i v_i^\top \succeq (1-3\eps) \cdot \mI_d.\]
\end{theorem}

Theorem~\ref{t:swap} was extended in~\cite{LZ20} to incorporate non-negative linear constraints, where the goal is to output a subset $S$ that approximately satisfies the spectral lower bound and also $\vc(S) \approx \inner{\vc}{\vx}$ for any non-negative linear constraint $\vc \in \R^n_+$.
The randomized exchange algorithm in this paper is almost the same as the iterative randomized rounding algorithm in~\cite{LZ20}. 
In the following, we describe the algorithm in~\cite{LZ20}
and state several results that we will use in our analysis.

In the iterative randomized rounding algorithm, 
we start with an initial solution set $S_0 \subseteq [n]$, 
which is constructed by sampling each vector $i$ with probability $\vx(i)$ independently for all $i \in [n]$. 
In each iteration $t$, with the current solution set $S_{t-1}$, we randomly choose a vector $i_t \in S_{t-1}$ and a vector $j_t \in [n] \setminus S_{t-1}$ and set $S_t \gets S_{t-1} - i_t + j_t$. 
The sampling distributions in each iteration depend on the action matrix $\mA_t$ as defined in eq.~\eqref{e:closed-form}. 
The probability of choosing a vector $i_t \in S_{t-1}$ is $\frac{1}{k} \cdot \big(1-\vx(i)\big) \cdot \big(1 - 2\alpha \inner{\vv_i \vv_i^\top}{\mA_t^{\frac12}})\big)$,
and the probability of choosing a vector $j_t \in [n] \setminus S_{t-1}$ is $\frac{1}{k} \cdot \vx(j) \cdot \big(1 + 2\alpha \inner{\vv_j \vv_j^\top}{\mA_t^{\frac12}}\big)$,
where $k$ is a large enough denominator so that the two distributions are well-defined.
Informally, the terms $1-\vx(i)$ and $\vx(j)$ in the sampling probability help us maintain the cost $\vc(S_t) \approx \inner{\vc}{\vx}$, 
and the term depending on $\mA_t$ help us improve the spectral lower bound.
The following are the precise statements that we will use in this paper.

\begin{theorem}[Theorem 3.8 of~\cite{LZ20}] \label{t:iterative_rounding_terminate}
Let $\alpha = \frac{\sqrt{d}}{\gamma}$ be the learning rate used in computing the action matrix as defined in~\eqref{e:closed-form}.
Let $\tau$ be the first time such that the solution $S_{\tau}$ of the iterative randomized rounding algorithm satisfies $\sum_{i \in S_{\tau}} \vv_i \vv_i^\top \succeq (1-2\gamma) \cdot \mI$.
The probability that $\tau \leq \frac{2k}{\gamma}$ is at least $1 - e^{-\Omega(\sqrt{d})}$.
\end{theorem}

\begin{theorem}[Theorem 3.12 of~\cite{LZ20}] \label{t:cost}
Let $\alpha = \frac{\sqrt{d}}{\gamma}$ be the learning rate used in computing the action matrix as defined in~\eqref{e:closed-form}.
Suppose that the solution $S_t$ of the iterative randomized rounding algorithm satisfies $\lambda_{\min}\left(\sum_{i \in S_{t}} \vv_i \vv_i^\top\right) < 1$ for all $1 \leq t < \tau$. 
Then, for any given $\vc \in \R_+^n$ and any $\delta \in [0,1]$,
\[
\Pr\left[ \vc(S_\tau) \leq (1+\delta) \cdot \inner{\vc}{\vx} + \frac{15d \norm{\vc}_{\infty}}{\gamma} \right] \geq 1 - e^{-\Omega\big(\frac{\delta d}{\gamma}\big)}.
\]
\end{theorem}

These two results together imply a one-sided spectral rounding result with non-negative linear constraints, which can be used to give a $(1+\eps)$-approximation algorithm for experimental design problems with knapsack constraints, as long as the budget $b \geq \Omega\left(\frac{d\norm{\vc}_{\infty}}{\eps^2}\right)$ (see Theorem 4.11 in~\cite{LZ20}).

We will also use the following two lemmas in~\cite{LZ20} that were used in proving Theorem~\ref{t:iterative_rounding_terminate}.
Define
\begin{align} \label{eq:Delta}
\Delta_t^+ :=  \frac{\inner{\vv_{j_t} \vv_{j_t}^\top}{\mA_t} }{1 + 2\alpha \inner{\vv_{j_t} \vv_{j_t}^\top}{\mA_t^{\frac12}} }, \quad
\Delta_t^- := \frac{\inner{\vv_{i_t} \vv_{i_t}^\top }{\mA_t}}{1 - 2\alpha \inner{\vv_{i_t} \vv_{i_t}^\top}{\mA_t^{\frac12}}}
\quad {\rm and} \quad
\Delta_t := \Delta_t^+ - \Delta_t^-.
\end{align}

\begin{lemma}[Lemma 3.5 and Lemma 3.7 of~\cite{LZ20}] \label{l:spec-exp-whp-t}
Let $\tau \geq \tau' \geq 1$ be two time steps in the iterative randomized rounding algorithm,
and let $\lambda := \max_{\tau' \leq t \leq \tau} \lambda_{\min}\left(\sum_{i \in S_t} \vv_i \vv_i^\top\right)$.
Let $\Delta_t$ be defined as in~\eqref{eq:Delta} and $k$ be the denominator used in the sampling distributions.
Then
\[
\sum_{t=\tau'+1}^\tau 
\E \left[ \Delta_t \mid S_{t-1}\right] 
\geq \sum_{t=\tau'+1}^\tau \frac{1}{k} \big(1-\gamma-\lambda_{\min}(\mZ_{t-1})\big)
\geq \frac{\tau - \tau'}{k} \cdot (1-\gamma-\lambda).
\]
Furthermore, for any $\eta > 0$,
\[
\Pr \left[ \sum_{t=\tau'+1}^\tau \Delta_t \leq \sum_{t=\tau'+1}^\tau \E[\Delta_t \mid S_{t-1}] - \eta \right] 
\leq \exp\bigg( -\frac{\eta^2 k\sqrt{d}/2}{(\tau - \tau') \gamma(1+\lambda+\gamma) + \eta k\gamma/3} \bigg).
\]
\end{lemma}

\subsection{Martingales} \label{ss:martingale}

A sequence of random variables $Y_1, \ldots, Y_\tau$ is a martingale with respect to a sequence of random variables $Z_1, \ldots, Z_\tau$ if for all $t > 0$, it holds that 
\begin{enumerate}
\item $Y_t$ is a function of $Z_1, \ldots, Z_{t-1}$;
\item $\E[|Y_t|] < \infty$; 
\item $\E[Y_{t+1} | Z_1, \ldots, Z_t] = Y_t$. 
\end{enumerate}

We will use the following theorem by Freedman to bound the probability that $Y_\tau$ is large.

\begin{theorem}[\cite{Fre75,Tro11}] \label{t:Freedman}
Let $\{Y_t\}_t$ be a real-valued martingale with respect to $\{Z_t\}_t$,
and $\{X_t=Y_t - Y_{t-1}\}_t$ be the difference sequence. 
Assume that $X_t \leq R$ deterministically for $1 \leq t \leq \tau$. 
Let $W_t := \sum_{j=1}^t \E[X_j^2 | Z_1, ..., Z_{j-1}]$ for $1 \leq t \leq \tau$.
Then, for all $\delta \geq 0$ and $\sigma^2 > 0$,
\[
\Pr\left( \exists t \in [\tau]: Y_t \geq \delta~{\rm and}~W_t \leq \sigma^2 \right) \leq \exp\left( \frac{-\delta^2/2}{\sigma^2 + R\delta/3} \right).
\]
\end{theorem}

\subsection{Inequalities for Objective Functions} \label{ss:comb-ineq}

To analyze the change of the D-design objective value under exchange operations, we will derive a bound for determinant under rank-two update based on the following well-known matrix determinant lemma.

\begin{lemma} \label{l:mdet}
For any invertible matrix $\mX$ and any vector $\vv \in \R^d$,
\[
\det(\mX \pm \vv \vv^\top) = \det(\mX)(1 \pm \inner{\vv \vv^\top}{\mX^{-1}}).
\]
\end{lemma}

The following determinant lower bound under a rank-two update is a simple consequence and was implicitly contained in~\cite{MSTX19}. 
We provide a proof for completeness.

\begin{lemma} \label{l:Dlower}
Given a matrix $\mA \succ 0$ and two vectors $\vu, \vv \in \R^d$, if $\inner{\vu \vu^\top}{\mA^{-1}} \leq 1$, then
\[
\det(\mA - \vu \vu^\top + \vv \vv^\top) \geq \det(\mA) \left( 1 - \inner{\vu \vu^\top}{\mA^{-1}} \right) \left( 1 + \inner{\vv \vv^\top}{\mA^{-1}} \right).
\]
\end{lemma}
\begin{proof}
We first consider the case when $\inner{\vu \vu^\top}{\mA^{-1}} < 1$. 
This is equivalent to $\vu^\top \mA^{-\frac12} \mA^{-\frac12} \vu < 1$, which implies that $\mA^{-\frac12} \vu \vu^\top \mA^{-\frac12} \prec \mI$ and thus $\mA - \vu \vu^\top \succ 0$. 
Applying Lemma~\ref{l:mdet} twice,
\begin{align*}
\det(\mA - \vu \vu^\top + \vv \vv^\top) & = \det(\mA) \cdot \big(1 - \inner{\vu \vu^\top}{\mA^{-1}}\big) \cdot \big(1 + \inner{\vv \vv^\top}{(\mA - \vu \vu^\top)^{-1}} \big) \\
& \geq \det(\mA) \left( 1 - \inner{\vu \vu^\top}{\mA^{-1}} \right) \left( 1 + \inner{\vv \vv^\top}{\mA^{-1}} \right),
\end{align*}
the last inequality holds as $0 \prec \mA - \vu \vu^\top \preccurlyeq \mA$.

In the case when $\inner{\vu \vu^\top}{\mA^{-1}} = 1$, the RHS of the lemma is zero. 
The same argument as above implies that $\mA - \vu \vu^\top \succcurlyeq 0$, and it follows that the LHS is non-negative.
\end{proof}

Similarly, to analyze the change of A-design objective value under exchange operations, we need an upper bound on the trace of the inverse of a matrix under rank-two update. 
We use the following observation in~\cite{AZLSW20} to derive the inequality below.

\begin{lemma}[Claim 2.10 in~\cite{AZLSW20}] \label{l:inverse-rank2}
Let $\mE = \diag(1,-1) \in \R^{2\times 2}$, $a,c > 0$ and $b \in \R$. If $2c < 1$ and $\left(\begin{smallmatrix} a & b \\ b & c\end{smallmatrix}\right) \succcurlyeq 0$, then
\[
\left( \mE + \begin{pmatrix} a & b \\ b & c \end{pmatrix} \right)^{-1} \succcurlyeq \left( \mE + \begin{pmatrix} 2 a & \\ & 2c \end{pmatrix} \right)^{-1}.
\]
\end{lemma}

\begin{lemma} \label{l:trace-rank2}
Let $\mA \in \R^{d \times d} \succ 0$ and $\vv, \vu \in \R^d$. 
If $2\inner{\vv \vv^\top}{\mA^{-1}} < 1$, then it holds for any $\mX \succcurlyeq 0$ that
\begin{align*}
& \inner{\mX}{\big(\mA - \vv \vv^\top + \vu \vu^\top\big)^{-1}} \leq \inner{\mX}{\mA^{-1}} + \frac{\inner{\mX}{\mA^{-1} \vv \vv^\top \mA^{-1}}}{1 - 2\inner{\vv \vv^\top}{\mA^{-1}}} - \frac{\inner{\mX}{\mA^{-1} \vu \vu^\top \mA^{-1}}}{1 + 2\inner{\vu \vu^\top}{\mA^{-1}}}.
\end{align*}
\end{lemma}

\begin{proof}
Let $\mP := [\vu, \vv] \in \R^{d \times 2}$, then we have
\[
\big(\mA - \vv\vv^\top + \vu\vu^\top\big)^{-1} 
= \big(\mA + \mP \mE \mP^\top\big)^{-1} 
= \mA^{-1} - \mA^{-1} \mP\big(\mE + \mP^\top \mA^{-1}\mP\big)^{-1} \mP^\top \mA^{-1},
\]
where the second equality follows by the Woodbury formula that
\[\big(\mA+ \mU\mC\mV\big)^{-1} = \mA^{-1} - \mA^{-1} \mU\big(\mC^{-1} + \mV \mA^{-1} \mU\big)^{-1} \mV \mA^{-1}.\]
Since
\[
\mP^\top \mA^{-1} \mP = \begin{pmatrix} \inner{\vu\vu^\top}{\mA^{-1}} & \inner{\vv\vu^\top}{\mA^{-1}} \\ \inner{\vu\vv^\top}{\mA^{-1}} & \inner{\vv\vv^\top}{\mA^{-1}} \end{pmatrix} \succcurlyeq 0 \qquad \text{and} \qquad 2\inner{\vv\vv^\top}{\mA^{-1}} < 1,
\]
we can apply Lemma~\ref{l:inverse-rank2} to conclude that
\begin{align*}
\big(\mA - \vv \vv^\top + \vu\vu^\top\big)^{-1} 
& \preceq \mA^{-1} - \mA^{-1}\mP\left(\mE + \Big(\begin{smallmatrix} 2 \inner{\vu\vu^\top}{\mA^{-1}} & \\ & 2\inner{\vv \vv^\top}{\mA^{-1}} \end{smallmatrix}\Big) \right)^{-1}\mP^\top \mA^{-1} \\
& = \mA^{-1} + \frac{\mA^{-1} \vv\vv^\top \mA^{-1}}{1-2\inner{\vv\vv^\top}{\mA^{-1}}} - \frac{\mA^{-1} \vu\vu^\top \mA^{-1}}{1+2\inner{\vu\vu^\top}{\mA^{-1}}}.
\end{align*}
Finally, the lemma follows by noting that $(\mA - \vv\vv^\top + \vu\vu^\top)^{-1} \succ 0$,
because $2\inner{\vv \vv^\top}{\mA^{-1}} < 1$ implies that $\vv \vv^\top \prec \mA$ (as was done in the proof of Lemma~\ref{l:Dlower}). 
\end{proof}

We will also use the following simple claim in the analysis of combinatorial algorithms, whose proof is by checking the derivatives of $f(x)$ and $g(x)$.

\begin{claim} \label{cl:func}
The functions $f(x) = \frac{x - c_1}{c_2+ c_3 \sqrt{x}}$ and $g(x) = \frac{x - c_1}{c_2 + c_3 x}$ with $c_1, c_2, c_3 \geq 0$ are monotone increasing for $x \geq 0$.
\end{claim}

\section{Combinatorial Algorithms} \label{s:comb}

In this section, we present combinatorial local search algorithms for D/A/E-design problems. 
In Section~\ref{ss:comb-D}, we show that Fedorov's exchange method is a polynomial time algorithm to achieve $\frac{b-d-1}{b}$-approximation for D-design, 
which extends the result in~\cite{MSTX19} to the without repetition setting. 
In Section~\ref{ss:comb-A}, we analyze Fedorov's exchange method for A-design, and prove that it works well as long as there is a well-conditioned optimal solution.
As a corollary, this extends the result in~\cite{MSTX19} for A-design to the without repetition setting, with an arguably simpler proof. 
In Section~\ref{ss:comb-E}, we show that Fedorov's exchange method does not work with the minimum eigenvalue objective, and we propose a modified local search algorithm and prove that it works well as long as there is a well-conditioned optimal solution. 

A common theme in the analysis of all these algorithms is that we compare the current integral solution $S$ to an optimal fractional solution $\vx$.
As long as the objective value of $\vx$ is significantly better than that of $S$, we use $\vx$ to define two probability distributions to sample a pair of vectors $\vv_i,\vv_j$ so that the expected objective value of $S-i+j$ improves that of $S$ considerably, and so we can conclude that the combinatorial algorithms will find such an improving exchange pair.
One advantage of this approach is that this allows us the flexibility to compare with a fractional solution with smaller budget (which still has objective value close to the optimal one), and this makes the analysis easier and simpler.

The following notations will be used throughout this section.
Given a fractional solution $\vx \in [0,1]^n$ and an integral solution $S \subseteq [n]$, we denote  
\[
\mX := \sum_{i=1}^n \vx(i) \cdot \vv_i \vv_i^\top,
\qquad 
\mX_{S} := \sum_{i \in S} \vx(i) \cdot \vv_i \vv_i^\top,
\qquad
%{\rm and~}
\vx(S) := \sum_{i \in S} \vx(i),
\qquad
\mZ := \sum_{i \in S} \vv_i \vv_i^\top.
\]

\subsection{Combinatorial Local Search Algorithm for D-Design} \label{ss:comb-D}

We analyze the following version of Fedorov's exchange method for D-design,
where we always choose a pair that maximizes the improvement of the objective value and we stop as soon as the improvement is not large enough.

\begin{framed}
{\bf Fedorov's Exchange Method for D-Design}

Input: $n$ vectors $\vv_1, ..., \vv_n \in \R^d$, a budget $b \geq d$.

    \begin{enumerate}
    
    \item Let $S_0 \subseteq [n]$ be an arbitrary set of full-rank vectors with  $|S_0| = b$. 
    
    \item Let $t \leftarrow 1$ and $\mZ_1 := \sum_{i \in S_{t-1}} \vv_i \vv_i^\top$.
    
    \item {\bf Repeat}
    \begin{enumerate}
        \item Find $i_t \in S_{t-1}$ and $j_t \in [n] \setminus S_{t-1}$ such that 
        \[
        (i_t, j_t) = \underset{(i,j): i \in S_{t-1}, j \in [n] \setminus S_{t-1}}{\arg\max} \det\Big( \mZ_t - \vv_i \vv_i^\top + \vv_j \vv_j^\top \Big).
        \]

        \item Set $S_t \leftarrow S_{t-1} \cup \{ j_t\} \setminus \{i_t\}$ and $\mZ_{t+1} \leftarrow \mZ_t - \vv_{i_t} \vv_{i_t}^\top + \vv_{j_t} \vv_{j_t}^\top $ and $t \gets t + 1$.
    \end{enumerate}
    {\bf Until} $\det(\mZ_t) < \Big( 1 + \frac{d}{4b^3} \Big) \det(\mZ_{t-1})$.
    \item Return $S_{t-2}$ as the solution set. 
\end{enumerate}
\end{framed}

To analyze the change of the objective value in each iteration, note that $\inner{\vv_{i_t} \vv_{i_t}^\top}{\mZ_t^{-1}} \leq 1$ for any $t$ as $i_t \in S_{t-1}$, 
and so it follows from Lemma~\ref{l:Dlower} that
\[
\det(\mZ_{t+1}) = \det(\mZ_t - \vv_{i_t} \vv_{i_t}^\top + \vv_{j_t} \vv_{j_t}^\top) \geq \det(\mZ_t) \cdot (1 - \underbrace{\inner{\vv_{i_t} \vv_{i_t}^\top}{\mZ_t^{-1}}}_{\text{loss}}) \cdot (1 + \underbrace{\inner{\vv_{j_t} \vv_{j_t}^\top}{\mZ_t^{-1}}}_{\text{gain}}).
\]
Therefore, in order to lower bound the determinant of the solution, we lower bound the ``gain'' term and upper bound the ``loss'' term to quantify the progress in each iteration. 
First, we prove the existence of $i_t$ with small loss, with respect to a fractional solution $\vx$ with $\norm{\vx}_1 = q < b$.

\begin{lemma}[Loss] \label{l:loss-D}
For any $\vx \in [0,1]^n$ with $\sum_{i=1}^n \vx(i) = q < b$ and any $S \subseteq [n]$ with $|S|=b$,
there exists $i \in S$ with 
\[
\inner{\vv_i \vv_i^\top}{\mZ^{-1}} \leq \frac{d - \inner{\mX_S}{\mZ^{-1}}}{b - \vx(S)}.
\]
\end{lemma}

\begin{proof}
Consider the probability distribution of removing a vector $\vv_i$ where each $i \in S$ is sampled with probability
$\big(1-\vx(i)\big) / \sum_{j \in S} \big(1-\vx(j)\big)$,
so that the ``staying'' probability is proportional to the value $\vx(i)$.
Note that the denominator is positive as $\vx(S) = q < b$, and thus the probability distribution is well-defined.
Then, the expected loss using this probability distribution is
\begin{align*}
\E\big[\inner{\vv_{i_t} \vv_{i_t}^\top}{\mZ^{-1}}\big] & 
= \frac{\sum_{i \in S} \big(1-\vx(i)\big) \cdot \inner{\vv_i \vv_i^\top}{\mZ^{-1}}}{\sum_{j \in S} \big(1-\vx(j)\big)} = \frac{d - \inner{\mX_S}{\mZ^{-1}}}{b - \vx(S)},
\end{align*}
where the last equality follows as $\sum_{i \in S} \vv_i \vv_i^\top = \mZ$ and $|S|=b$.
Therefore, there must exist one vector $i$ with $\inner{\vv_i \vv_i^\top}{\mZ^{-1}}$ at most the expected value.
\end{proof}

Next, we prove the existence of $j_t$ with large gain, again with respect to a fractional solution $\vx$ with $\norm{\vx}_1 = q < b$.

\begin{lemma}[Gain] \label{l:gain-D}
For any $\vx \in [0,1]^n$ with $\sum_{i=1}^n \vx(i) = q < b$ and any $S \subseteq [n]$ with $|S|=b$ and $\vx(S) < q$,
there exists $j \in [n] \backslash S$ with 
\begin{align*}
    \inner{\vv_j \vv_j^\top}{\mZ^{-1}} \geq \frac{\inner{\mX}{\mZ^{-1}} - \inner{\mX_S}{\mZ^{-1}}}{q - \vx(S)}.
\end{align*}
\end{lemma}

\begin{proof}
Consider the probability distribution of adding a vector $\vv_j$ where each $j \in [n]\setminus S$ is sampled with probability
$\vx(j) / \sum_{i \in [n] \backslash S} \vx(i)$,
so that the ``adding'' probability is proportional to the value $\vx(i)$.
Note that the denominator is positive by our assumption that $\vx(S) < q$, and so the probability distribution is well-defined.
Then, the expected gain using this probability distribution is
\begin{align*}
\E[\inner{\vv_j \vv_j^\top}{\mZ^{-1}}] 
& = \frac{\sum_{j \in [n] \backslash S} \vx(j) \cdot \inner{\vv_j \vv_j^\top}{\mZ^{-1}}}{\sum_{i \in [n] \backslash S} \vx(i)} 
= \frac{\inner{\mX}{\mZ^{-1}} - \inner{\mX_S}{\mZ^{-1}}}{q - \vx(S)}.
\end{align*}
Therefore, there must exist one vector $j$ with $\inner{\vv_j \vv_j^\top}{\mZ^{-1}}$ at least the expected value.
\end{proof}

We are about ready to analyze when the objective value would increase.
The following lemma will be used to relate the numerator of the gain term to the current objective value $\det(\mZ)$.

\begin{lemma} \label{l:det}
For any given $d \times d$ positive definite matrices $\mA, \mB \succ 0$, 
\[
\inner{\mA}{\mB} \geq d \cdot \det(\mA)^{\frac{1}{d}} \cdot \det(\mB)^{\frac{1}{d}}.
\]
\end{lemma}
\begin{proof}
Let $\mA = \sum_{i=1}^d a_i \vu_i \vu_i^\top$ and $\mB = \sum_{j=1}^d b_j \vw_j \vw_j^\top$ be the spectral decompositions of $\mA$ and $\mB$.
\begin{align*}
\frac{1}{d} \cdot \inner{\mA}{\mB} 
= \sum_{1 \leq i,j \leq d} a_i b_j \cdot \frac{\inner{\vu_i}{\vw_j}^2}{d} 
& \geq \prod_{1 \leq i,j \leq d} \big( a_i b_j \big)^{\frac{\inner{\vu_i}{\vw_j}^2}{d}} 
\\
& = \bigg(\prod_{i=1}^d \prod_{j=1}^d a_i^{\frac{\inner{\vu_i}{\vw_j}^2}{d}}\bigg) \bigg(\prod_{j=1}^d \prod_{i=1}^d b_j^{\frac{\inner{\vu_i}{\vw_j}^2}{d}} \bigg) 
\\
& = \prod_{i=1}^d a_i^{\frac{1}{d}} \cdot \prod_{j=1}^d b_j^{\frac{1}{d}} 
\\
& = \det(\mA)^{\frac{1}{d}} \det(\mB)^{\frac{1}{d}},
\end{align*}
where the inequality follows by the weighted AM-GM inequality as $\sum_{i,j=1}^d \inner{\vu_i}{\vw_j}^2 = d$, and the second last equality follows as $\{\vu_i\}_{i=1}^d$ and $\{\vw_j\}_{j=1}^d$ are orthonormal bases.
\end{proof}

The following is the main technical result for D-design,
which lower bounds the improvement of the objective value in each iteration.
In the proof, we compare our current integral solution $S$ with size $b$ to a fractional solution $\vy$ with size $q = b - d - \frac{1}{2}$.

\begin{proposition}[Progress] \label{prop:progress-D}
Let $\vx \in [0,1]^n$ be a feasible solution to the convex programming relaxation \eqref{eq:convex} for D-design with $\sum_{i=1}^n \vx(i) = b$ for $b \geq d+1$. Let $\mZ_t$ be the current solution in the $t$-th iteration of Fedorov's exchange method. Then
\[
\det(\mZ_t)^{\frac1d} \leq \frac{b-d-1}{b} \cdot \det(\mX)^{\frac1d}
\quad \implies \quad
\det(\mZ_{t+1}) \geq \Big( 1 + \frac{d}{4b^3} \Big) \cdot \det(\mZ_t).
\]
\end{proposition}

\begin{proof}
We consider the following scaled-down version $\vy,\mY$ of the fractional solution $\vx,\mX$.
Define 
\[
q := b-d-\frac{1}{2}, \qquad
\vy := \frac{q}{b} \cdot \vx, \qquad
\mY := \sum_{i=1}^n \vy(i) \cdot \vv_i \vv_i^\top = \frac{q}{b} \cdot \mX.
\]
Note that $\det(\mY)^{\frac{1}{d}} = \frac{q}{b} \cdot \det(\mX)^{\frac{1}{d}}$
and $\frac{1}{2} \leq q < b$.
Let $S := S_{t-1}$ be the current solution set at time $t$.
Note that we can assume $\vx(S) < b$ and hence $\vy(S) < q$, as otherwise $\det(\mZ_t) \geq \det(\mX)$ and there is nothing to prove.
Hence, we can apply Lemma~\ref{l:loss-D} and Lemma~\ref{l:gain-D} on $\mY$ and $S$ to ensure the existence of $i_t \in S$ and $j_t \in [n] \setminus S$ such that
\begin{align*}
\det(\mZ_{t+1}) 
& \geq \det(\mZ_t) \cdot \left( 1 - \frac{d - \inner{\mY_S}{\mZ_t^{-1}}}{b - \vy(S)} \right) 
\cdot \left( 1 + \frac{\inner{\mY}{\mZ_t^{-1}} - \inner{\mY_S}{\mZ_t^{-1}}}{q - \vy(S)} \right) \\
& \geq \det(\mZ_t) \cdot \left( 1 - \frac{d - \inner{\mY_S}{\mZ_t^{-1}}}{b - \vy(S)} \right) \cdot \left( 1 + \frac{d \cdot \det(\mY)^{\frac1d} \cdot \det(\mZ_t^{-1})^{\frac1d} - \inner{\mY_S}{\mZ_t^{-1}}}{q - \vy(S)} \right) \\
& \geq \det(\mZ_t) \cdot \left( 1 - \frac{d - \inner{\mY_S}{\mZ_t^{-1}}}{b - \vy(S)} \right) \cdot \left( 1 + \frac{ d \cdot \frac{q}{b} \det(\mX)^{\frac1d} \cdot \frac{b}{b-d-1} \det(\mX)^{-\frac1d} - \inner{\mY_S}{\mZ_t^{-1}}}{q - \vy(S)} \right) \\
& = \det(\mZ_t) \cdot \left( 1 - \frac{d - \inner{\mY_S}{\mZ_t^{-1}}}{b - \vy(S)} \right) \cdot \left( 1 + \frac{ \big(1 + \frac{1}{2q-1}\big) d - \inner{\mY_S}{\mZ_t^{-1}}}{q - \vy(S)} \right),
\end{align*}
where the second inequality follows from Lemma~\ref{l:det}, the third inequality follows from $\mY = \frac{q}{b} \mX$ and the assumption $\det(\mZ)^{\frac1d} \leq \frac{b-d-1}{b} \det(\mX)^{\frac1d}$, and the last equality is by $q = b - d - \frac12$.

To lower bound the improvement, we write $a:= d - \inner{\mY_S}{\mZ_t^{-1}}$ as a shorthand, and then the multiplicative factor is
\begin{align*}
    \left( 1 - \frac{d - \inner{\mY_S}{\mZ_t^{-1}}}{b - \vy(S)} \right) \cdot \left( 1 + \frac{\big(1+\frac{1}{2q-1}\big)d - \inner{\mY_S}{\mZ_t^{-1}}}{q - \vy(S)} \right) 
    & = \left( 1 - \frac{a}{b - \vy(S)} \right) \cdot \left( 1 + \frac{a + \frac{d}{2q-1}}{q - \vy(S)} \right) \\
    & = 1 + \frac{(b-q)a - a^2 + \frac{(b-\vy(S))d}{2q-1} - \frac{ad}{2q-1}}{\big(b-\vy(S)\big) \cdot \big(q - \vy(S)\big)} \\
    & \geq 1 + \frac{(b-q)a - a^2 + \frac{(b-q)d}{2q-1} - \frac{ad}{2q-1}}{\big(b-\vy(S)\big) \cdot \big(q - \vy(S)\big)} ,
\end{align*}
where the last inequality follows as $\vy(S) \leq q$.
Let $f(x) = -x^2 + (b-q)x - \frac{dx}{2q-1} + \frac{(b-q)d}{2q-1}$ be a univariate quadratic function in $x$. 
Note that $f''(x) < 0$, and thus $\min_{x \in [x_1,x_2]} f(x)$ is attained at one of the two ends $x=x_1$ or $x=x_2$. 
Since $a = d - \inner{\mY_S}{\mZ_t^{-1}} \in [0,d]$ (the expected loss is nonnegative and $\mY_S, \mZ \succcurlyeq 0$), 
the numerator of the second term above is lower bounded by
\begin{align*}
f(a) \geq \min_{x \in [0,d]} f(x) \geq \min\{ f(0), f(d)\} & = \min\left\{ \frac{(b-q)d}{2q-1}, \frac{2qd(b-q-d)}{2q-1} \right\} \\
& = \min\left\{ \frac{(d+\frac12)d}{2(b-d-1)}, \frac{(b-d-\frac12)d}{2(b-d-1)} \right\} 
\geq \frac{d}{4(b-d-1)},
\end{align*}
where the equality in the second line is by plugging in $q = b- d - \frac12$, 
and the last inequality follows the assumption $b \geq d+1$.
Therefore, we conclude that
\begin{align*}
    \det(\mZ_{t+1}) & \geq \det(\mZ_t) \cdot \bigg(1 + \frac{d}{4(b-d-1)\big(b-\vy(S)\big) \big(q - \vy(S)\big)} \bigg) \geq \det(\mZ_t) \cdot \left(1 + \frac{d}{4b^3} \right). \qedhere
\end{align*}
\end{proof}

The main result in this subsection follows immediately from Proposition~\ref{prop:progress-D}.

\CombD*

\begin{proof}
Let $\mX^* = \sum_{i=1}^n \vx^*(i) \cdot \vv_i \vv_i^\top$ be an optimal fractional solution for D-design with budget $b$ for $\vx^* \in [0,1]^n$.
Let $\mZ_1 \succ 0$ be an arbitrary initial solution. 

When the combinatorial local search algorithm terminates at the $\tau$-th iteration, 
the termination condition implies that $\det(\mZ_{\tau+1}) < \Big(1 + \frac{d}{4b^3} \Big) \det(\mZ_{\tau})$. 
It follows from Proposition~\ref{prop:progress-D} with $\mX = \mX^*$ that
\[
\det(\mZ_{\tau})^{\frac1d} \geq  \frac{b-d-1}{b} \cdot \det(\mX^*)^{\frac1d},
\]
and thus the returned solution of the Fedorov's exchange method is an $\frac{b-d-1}{b}$-approximate solution.

Finally, we bound the time complexity of the algorithm.
If the algorithm runs for $\tau > \frac{8b^3}{d} \ln \frac{\det(\mX^*)}{\det(\mZ_1)}$ iterations, 
then the termination condition implies that the determinant of $\mZ_{\tau+1}$ is at least
\[
\det(\mZ_{\tau+1}) \geq \left( 1 + \frac{d}{4b^3} \right)^\tau \cdot \det(\mZ_1) \geq e^{\frac{d\tau}{8b^3}} \cdot \det(\mZ_1) > \det(\mX^*),
\]
where the second inequality follows as $(1+\frac{d}{4b^3}) \geq e^\frac{d}{8b^3}$ for $\frac{d}{4b^3} \leq \frac14$.
It was proved in Appendix~C~of~\cite{MSTX19} that $\ln \frac{\det(\mX^*)}{\det(\mZ_1)}$ is polynomial in $d,b$ and $\ell$, which is the maximum number of bits to represent the numbers in the entries of the vectors.
Specifically, they proved that $\det(\mZ_1) \geq 2^{-4(2b\ell+1)d^2}$ and $\det(\mX^*) \leq 2^{4(2n\ell+1)d^2}$, and so $\tau = O(db^3n\ell)$ iterations of the algorithm is enough.
\end{proof}

\subsection{Combinatorial Local Search Algorithm for A-Design} \label{ss:comb-A}

We analyze the following version of Fedorov's exchange method for A-design,
where we always choose a pair that maximizes the improvement of the objective value and we stop as soon as the improvement is not large enough.

\begin{framed}
{\bf Fedorov's Exchange Method for A-Optimal Design}

Input: $n$ vectors $\vv_1, ..., \vv_n \in \R^d$, a budget $b \geq d$, and an accuracy parameter $\eps \in (0,1)$.

    \begin{enumerate}
    
    \item Let $S_0 \subseteq [n]$ be an arbitrary set of full-rank vectors with $|S_0| = b$. 
    
    \item Let $t \leftarrow 1$ and $\mZ_1 \leftarrow \sum_{i \in S_0} \vv_i \vv_i^\top$.

    \item {\bf Repeat}
    \begin{enumerate}
        \item Find $i_t \in S_{t-1}$ and $j_t \in [n] \setminus S_{t-1}$ such that 
        \[
        (i_t, j_t) = \underset{(i,j): i \in S_{t-1}, j \in [n] \setminus S_{t-1}}{\arg\min} \tr\Big( \Big( \mZ_t - \vv_i \vv_i^\top + \vv_j \vv_j^\top \Big)^{-1} \Big).
        \]

        \item Set $S_t \leftarrow S_{t-1} \cup \{ j_t\} \setminus \{i_t\}$ and $\mZ_{t+1} \leftarrow \mZ_t - \vv_{i_t} \vv_{i_t}^\top + \vv_{j_t} \vv_{j_t}^\top $ and $t \gets t + 1$.
    \end{enumerate}
    {\bf Until} $\tr(\mZ_t^{-1}) > \Big( 1 - \frac{\eps}{b} \Big) \tr(\mZ_{t-1}^{-1}) $.
    \item Return $S_{t-2}$ as the solution set. 
\end{enumerate}
\end{framed}

To analyze the change of the objective value in each iteration,
we apply Lemma~\ref{l:trace-rank2} which states that if $2\inner{\vv_{i_t} \vv_{i_t}^\top}{\mZ_t^{-1}} < 1$ then
\begin{equation} \label{eq:progress-A}
\tr (\mZ_{t+1}^{-1}) - \tr (\mZ_t^{-1}) \leq \underbrace{\frac{\inner{\vv_{i_t} \vv_{i_t}^\top}{\mZ_t^{-2}}}{1 - 2 \inner{\vv_{i_t} \vv_{i_t}^\top}{\mZ_t^{-1}}}}_{\text{\rm loss}} - \underbrace{\frac{\inner{\vv_{j_t} \vv_{j_t}^\top}{\mZ_t^{-2}}}{1+2 \inner{\vv_{j_t} \vv_{j_t}^\top}{\mZ^{-1}_t}}}_{\text{\rm gain}}.
\end{equation}
Therefore, to upper bound the A-design objective of the solution,
we upper bound the loss term and lower bound the gain term to quantify the progress in each iteration.

In the following lemma, we first prove the existence of $i_t$ with small loss term, with respect to a fractional solution $\vx$ with $\norm{\vx}_1 = q < b - 2d$.
Note that we only restrict our choice of $i_t$ to those vectors that satisfy $2\inner{\vv_{i_t} \vv_{i_t}^\top}{\mZ_t^{-1}} < 1$ so that~\eqref{eq:progress-A} applies, clearly Fedorov's exchange method could only do better by considering all possible vectors in the current solution.

\begin{lemma}[Loss] \label{l:loss-A}
For any $\vx \in [0,1]^n$ with $\sum_{i=1}^n \vx(i) = q < b - 2d$ and any $S \subseteq [n]$ with $|S|=b$,
there exists $i \in S':= \{j \in S: 2\inner{\vv_j\vv_j^\top}{\mZ^{-1}} < 1 \}$ with
\[
\frac{\inner{\vv_i \vv_i^\top}{\mZ^{-2}}}{1 - 2 \inner{\vv_i \vv_i^\top}{\mZ^{-1}}} \leq \frac{\tr (\mZ^{-1})  - \inner{\mX_S}{\mZ^{-2}}}{b - \vx(S) - 2d}.
\]
\end{lemma}
\begin{proof}
Consider the probability distribution of removing a vector $\vv_i$ with probability 
\[
\Pr[i_t = i] =
\frac{\big(1-\vx(i)\big) \cdot \big(1 - 2\inner{\vv_i \vv_i^\top}{\mZ^{-1}}\big)}{\sum_{j \in S'} \big(1-\vx(j)\big) \cdot \big(1 - 2\inner{\vv_j \vv_j^\top}{\mZ^{-1}}\big) } \quad \text{for each } i \in S'.
\]
We first check that the probability distribution is well-defined.
Note that the numerator is non-negative as $1-2\inner{\vv_i \vv_i^\top}{\mZ^{-1}} > 0$ for each $i \in S'$.
The denominator is 
\begin{align*}
    \sum_{j \in S'} \big(1-\vx(j)\big) \cdot \big(1 - 2 \inner{\vv_j \vv_j^\top}{\mZ^{-1}}\big) 
& \geq \sum_{j \in S} \big(1-\vx(j)\big) \cdot \big(1-2\inner{\vv_j \vv_j^\top}{\mZ^{-1}}\big) \nonumber 
\\
    & \geq \sum_{j \in S} \big(1-\vx(j)\big) - 2\sum_{j \in S} \inner{\vv_j\vv_j^\top}{\mZ^{-1}} \nonumber 
\\
    & = b - \vx(S) - 2d ~>~ 0, 
\end{align*}
where the first inequality holds as $1-2\inner{\vv_j \vv_j^\top}{\mZ^{-1}} \leq 0$ for $j \in S\setminus S'$, 
the second inequality follows from $1-\vx(j) \leq 1$ for each $j \in [n]$,
and the equality is by $|S| = b$ and $\inner{ \sum_{j \in S} \vv_j \vv_j^\top}{\mZ^{-1}} = \inner{\mZ}{\mZ^{-1}} = d$, 
and the strict inequality is by the assumption $b > q + 2d \geq \vx(S) + 2d$. 
Thus, $\Pr[i_t = i] \geq 0$ for each $i \in S'$, and clearly $\sum_{i \in S'} \Pr[i_t=i]=1$.

The expected loss using this probability distribution is
\begin{align*}
\E\left[\frac{\inner{\vv_{i_t} \vv_{i_t}^\top}{\mZ^{-2}}}{1 - 2 \inner{\vv_{i_t} \vv_{i_t}^\top}{\mZ^{-1}}}\right] 
& = \sum_{i \in S'} \frac{\big(1-\vx(i)\big) \cdot \big(1-2\inner{\vv_i \vv_i^\top}{\mZ^{-1}}\big)}{\sum_{j \in S'} \big(1-\vx(j)\big) \cdot \big(1-2\inner{\vv_j \vv_j^\top}{\mZ^{-1}}\big)} \cdot \frac{\inner{\vv_i \vv_i^\top}{\mZ^{-2}}}{1 - 2 \inner{\vv_i \vv_i^\top}{\mZ^{-1}}} 
\\
& = \frac{\sum_{i \in S'} \big(1-\vx(i)\big) \cdot \inner{\vv_i \vv_i^\top}{\mZ^{-2}}}{\sum_{j \in S'} \big(1-\vx(j)\big) \cdot \big(1 - 2 \inner{\vv_j \vv_j^\top}{\mZ^{-1}}\big)} 
\\
& \leq \frac{\tr(\mZ^{-1}) - \inner{\mX_S}{\mZ^{-2}}}{b - \vx(S) - 2d},
\end{align*}
where the last inequality follows from 
the inequality above for the denominator and 
\[
\sum_{i \in S'}\big(1-\vx(i)\big) \cdot \inner{\vv_i \vv_i^\top}{\mZ^{-2}} 
\leq \sum_{i \in S} \big(1-\vx(i)\big) \cdot \inner{\vv_i \vv_i^\top}{\mZ^{-2}} 
= \inner{\mZ}{\mZ^{-2}} - \inner{\mX_S}{\mZ^{-2}} 
= \tr (\mZ^{-1}) - \inner{\mX_S}{\mZ^{-2}}
\]
for the numerator. 
Therefore, there exists an $i \in S'$ with loss at most the expected value.
\end{proof}

Next we show the existence of $j_t$ with large gain term, again with respect to a fractional solution $\vx$.

\begin{lemma}[Gain] \label{l:gain-A}
For any $\vx \in [0,1]^n$ with $\sum_{i=1}^n \vx(i) = q < b$ and any $S \subseteq [n]$ with $|S| = b$ and $\vx(S) < q$,
there exists $j \in [n] \setminus S$ with
\[
\frac{\inner{\vv_j \vv_j^\top}{\mZ^{-2}}}{1 + 2 \inner{\vv_j \vv_j^\top}{\mZ^{-1}}} \geq \frac{\inner{\mX}{\mZ^{-2}}  - \inner{\mX_S}{\mZ^{-2}}}{q - \vx(S) + 2\inner{\mX}{\mZ^{-1}}}.
\]
\end{lemma}
\begin{proof}
Consider the probability distribution of adding a vector $\vv_j$ where each $j \in [n]\setminus S$ is sampled with probability
\[
\Pr[j_t = j] = \frac{\vx(j) \cdot \big(1+2\inner{\vv_j \vv_j^\top}{\mZ^{-1}}\big) }{ \sum_{i \in [n] \setminus S} ~\vx(i) \cdot \big(1+2\inner{\vv_i \vv_i^\top}{\mZ^{-1}}\big) }
\quad \text{for each } j \in [n] \setminus S.
\] 
Note that the denominator is positive by the assumption $\vx(S) < q$ which implies that $\vx([n] \setminus S) > 0$, and so the probability distribution is well-defined.

The expected gain using this probability distribution is
\begin{align*}
    \E\left[\frac{\inner{\vv_{j_t} \vv_{j_t}^\top}{\mZ^{-2}}}{1 + 2 \inner{\vv_{j_t} \vv_{j_t}^\top}{\mZ^{-1}}}\right] 
    & = \sum_{j \in [n] \backslash S} \frac{\vx(j) \cdot \big(1+2\inner{\vv_j \vv_j^\top}{\mZ^{-1}}\big)}{\sum_{i \in [n] \backslash S} ~\vx(i) \cdot \big(1+2\inner{\vv_i \vv_i^\top}{\mZ^{-1}}\big)} \cdot \frac{\inner{\vv_j \vv_j^\top}{\mZ^{-2}}}{1 + 2 \inner{\vv_j \vv_j^\top}{\mZ^{-1}}} 
\\
    & = \frac{\sum_{j \in [n] \backslash S} ~\vx(j) \cdot \inner{\vv_j \vv_j^\top}{\mZ^{-2}}}{\sum_{i \in [n] \backslash S}  ~\vx(i) \cdot \big(1 + 2 \inner{\vv_i \vv_i^\top}{\mZ^{-1}}\big)}  
\\
    & = \frac{\inner{\mX}{\mZ^{-2}} - \inner{\mX_S}{\mZ^{-2}}}{q - \vx(S) + \sum_{i \in [n] \backslash S} 2\vx(i) \cdot \inner{\vv_i \vv_i^\top}{\mZ^{-1}}} 
\\
    & \geq \frac{\inner{\mX}{\mZ^{-2}} - \inner{\mX_S}{\mZ^{-2}}}{q - \vx(S) + 2\inner{\mX}{\mZ^{-1}}},
\end{align*}
where the third equality is by $\sum_{i=1}^n \vx(i) = q$, 
and the last inequality holds as $\sum_{i \in [n] \backslash S} \vx(i) \cdot \vv_i \vv_i^\top \preccurlyeq \mX$. 
Therefore, there exists $j \in [n] \backslash S$ with gain at least the expected value.
\end{proof}

We are about ready to analyze when the objective value would decrease.
The following lemma will be used to bound the denominator of the gain term,
and also to relate the numerator of the gain term to the current objective value $\tr(\mZ^{-1})$.

\begin{lemma} \label{l:trace}
For any given $d \times d$ positive definite matrices $\mA, \mB \succ 0$,
\begin{align}
& \inner{\mA}{\mB^2} \geq \frac{(\tr(\mB))^2}{\tr\big(\mA^{-1}\big)} \qquad \qquad \text{\rm and} \label{eq:inner_lower} \\
& \inner{\mA}{\mB} \leq \sqrt{\tr(\mA) \cdot \inner{\mA}{\mB^2}}. \label{eq:inner_upper}
\end{align}
\end{lemma}
\begin{proof}
Let $\mA = \sum_{i=1}^d a_i \vu_i \vu_i^\top$ and $\mB = \sum_{j=1}^d b_j \vw_j \vw_j^\top$ be the eigendecomposition of $\mA$ and $\mB$. 
Then,
\begin{align*}
\tr(\mB) = \sum_{j=1}^d b_j 
= \sum_{1 \leq i,j \leq d} b_j \cdot \inner{\vu_i}{\vw_j}^2 
& = \sum_{1 \leq i,j \leq d} \sqrt{a_i} b_j \inner{\vu_i}{\vw_j} \cdot \frac{1}{\sqrt{a_i}} \inner{\vu_i}{\vw_j} 
\\
& \leq \sqrt{\sum_{1 \leq i,j \leq d} a_i b_j^2 \inner{\vu_i}{\vw_j}^2 \cdot \sum_{1 \leq i,j \leq d} \frac{1}{a_i} \inner{\vu_i}{\vw_j}^2} 
\\
& = \sqrt{ \inner{\mA}{\mB^2} \cdot  \tr(\mA^{-1})},
\end{align*}
where the second equality and the last equality hold as $\{\vu_i\}_{i=1}^d$ and $\{\vw_j\}_{j=1}^d$ are orthonormal bases, 
and the inequality is by Cauchy-Schwarz.
For the second inequality,
\begin{align*}
\inner{\mA}{\mB} = \sum_{1 \leq i,j \leq d} a_i b_j \inner{\vu_i}{\vw_j}^2 
\leq \sqrt{\sum_{1 \leq i,j \leq d} a_i \inner{\vu_i}{\vw_j}^2 \cdot \sum_{1 \leq i,j \leq d} a_i b_j^2 \inner{\vu_i}{\vw_j}^2} 
= \sqrt{\tr(\mA) \cdot \inner{\mA}{\mB^2}},
\end{align*}
where the equalities hold as $\{\vu_i\}$ and $\{\vw_j\}$ are orthonormal bases and the inequality is by Cauchy-Schwarz.
\end{proof}

The following is the main technical result for A-design,
which lower bounds the improvement of the objective value in each iteration.
Note that the result depends on $\tr(\mX) \cdot \tr(\mX^{-1})$.

\begin{proposition}[Progress] \label{prop:progress-A}
Let $\vx \in [0,1]^n$ be a fractional solution with $\sum_{i=1}^n \vx(i) = q$.
Let $\mZ_t$ be the current solution in the $t$-th iteration of Fedorov's exchange method. 
For any $\eps > 0$, if
\[
\tr (\mZ_t^{-1}) \geq (1+\eps) \tr(\mX^{-1}) \qquad \text{and} \qquad b \geq q + 2d + 2(1+\eps)\sqrt{\tr (\mX) \cdot \tr( \mX^{-1})},
\]
then
\[
\tr \big(\mZ_{t+1}^{-1}\big) \leq \left(1- \frac{\eps}{b} \right) \cdot \tr \big(\mZ_t^{-1}\big). 
\]
\end{proposition}

\begin{proof}
Let $S := S_{t-1}$ be the current solution set at time $t$. 
Note that $\vx(S) < q$, as otherwise $\tr(\mZ^{-1}) \leq \tr(\mX^{-1})$ and the assumption does not hold.
Hence, we can apply Lemma~\ref{l:gain-A} to prove the existence of a $j_t \in [n] \backslash S$ such that the gain term is
\begin{eqnarray*}
\frac{\inner{\vv_{j_t} \vv_{j_t}^\top}{\mZ^{-2}}}{1 + 2 \inner{\vv_{j_t} \vv_{j_t}^\top}{\mZ^{-1}}} 
& \geq & \frac{\inner{\mX}{\mZ^{-2}}  - \inner{\mX_S}{\mZ^{-2}}}{q - \vx(S) + 2\inner{\mX}{\mZ^{-1}}} 
\\
& \geq & \frac{\inner{\mX}{\mZ^{-2}}  - \inner{\mX_S}{\mZ^{-2}}}{q - \vx(S) + 2\sqrt{\tr (\mX) \cdot \inner{\mX}{\mZ^{-2}}}}
\\
& \geq & \frac{\frac{\tr (\mZ^{-1})^2}{\tr(\mX^{-1})} - \inner{\mX_S}{\mZ^{-2}}}{q - \vx(S) + 2\sqrt{\tr (\mX) \cdot \frac{\tr (\mZ^{-1})^2}{\tr(\mX^{-1})}}}
\\ 
& = & \frac{\frac{\tr (\mZ^{-1})}{\tr(\mX^{-1})} \cdot \tr (\mZ^{-1})- \inner{\mX_S}{\mZ^{-2}} }{q- \vx(S) + \frac{2\tr (\mZ^{-1})}{\tr(\mX^{-1})} \cdot \sqrt{\tr (\mX) \cdot \tr (\mX^{-1})}}   
\\
& \geq & \frac{(1+\eps)\tr (\mZ^{-1})- \inner{\mX_S}{\mZ^{-2}} }{q- \vx(S) + 2(1+\eps) \sqrt{\tr (\mX) \cdot \tr (\mX^{-1})}} 
\\
& \geq & \frac{(1+\eps)\tr (\mZ^{-1})- \inner{\mX_S}{\mZ^{-2}} }{b - \vx(S) -2d}, 
\end{eqnarray*}
where the second inequality is by \eqref{eq:inner_upper},
the third inequality follows from $\inner{\mX}{\mZ^{-2}} \geq \frac{(\tr (\mZ^{-1}))^2}{\tr(\mX^{-1})}$ by~\eqref{eq:inner_lower} and an application of Claim~\ref{cl:func} with $f(x) = \frac{x - c_1}{c_2 + c_3\sqrt{x}}$ to establish monotonicity,
the fourth inequality follows from the first assumption that $\tr (\mZ^{-1}) \geq (1+\eps) \tr (\mX^{-1})$ and another application of Claim~\ref{cl:func} with $g(x) = \frac{x - c_1}{c_2 + c_3 x}$ to establish monotonicity,
and the last inequality follows from the second assumption that $b \geq q + 2d + 2(1+\eps) \sqrt{\tr (\mX) \cdot \tr(\mX^{-1})}$.  

For the loss term, note that $q < b-2d$ by the assumption on $b$, 
and so we can apply Lemma~\ref{l:loss-A} to prove the existence of an $i_t \in S' \subseteq S$ such that the loss term is
\begin{align*}
\frac{\inner{\vv_{i_t} \vv_{i_t}^\top}{\mZ^{-2}}}{1 - 2 \inner{\vv_{i_t} \vv_{i_t}^\top}{\mZ^{-1}}} \leq \frac{\tr (\mZ^{-1})  - \inner{\mX_S}{\mZ^{-2}}}{b - \vx(S) - 2d}.
\end{align*}

Since $i_t \in S'$ satisfies $2\inner{\vv_{i_t} \vv_{i_t}^\top}{\mZ_t^{-1}} < 1$, 
we can apply~\eqref{eq:progress-A} to conclude that
\begin{align*}
\tr (\mZ_{t+1}^{-1}) - \tr (\mZ_t^{-1}) 
& = \tr \left((\mZ_t - \vv_{i_t} \vv_{i_t}^\top + \vv_{j_t} \vv_{j_t}^\top)^{-1}\right) - \tr \left(\mZ_t^{-1}\right) 
\\
& \leq \frac{\inner{\vv_{i_t} \vv_{i_t}^\top}{\mZ^{-2}}}{1 - 2 \inner{\vv_{i_t} \vv_{i_t}^\top}{\mZ^{-1}}} - \frac{\inner{\vv_{j_t} \vv_{j_t}^\top}{\mZ^{-2}}}{1 + 2 \inner{\vv_{j_t} \vv_{j_t}^\top}{\mZ^{-1}}} 
\leq \frac{-\eps \tr (\mZ_t^{-1})}{b - \vx(S) - 2d}  \leq -\frac{\eps }{b} \tr (\mZ_t^{-1}). \qedhere
\end{align*}
\end{proof}

The main result in this subsection follows from Proposition~\ref{prop:progress-A} by a simple scaling argument.

\CombA*

\begin{proof}
We consider the following scaled-down version $\vy, \mY$ of the fractional solution $\vx,\mX$.
Let
\[
q := b - 2d - 2(1+\eps)\sqrt{\tr(\mX) \cdot \tr((\mX)^{-1})}, \qquad
\vy := \frac{q}{b} \cdot \vx, \qquad
\mY := \sum_{i=1}^n \vy(i) \cdot \vv_i \vv_i^\top = \frac{q}{b} \cdot \mX. 
\]
Note that $\tr(\mY) \cdot \tr\left(\mY^{-1}\right) = \tr(\mX) \cdot \tr\left(\mX^{-1}\right)$ and so it holds that $b \geq q + 2d + 2(1+\eps) \sqrt{\tr(\mY) \cdot \tr(\mY^{-1})}$.
Thus, we can apply Proposition~\ref{prop:progress-A} on $\vy$ to conclude that if the algorithm terminates at the $\tau$-th iteration such that $\tr\left(\mZ_{\tau+1}^{-1}\right) > \Big( 1 - \frac{\eps}{b} \Big) \tr\left(\mZ_{\tau}^{-1}\right)$ then
\[
\tr\left(\mZ_{\tau}^{-1}\right) 
< (1+\eps) \cdot \tr\left(\mY^{-1}\right) 
= \frac{(1+\eps)b}{q} \cdot \tr\left(\mX^{-1}\right) 
\leq \big(1+O(\eps)\big) \cdot \tr\left(\mX^{-1}\right),
\]
where the last inequality follows from the assumption $b = \Omega\Big( \frac{1}{\eps} \big( d + \sqrt{\tr(\mX) \tr(\mX^{-1})} \big) \Big)$ which implies that $q \geq (1-O(\eps)) b$. 
This proves the approximation guarantee of the returned solution.

Finally, we bound the time complexity of the algorithm.
If the algorithm runs for $\tau > \frac{b}{\eps} \ln \frac{\tr(\mZ_1^{-1})}{\tr(\mX^{-1})}$ iterations, then the termination condition implies that the objective value of $\mZ_{\tau+1}$ is at most
\[
\tr(\mZ_{\tau+1}^{-1}) 
\leq \left( 1 - \frac{\eps}{b} \right)^{\tau} \cdot \tr\left(\mZ^{-1}_1\right) 
\leq e^{- \frac{\eps \tau}{b}} \cdot \tr\left(\mZ^{-1}_1\right) 
\leq \tr\left(\mX^{-1}\right).
\]
Note that $\ln \frac{\tr(\mZ_1^{-1})}{\tr(\mX^{-1})}$ is upper bounded by a polynomial in $d,n$ and the input size as proved in~\cite{MSTX19} (and the corresponding bound for D-design is discussed in the proof of Theorem~\ref{t:D-combin} in Section~\ref{ss:comb-D}).
\end{proof}

As a corollary, we extend the analysis of Fedorov's exchange method in~\cite{MSTX19} to the more general without repetition setting.

\begin{corollary}
Let $\vx \in [0,1]^n$ be a fractional solution to the convex programming relaxation~\eqref{eq:convex} for A-design with $\sum_{i=1}^n \vx(i) = b$.
If $\norm{\vv_i}^2 \leq \frac{ \eps^2 b}{2 \tr(\mX^{-1})}$ for each $1 \leq i \leq n$ and $b \geq \Omega\left(\frac{d}{\eps}\right)$ for some $\eps \in (0,1)$,
then Fedorov's exchange method for A-design returns a solution with at most $b$ vectors with objective value at most $\big(1+O(\eps)\big) \cdot \tr\left(\mX^{-1}\right)$ in polynomial time.
\end{corollary}
\begin{proof}
It follows from the assumption $\norm{\vv_i}^2 \leq \frac{ \eps^2 b}{2\tr(\mX^{-1})}$ that 
\[
\tr\left(\mX^{-1}\right) \cdot \tr(\mX) 
= \tr\left(\mX^{-1}\right) \cdot \sum_{i=1}^n \vx(i) \cdot \|\vv_i\|_2^2
\leq \tr(\mX^{-1}) \cdot \frac{ \eps^2 b^2 }{2\tr(\mX^{-1})} = \frac{\eps^2 b^2}{2}.
\]
Thus, for $b \geq \Omega\left(\frac{d}{\eps}\right)$, 
it follows that $b \geq \Omega\left(\frac{1}{\eps} \Big(d + \sqrt{ \eps^2 b^2/2}\Big) \right) = \Omega\left(\frac{1}{\eps} \left( d + \sqrt{\tr(\mX) \tr(\mX^{-1})} \right) \right)$,
and so Theorem~\ref{t:A-combin} implies that Fedorov's exchange method will find a $(1+O(\eps))$-approximate solution in polynomial time. 
\end{proof}

\subsection{Combinatorial Local Search Algorithm for E-Design} \label{ss:comb-E}

Unlike D-design and A-design, there are simple examples (see Section~\ref{ss:examples}) showing that Fedorov's exchange method does not work for E-design, even if there is a well-conditioned optimal solution.  

Instead, we prove that the rounding algorithm by Allen-Zhu, Li, Singh and Wang~\cite{AZLSW20} for E-design can be used as a combinatorial local search algorithm as well.
The only difference is that the rounding algorithm in~\cite{AZLSW20} will first compute an optimal fractional solution $\vx$ to the convex programming relaxation and then perform a linear transformation so that $\sum_{i=1}^n \vx(i) \cdot \vv_i \vv_i^\top = \mI$, before applying the following combinatorial algorithm.
Our analysis will show that the combinatorial algorithm works well as long as there is an approximately optimal fractional solution with good condition number,
so this tells us that the only essential use of an optimal fractional solution in the rounding algorithm is for preconditioning.

The following algorithm assumes the knowledge of the objective value $\lambda^*$ of the targeted fractional solution.  
We will guess this value in the proof of Theorem~\ref{t:E-combin}.

\begin{framed}
{\bf Combinatorial Local Search Algorithm for E-Optimal Design}

Input: $n$ vectors $\vv_1, ..., \vv_n \in \R^d$, a budget $b \geq d$,
an accuracy parameter $\eps \in (0,1)$, 
and a targeted objective value $\lambda^*$.
    \begin{enumerate}
    
    \item Initialization: Let $S_0 \subseteq [n]$ be an arbitrary set with $|S_0| = b$.
          Set $t \leftarrow 1$ and $\alpha \gets \frac{\sqrt{d}}{\eps \lambda^*}$.
    
    \item {\bf Repeat}
    \begin{enumerate}%[label*=\arabic*.]
        \item Let $\mZ_t := \sum_{i \in S_{t-1}} \vv_i \vv_i^\top$. Compute $\mA_t \leftarrow \left(\alpha \mZ_t- l_t \mI\right)^{-2}$ where $l_t \in \R$ is the unique scalar such that $\mA_t \succ 0$ and $\tr(\mA_t) = 1$.

        \item Let $S'_{t-1} := \{i \in S_{t-1}: 2\alpha \inner{\vv_i \vv_i^\top}{ \mA_t^{1/2}} < 1 \}$.
        
        \item Find $i_t \in S'_{t-1}$ and $j_t \in [n] \setminus S_{t-1}$ such that
        \[
            (i_t, j_t) = \underset{(i,j):~ i \in S'_{t-1},~ j \in [n] \setminus S_{t-1}}{\arg\max} \Phi(\mA_t, i, j) 
            := \frac{\inner{\vv_j \vv_j^\top}{\mA_t}}{1+2\alpha \inner{\vv_j \vv_j^\top}{\mA^{\frac12}_t}} - \frac{\inner{\vv_i \vv_i^\top}{\mA_t}}{1 - 2\alpha \inner{\vv_i \vv_i^\top}{\mA_t^{\frac12}}}.
        \]
        
        \item Set $S_t \leftarrow S_{t-1} \cup \{ j_t\} \setminus \{i_t\}$ and $t \gets t+1$.
    \end{enumerate}
    {\bf Until} $\Phi(\mA_{t-1}, i_{t-1}, j_{t-1}) < \frac{\eps \lambda^*}{b}$ or $\lambda_{\min}(\mZ_{t-1}) \geq (1-2\eps) \lambda^*$.
    \item Return $S_{t-2}$ as the solution set. 
\end{enumerate}
\end{framed}

The regret minimization framework developed in~\cite{AZLO15,AZLSW20} bounds the minimum eigenvalue of the current solution using the potential functions $\Phi(\mA_t, i,j)$ that we are optimizing in each iteration.
Applying Corollary~\ref{cor:lambda-min-rank-two} with feedback matrices $\mF_0 = \mZ_1 \succcurlyeq 0$ and $\mF_t = \vv_{j_t} \vv_{j_t}^\top - \vv_{i_t} \vv_{i_t}^\top$ for $t \geq 1$, as long as $1 > 2\alpha \inner{\vv_{i_t} \vv_{i_t}^\top}{\mA_t^{\frac12}}$ for all $1 \leq t \leq \tau$, we have
\begin{align} \label{eq:lambda_min_lower}
\lambda_{\min}(\mZ_{\tau+1}) \geq \sum_{t=1}^\tau \Bigg( 
\underbrace{\frac{\inner{\vv_{j_t} \vv_{j_t}^\top}{\mA_t} }{1 + 2\alpha \inner{\vv_{j_t} \vv_{j_t}^\top}{\mA_t^{\frac12}} }}_{\text{gain}} 
- \underbrace{\frac{\inner{\vv_{i_t} \vv_{i_t}^\top }{\mA_t}}{1 - 2\alpha \inner{\vv_{i_t} \vv_{i_t}^\top}{\mA_t^{\frac12}}}}_{\text{loss}} \Bigg)
- \frac{2\sqrt{d}}{\alpha} 
= \sum_{t=1}^{\tau} \Phi(\mA_t, i_t, j_t) - \frac{2\sqrt{d}}{\alpha}.
\end{align}
Therefore, in order to lower bound the minimum eigenvalue of the solution, we upper bound the loss term and lower bound the gain term to quantify the progress in each iteration.

First, we show the existence of $i_t$ with small loss, with respect to a fractional solution $\vx$.

\begin{lemma}[Loss] \label{l:loss-E}
Let $S := S_{t-1}$, $S' := S'_{t-1}$, $\mZ := \mZ_t$ and $\mA := \mA_t$.
For any $\vx \in [0,1]^n$ with $\sum_{i=1}^n \vx(i) = q < b - 2\alpha\inner{\mZ}{\mA^{\frac12}}$, there exists $i \in S'$ with
\begin{align*}
    \frac{\inner{\vv_i \vv_i^\top}{\mA} }{1 - 2\alpha \inner{\vv_i \vv_i^\top}{\mA^{\frac12}} } \leq \frac{\inner{\mZ}{\mA} - \inner{\mX_S}{\mA}}{b - \vx(S) - 2\alpha \inner{\mZ}{\mA^{\frac12}}}.
\end{align*}
\end{lemma}

\begin{proof}
Consider the probability distribution of removing a vector $\vv_i$ with probability 
\[
\Pr[i_t = i] = \frac{(1-\vx(i)) (1 - 2\alpha \inner{\vv_i \vv_i^\top}{\mA^{\frac12}})}{\sum_{j \in S'} (1-\vx(j)) (1 - 2\alpha \inner{\vv_j \vv_j^\top}{\mA^{\frac12}}) } \quad \text{for all } i \in S'.
\]
We check that the probability distribution is well-defined.
Note that the numerator is non-negative as $1-2 \alpha \inner{\vv_i \vv_i^\top}{\mA^{\frac12}} > 0$ for each $i \in S'$. 
The denominator is
\begin{align*}
\sum_{j \in S'} \big(1-\vx(j)\big) \cdot \big(1 - 2 \alpha \inner{\vv_j \vv_j^\top}{\mA^{\frac12}}\big) 
& \geq \sum_{j \in S} \big(1-\vx(j)\big) \cdot \big(1-2\alpha \inner{\vv_j \vv_j^\top}{\mA^{\frac12}}\big) \nonumber 
\\
& \geq \sum_{j \in S} \big(1-\vx(j)\big) - 2\alpha \sum_{j \in S} \inner{\vv_j\vv_j^\top}{\mA^{\frac12}} \nonumber 
\\
& = b - \vx(S) - 2 \alpha \inner{\mZ}{\mA^{\frac12}} > 0
\end{align*}
where the first inequality holds as $1-2\alpha\inner{\vv_j \vv_j^\top}{\mA^{\frac12}} \leq 0$ for $j \in S \setminus S'$, 
the second inequality follows from $1-\vx(j) \leq 1$ for each $j \in [n]$, 
and the equality is by $|S| = b$,
and the strict inequality is by the assumption $b > q + 2 \alpha \inner{\mZ}{\mA^{\frac12}} \geq \vx(S) + 2 \alpha \inner{\mZ}{\mA^{\frac12}}$. 
Thus, $\Pr[i_t=i] \geq 0$ for each $i \in S'$, and clearly $\sum_{i \in S'} \Pr[i_t = i] = 1$.

The expected loss using this probability distribution is
\begin{align*}
\E\left[ \frac{\inner{\vv_{i_t} \vv_{i_t}^\top}{\mA} }{1 - 2\alpha \inner{\vv_{i_t} \vv_{i_t}^\top}{A^{\frac12}} } \right] 
& = \sum_{i \in S'} \frac{\big(1-\vx(i)\big) \cdot \big(1-2\alpha \inner{\vv_i \vv_i^\top}{\mA^{\frac12}}\big)}{\sum_{j \in S'} \big(1-\vx(j)\big) \cdot \big(1-2\alpha \inner{\vv_j \vv_j^\top}{\mA^{\frac12}}\big)} \cdot \frac{\inner{\vv_i \vv_i^\top}{\mA}}{1-2\alpha \inner{\vv_i \vv_i^\top}{\mA^{\frac12}}} 
\\
& = \frac{\sum_{i \in S'} \big(1-\vx(i)\big) \cdot \inner{\vv_i \vv_i^\top}{\mA}}{\sum_{j \in S'} \big(1-\vx(j)\big) \cdot \big(1-2\alpha \inner{\vv_j \vv_j^\top}{\mA^{\frac12}}\big)} 
\\
& \leq \frac{\inner{\mZ}{\mA} - \inner{\mX_S}{\mA}}{b - \vx(S) - 2 \alpha \inner{\mZ}{\mA^{\frac12}}},
\end{align*}
where the inequality is from the above inequality for the denominator and 
\begin{align*}
    \sum_{i \in S'} \big(1-\vx(i)\big) \cdot \inner{\vv_i \vv_i^\top}{\mA} 
\leq \sum_{i \in S} \big(1-\vx(i)\big) \cdot \inner{\vv_i \vv_i^\top}{\mA} = \inner{\mZ}{\mA} - \inner{\mX_S}{\mA}
\end{align*}
for the numerator. Therefore, there exists an $i \in S'$ with loss at most the expected value.
\end{proof}

Next, we show the existence of $j_t$ with large gain term, again with respect to a fractional solution.

%HERE
\begin{lemma}[Gain] \label{l:gain-E}
Let $S := S_{t-1}$ and $\mA := \mA_t$.
For any $\vx \in [0,1]^n$ with $\sum_{i=1}^n \vx(i) = q < b$ and $\vx(S) < q$, there exists $j \in [n] \setminus S$ with 
\begin{align*}
    \frac{\inner{\vv_j \vv_j^\top}{\mA} }{1 + 2\alpha \inner{\vv_j \vv_j^\top}{\mA^{\frac12}} } \geq \frac{\inner{\mX}{\mA} - \inner{\mX_S}{\mA}}{q - \vx(S) + 2\alpha \inner{\mX}{\mA^{\frac12}}}.
\end{align*}
\end{lemma}
\begin{proof}
Consider the probability distribution of adding a vector $\vv_j$ where each $j \in [n]\setminus S$ is sampled with probability
\[
\Pr[j_t = j] = \frac{\vx(j) \cdot \big(1+2\alpha \inner{\vv_j \vv_j^\top}{\mA^{\frac12}}\big)}{\sum_{i \in [n] \setminus S} ~\vx(i) \cdot \big(1+2\alpha \inner{\vv_i \vv_i^\top}{\mA^{\frac12}}\big)}
\quad \text{for each } j \in [n] \setminus S.
\] 
Note that the denominator is positive by the assumption $\vx(S) < q$ which implies that $\vx([n]\setminus S)>0$.

The expected gain with respect to this probability distribution is
\begin{align*}
    \E\left[ \frac{\inner{\vv_{j_t} \vv_{j_t}^\top}{\mA} }{1 + 2\alpha \inner{\vv_{j_t} \vv_{j_t}^\top}{\mA^{\frac12}} } \right] 
    & = \sum_{j \in [n] \setminus S} \frac{\vx(j) \cdot \big(1+2\alpha \inner{\vv_j \vv_j^\top}{\mA^{\frac12}}\big)}{\sum_{i \in [n] \setminus S} ~\vx(i) \big(1+2\alpha \inner{\vv_i \vv_i^\top}{\mA^{\frac12}}\big)} \cdot \frac{\inner{\vv_j \vv_j^\top}{\mA}}{1+2\alpha \inner{\vv_j \vv_j^\top}{\mA^{\frac12}}} \\
    & = \frac{\sum_{j \in [n] \setminus S} ~\vx(j) \cdot \inner{\vv_j \vv_j^\top}{\mA}}{\sum_{i \in [n] \setminus S} ~\vx(i) \cdot \big(1+2\alpha \inner{\vv_i \vv_i^\top}{\mA^{\frac12}}\big)} \\
    & = \frac{\inner{\mX}{\mA} - \inner{\mX_S}{\mA}}{q - \vx(S) + 2\alpha \sum_{i \in [n] \setminus S} ~\vx(i) \inner{\vv_i \vv_i^\top} {\mA^{\frac12}}} \\
    & \geq \frac{\inner{\mX}{\mA} - \inner{\mX_S}{\mA}}{q - \vx(S) + 2\alpha \inner{\mX} {\mA^{\frac12}}}.
\end{align*}
where the third equality is by $\sum_{i=1}^n \vx(i) = q$ 
and the last inequality holds as $\sum_{i \in [n] \setminus S} \vx(i) \cdot \vv_i \vv_i^\top \preccurlyeq \mX$. 
Therefore, there exist $j \in [n] \setminus S$ with gain at least the expected value.
\end{proof}

The following is the main technical result for E-design, which lower bounds the improvement of the potential function in each iteration.
Note that the result depends on the condition number of the fractional solution.

\begin{proposition}[Progress] \label{prop:progress-E}
Let $\vx \in [0,1]^n$ be a fractional solution with $\sum_{i=1}^n \vx(i) = q$. 
Let $\mZ_t = \sum_{i \in S_{t-1}} \vv_i \vv_i^\top$ be the current solution in the $t$-th iteration.
For any $0 < \eps < \frac12$, if
\[
\alpha = \frac{\sqrt{d}}{\eps \cdot \lambda_{\min}(\mX)}, 
\qquad \lambda_{\min}(\mZ_t) \leq (1-2\eps) \cdot \lambda_{\min}(\mX), 
\qquad \text{and} \qquad b \geq q + 2\Big(d+\frac{d}{\eps}\Big) + \frac{2d}{\eps} \sqrt{\frac{\lambda_{\rm{avg}}(\mX)}{\lambda_{\min}(\mX)}}
\]
where $\lambda_{\rm avg}(\mX) = \frac{\tr(\mX)}{d}$ is the average eigenvalue of $\mX$,
then the value of the potential function is
\[
\Phi(\mA_t, i_t, j_t) = \frac{\inner{\vv_{j_t} \vv_{j_t}^\top}{\mA_t} }{1 + 2\alpha \inner{\vv_{j_t} \vv_{j_t}^\top}{\mA_t^{\frac12}} }
- \frac{\inner{\vv_{i_t} \vv_{i_t}^\top }{\mA_t}}{1 - 2\alpha \inner{\vv_{i_t} \vv_{i_t}^\top}{\mA_t^{\frac12}}} \geq \frac{\eps}{b} \cdot \lambda_{\min}(\mX).
\]
\end{proposition}

\begin{proof}
Let $S := S_{t-1}$ be the current solution set at time $t$, $\mA = \mA_t$ and $\mZ = \mZ_t$. 
Note that $\vx(S) < q$, as otherwise $\lambda_{\min}(\mZ) \geq \lambda_{\min}(\mX)$ and the assumption does not hold.
Hence, we can apply Lemma~\ref{l:gain-E} to show the existence of $j_t \in [n] \backslash S$ with gain
\begin{align*}
\frac{\inner{\vv_{j_t} \vv_{j_t}^\top}{\mA} }{1 + 2\alpha \inner{\vv_{j_t} \vv_{j_t}^\top}{\mA^{\frac12}} } 
& \geq \frac{\inner{\mX}{\mA} - \inner{\mX_S}{\mA}}{q - \vx(S) + 2\alpha \inner{\mX}{\mA^{\frac12}}}
\geq \frac{\inner{\mX}{\mA} - \inner{\mX_S}{\mA}}{q - \vx(S) + 2\alpha \sqrt{\tr(\mX) \cdot \inner{\mX}{\mA}}}
\\
& \geq \frac{\lambda_{\min}(\mX) - \inner{\mX_S}{\mA}}{q - \vx(S) + 2\alpha \sqrt{\tr(\mX) \cdot \lambda_{\min}(\mX)}}
=
\frac{\lambda_{\min}(\mX) - \inner{\mX_S}{\mA}}{q - \vx(S) + \frac{2d}{\eps} \sqrt{ \frac{\lambda_{\rm avg}(\mX)}{\lambda_{\min}(\mX)}}}, 
\end{align*}
where the second inequality is by \eqref{eq:inner_upper} in Lemma~\ref{l:trace},
the third inequality is by Claim~\ref{cl:func} and the fact that $\inner{\mX}{\mA} \geq \inner{\lambda_{\min}(\mX) \cdot \mI}{\mA} \geq \lambda_{\min}(\mX)$ as $\tr(\mA) = 1$,
and the last equality is by the choice $\alpha = \frac{\sqrt{d}}{\eps \lambda_{\min}(\mX)}$ and the definition of $\lambda_{\rm avg}(\mX)$.

For the loss term, we need to check the condition that $b > q + 2\alpha \inner{\mZ_t}{\mA^{\frac12}}$ before applying Lemma~\ref{l:loss-E}. 
It follows from Lemma~\ref{l:cospectral} and the assumptions of $\alpha$, $\lambda_{\min}(\mZ)$, and $b$ that
\begin{align*} 
\alpha \inner{\mZ}{\mA^{\frac12}} 
\leq d + \alpha \sqrt{d} \cdot \lambda_{\min}(\mZ) 
\leq d + \frac{d}{\eps}
\quad \implies \quad
b > q + 2\alpha \cdot \inner{\mZ}{\mA^{\frac{1}{2}}}.
\end{align*}
Hence, Lemma~\ref{l:loss-E} implies the existence of an $i_t \in S$ with loss
\[
\frac{\inner{\vv_{i_t} \vv_{i_t}^\top }{\mA}}{1 - 2\alpha \inner{\vv_{i_t} \vv_{i_t}^\top}{\mA^{\frac12}}} \!\leq\! \frac{\inner{\mZ_t}{\mA} - \inner{\mX_S}{\mA}}{b - \vx(S) - 2\alpha \inner{\mZ}{\mA^{\frac12}}}
\leq \frac{\lambda_{\min}(\mZ) + \frac{\sqrt{d}}{\alpha} - \inner{\mX_S}{\mA}}{b-\vx(S)-2\left(d + \frac{d}{\eps}\right)}
\leq \frac{(1-\eps)\lambda_{\min}(\mX) - \inner{\mX_S}{\mA}}{q - \vx(S) + \frac{2d}{\eps} \sqrt{\frac{\lambda_{\rm{avg}}(\mX)}{\lambda_{\min}(\mX)}} },
\]
where the second inequality is by Lemma~\ref{l:cospectral} and the inequality above about $\alpha\inner{\mZ}{\mA^{\frac12}}$,
and the last inequality is by our assumptions about $\alpha$, $\lambda_{\min}(\mZ)$ and $b$.

Therefore, we conclude that the progress in each iteration is
\[
\Phi(\mA, i_t, j_t) = 
\frac{\inner{\vv_{j_t} \vv_{j_t}^\top}{\mA} }{1 + 2\alpha \inner{\vv_{j_t} \vv_{j_t}^\top}{\mA^{\frac12}} }
- \frac{\inner{\vv_{i_t} \vv_{i_t}^\top }{\mA}}{1 - 2\alpha \inner{\vv_{i_t} \vv_{i_t}^\top}{\mA\frac12}} 
\geq \frac{\eps \cdot \lambda_{\min}(\mX)}{q - \vx(S) + \frac{2d}{\eps} \sqrt{ \frac{\lambda_{\rm avg}(\mX)}{\lambda_{\min}(\mX)}}}
\geq \frac{\eps}{b} \cdot \lambda_{\min}(\mX),
\]
where the last inequality follows from the assumption about $b$.
\end{proof}

%HERE

By guessing the targeted objective value,
the main result in this subsection follows from Proposition~\ref{prop:progress-E} by a simple scaling argument.

\CombE*

\begin{proof}
We consider the following scaled-down version $\vy,\mY$ of the fractional solution $\vx,\mX$.
Let
\[
q = b - 2\Big( d+\frac{d}{\eps} \Big) - \frac{2d}{\eps} \sqrt{ \frac{\lambda_{\rm avg}(\mX)}{\lambda_{\min}(\mX)}},
\qquad
\vy := \frac{q}{b} \cdot \vx,
\qquad
\mY := \sum_{i=1}^n \vy(i) \cdot \vv_i \vv_i^\top = \frac{q}{b} \cdot \mX.
\]
Note that $\lambda_l(Y) = \frac{q}{b} \cdot \lambda_l(X)$ for each $1 \leq l \leq d$, and this implies that $b = q + 2\Big(d + \frac{d}{\eps} \Big) + \frac{2d}{\eps} \sqrt{ \frac{\lambda_{\rm avg}(\mY)}{\lambda_{\min}(\mY)}}$.

Suppose the combinatorial local search algorithm is run with the accuracy parameter $\eps$ and $\lambda^* := \lambda_{\min}(\mY)$
and terminates at the $\tau$-th iteration. 
If $\Phi(\mA_{\tau},i_\tau,j_\tau) < \frac{\eps}{b} \cdot \lambda_{\min}(\mY)$,
then we can apply Proposition~\ref{prop:progress-E} on $\vy$ to conclude that 
\[
\lambda_{\min}(\mZ_{\tau}) > (1 - 2\eps) \cdot \lambda_{\min}(\mY) = \frac{(1-2\eps) q}{b} \cdot \lambda_{\min}(\mX) \geq \big(1-O(\eps)\big) \cdot \lambda_{\min}(\mX),
\]
where the last inequality is by the assumption that $b = \Omega\Big( \frac{d}{\eps^2} \sqrt{\frac{\lambda_{\rm avg}(\mX)}{\lambda_{\min}(\mX)}} \Big)$.
This proves the approximate guarantee of the returned solution if the algorithm is run with $\lambda^* = \lambda_{\min}(\mY)$.

Our final algorithm runs the local search algorithm on different values of $\lambda^*$.
Initially, we start from an upper bound on $\lambda_{\min}(\mX)$ by setting $\lambda^* = \lambda_{\min}\left( \sum_{i=1}^n \vv_i \vv_i^\top \right)$.
Then it runs the local search algorithm with targeted objective value $\lambda^*$.
If the returned solution $\mZ$ satisfies $\lambda_{\min}(\mZ) \geq \big(1-O(\eps)\big) \cdot \lambda^*$ then it stops and returns $\mZ$ as our final solution;
otherwise, we set $\lambda^* \gets (1-\eps) \cdot \lambda^*$ and repeat until the first time that the local search algorithm finds a solution with $\lambda_{\min}(\mZ) \geq \big(1-O(\eps)\big) \cdot \lambda^*$.
For correctness, it is enough to show that the algorithm will stop when $(1-\eps) \cdot \lambda_{\min}(\mX) \leq \lambda^* \leq \lambda_{\min}(\mX)$.
This follows by applying the argument in the previous paragraphs on $\mX' := \frac{\lambda^*}{\lambda_{\min}(\mX)} \cdot \mX$, 
so that $\lambda_{\min}(\mX') = \lambda^*$ and then the returned solution $\mZ'$ will satisfy $\lambda_{\min}(\mZ') \geq \big(1-O(\eps)\big) \cdot \lambda_{\min}(\mX') \geq \big(1-O(\eps)\big) \cdot \lambda_{\min}(\mX)$.  

Finally, we bound the time complexity of the algorithm.
Note that $\frac{b}{n} \sum_{i=1} \vv_i \vv_i^\top$ is a feasible solution with objective value $\frac{b}{n} \cdot \lambda_{\min} \left( \sum_{i=1}^n \vv_i \vv_i^\top \right)$.
This implies that the number of executions of the local search algorithm is at most $O\left( \frac1\eps \log \frac{n}{b} \right)$.
In each execution with a fixed $\lambda^*$,
if the algorithm runs for $\tau \geq \frac{b}{\eps}$ iterations, 
the termination condition together with~\eqref{eq:lambda_min_lower} imply that
\[
\lambda_{\min}(\mZ_{\tau+1}) 
\geq \sum_{t=1}^{\tau} \Phi(\mA_t, i_t, j_t) - \frac{2\sqrt{d}}{\alpha} 
\geq \tau \Big( \frac{\eps \lambda^*}{b} \Big) - 2\eps \lambda^*
> (1-2\eps) \lambda^*,
\]
and so it would stop.  
Thus, the total number of iterations is at most $O\left( \frac{b}{\eps^2} \log \frac{n}{b} \right)$.
Each iteration can be implemented in polynomial time as shown in~\cite{AZLSW20}.
\end{proof}

The following is a corollary in the short vector setting.

\begin{corollary}
Let $\vx \in [0,1]^n$ be a fractional solution to the E-design problem with budget $b$. 
For any $0 < \eps < 1$, 
if $\norm{\vv_i}^2 \leq \eps^2 \cdot \lambda_{\min}(\mX)$ for $1 \leq i \leq n$ and $b \geq \Omega\left(\frac{d}{\eps^2}\right)$,
then the combinatorial local search algorithm for E-design returns a solution with at most $b$ vectors and objective value at least $\big(1-O(\eps)\big) \cdot \lambda_{\min}(\mX)$ in polynomial time.
\end{corollary}
\begin{proof}
It follows from the assumption $\norm{\vv_i}^2 \leq \eps^2 \cdot \lambda_{\min}(\mX)$ that 
\[
\lambda_{\rm avg}(\mX) 
= \frac{\tr(\mX)}{d} 
\leq \frac{b \eps^2 \cdot \lambda_{\min}(\mX)}{d}
\qquad \implies \qquad
\frac{2d}{\eps^2} \sqrt{ \frac{\lambda_{\rm avg}(\mX)}{\lambda_{\min}(\mX)}} 
\leq \frac{2\sqrt{bd}}{\eps}.
\]
Thus, for $b \geq \Omega(\frac{d}{\eps^2})$, 
it follows that $b \geq \Omega(\frac{\sqrt{bd}}{\eps}) \geq \Omega\Big(\frac{d}{\eps^2}\sqrt{\frac{\lambda_{\rm avg}(\mX)}{\lambda_{\min}(\mX)}}\Big)$,
and so Theorem~\ref{t:E-combin} implies that the combinatorial local search algorithm will find a $\big(1-O(\eps)\big)$-approximate solution in polynomial time.
\end{proof}

\subsubsection{Maximizing Algebraic Connectivity} \label{ss:comb-app-E}

In this problem, we are given a graph $G=(V,E)$ with Laplacian matrix $\mL_G = \sum_{e \in E} \vb_e \vb_e^\top$, 
and the goal is to find a subgraph $H$ with at most $b$ edges to maximize 
$\lambda_2(\mL_H)$. 
This problem is known as maximizing algebraic connectivity in the literature.
It is a special case of E-design and Theorem~\ref{t:E-combin} bounds the performance guarantee of a simple combinatorial local search algorithm.

\CombLambda*

\begin{proof}
Note that this is an E-design problem by using the vectors $\vb_e$ projecting onto the $(n-1)$-dimensional subspace orthogonal to the all-one vector.
Let $H^*$ be an optimal subgraph with $b$ edges.
Note that $\lambda_{\rm avg}(\mL_{H^*}) = \frac{\tr(\mL_{H^*})}{n} \leq \frac{2b}{n}$, 
and so
\[
b \geq \Omega\bigg( \frac{n}{\eps^4  \lambda_2(\mL_{H^*})} \bigg) 
\implies
b \geq \Omega\bigg( \frac{n}{\eps^2} \sqrt{\frac{2b}{n\lambda_2(\mL_{H^*})}} \bigg)
\geq 
\Omega\bigg( \frac{n}{\eps^2} \sqrt{\frac{\lambda_{\rm avg}(\mL_{H^*})}{\lambda_2(\mL_{H^*})}} \bigg).
\]
Therefore, by Theorem~\ref{t:E-combin}, the combinatorial local search algorithm for E-design returns a subgraph $H$ with $\lambda_2(\mL_H) \geq \big(1-O(\eps)\big) \cdot \lambda_2(\mL_{H^*})$ in polynomial time whenever $b \geq \Omega\Big(\frac{n}{\eps^4 \lambda_2(\mL_{H^*})} \Big)$.
\end{proof}

\subsubsection{Bad Examples for Local Search Algorithms} \label{ss:examples}

We first present a simple example showing that Fedorov's exchange method does not work with the E-design objective function, even if there is a well-conditioned optimal solution. 
The reason is simply that the E-design objective function is not smooth and sometimes it is impossible to improve it by an exchange operation.

\begin{example} 
Suppose the input vectors $\vv_1, ..., \vv_n$ are in $\R^d$ for some $d \geq 3$. 
Suppose that we have an initial solution set $S_0 \subseteq [n]$ such that $\mZ_1 = \sum_{i \in S_0} \vv_i \vv_i^\top = \mI$. 
For any $i_1 \in S_0$ and $j_1 \in [n] \backslash S_0$,
note that $\lambda_{\min}(\mZ_1 - \vv_{i_1} \vv_{i_1}^\top + \vv_{j_1} \vv_{j_1}^\top) \leq 1$. 
Therefore, Fedorov's method fails to improve the objective value even if there is a well-conditioned optimal solution say $N\ve_1, \ldots, N\ve_d$ for a large $N$.
\end{example}

Then, we present an example where all exchanges strictly decrease the minimum eigenvalue, even though the current solution is far away from the well-conditioned optimal solution.

\begin{example} 
Let $N \geq 0$ be some large scalar. 
The input contains exactly $\frac{b}{2}$ copies of each $\vv_1, \vv_2, \vw_1, \vw_2 \in \R^2$ defined as follows:
\begin{align*}
&\vv_1 \vv_1^\top = \begin{pmatrix} 1 & 0 \\ 0 & 0 \end{pmatrix}, & &
\vv_2 \vv_2^\top = \begin{pmatrix} 0 & 0 \\ 0 & 1  \end{pmatrix}, \\
&\vw_1 \vw_1^\top = \frac{N}{2} \begin{pmatrix} 1 & 1 \\ 1 & 1 \end{pmatrix}, & &
\vw_2 \vw_2^\top = \frac{N}{2} \begin{pmatrix} 1 & -1 \\ -1 & 1 \end{pmatrix}.
\end{align*}
The optimal solution $\mZ^*$ contains $\frac{b}{2}$ copies of $\vw_1 \vw_1^\top$ and $\vw_2 \vw_2^\top$.
Suppose the algorithm starts with the solution $\mZ_1$ containing $\frac{b}{2}$ copies of $\vv_1 \vv_1^\top$ and $\vv_2 \vv_2^\top$ such that
\begin{align*}
\mZ^* = \begin{pmatrix} \frac{b N}{2} & \\ & \frac{b N}{2} \end{pmatrix} \text{~\rm with~} \lambda_{\min}(\mZ^*) = \frac{b N}{2}, & & \mZ_1 = \begin{pmatrix} \frac{b}{2} & \\ & \frac{b}{2} \end{pmatrix}  \text{~\rm with~} \lambda_{\min}(\mZ_1) = \frac{b}{2}.
\end{align*}
Without loss of generality, we assume the exchange step removes $\vv_1$ and adds $\vw_1$. 
After the exchange, the solution is
\[
\mZ_2 = \begin{pmatrix} \frac{b+N}{2} - 1& \frac{N}{2} \\ \frac{N}{2} & \frac{b+N}{2} \end{pmatrix}.
\]
We can verify that the minimum eigenvalues of $\mZ_2$ is $\frac{b-1 + N - \sqrt{N^2 - 1}}{2}$, which tends to $\frac{b-1}{2}$ when $N \rightarrow \infty$. 
Since all other exchanges are symmetric, we conclude that all exchanges will decrease the objective value by $\frac12$, and thus Fedorov's exchange method fails.
\end{example}

Finally, we adopt an example by Madan, Singh, Tantipongpipat and Xie~\cite{MSTX19} to show that even we use a smooth objective function from the regret minimization framework, the combinatorial local search algorithm may return bad solution when there are no well-conditioned optimal solutions.

\begin{example}
Let $N \geq 0$ be some large scalar. 
The input contains $M \gg b \geq 3$ copies of each $\vv_1, \vv_2, \vw_1, \vw_2 \in \R^2$ defined as follows:
\begin{align*}
& \vv_1 \vv_1^\top = \begin{pmatrix} N^2 & 1 \\ 1 & \frac{1}{N^2} \end{pmatrix}, \qquad
\vv_2 \vv_2^\top = \begin{pmatrix} N^2 & -1 \\ -1 & \frac{1}{N^2} \end{pmatrix}, \\
& \vw_1 \vw_1^\top = \frac1b \cdot \begin{pmatrix} N^8 & N^4 \\ N^4 & 1 \end{pmatrix}, \qquad
\vw_2 \vw_2^\top = \frac1b \cdot \begin{pmatrix} N^8 & -N^4 \\ -N^4 & 1 \end{pmatrix}.
\end{align*}
\end{example}

\begin{lemma}
The combinatorial local search algorithm proposed in this subsection may return a solution with an unbounded approximation ratio.
\end{lemma}
\begin{proof}
Note that $\frac{b}{2}$ copies of $\vw_1 \vw_1^\top$ and $\frac{b}{2}$ copies of $\vw_2 \vw_2^\top$ form an optimal solution $\mZ^*$ with budget $b$ such that
\[
\mZ^* = \begin{pmatrix} N^8 & \\ & 1 \end{pmatrix} \qquad \text{\rm and} \qquad \lambda_{\min}(\mZ^*) = 1.
\]
So our algorithm will choose $\alpha = \frac{\sqrt{d}}{\eps \lambda_{\min}(\mZ^*)} = \frac{\sqrt{d}}{\eps} = \frac{\sqrt{2}}{\eps}$.

Consider an initial solution $\mZ$ containing $\frac{b}{2}$ copies of $\vv_1 \vv_1^\top$ and $\frac{b}{2}$ copies of $\vv_2 \vv_2^\top$ such that
\[
\mZ = \begin{pmatrix} b N^2 & \\ & \frac{b}{N^2} \end{pmatrix} \qquad \text{\rm and} \qquad \lambda_{\min}(\mZ) = \frac{b}{N^2}.
\]
The approximation ratio between $\mZ$ and $\mZ^*$ is $\frac{N^2}{b}$, which is unbounded for fixed $b$ when $N \rightarrow \infty$.

With $\mZ$ as the current solution, the action matrix $\mA$ is
\[
\mA = (\alpha \mZ - l \mI)^{-2} = \begin{pmatrix} \frac{\sqrt{2} b N^2}{\eps} - l & \\ &  \frac{\sqrt{2} b}{\eps N^2} - l \end{pmatrix}^{-2} \approx \begin{pmatrix} \frac{\eps^2}{2 b^2 N^4} & \\ & 1 \end{pmatrix},
\]
where the last approximate equality holds when $N \rightarrow \infty$ as $\tr(\mA) = 1$. 

The loss of removing vector $\vv_1$ (removing $\vv_2$ is similar) from the current solution is
\begin{align*}
    \frac{\inner{\vv_1 \vv_1^\top}{\mA}}{1 - 2\alpha \inner{\vv_1 \vv_1^\top}{\mA^{\frac12}}} \approx \frac{\frac{\eps^2}{2 b^2 N^2} + \frac{1}{N^2}}{1 - \frac{2\sqrt{2}}{\eps} \big(\frac{\eps}{\sqrt{2} b} + \frac{1}{N^2}\big)} \geq \frac{1}{N^2},
\end{align*}
where we used $b \geq 3$ and $N$ is large for the last inequality.

The gain of adding vector $\vv_2$ is strictly less than the loss of removing $\vv_2$
\[
\frac{\inner{\vv_2 \vv_2^\top}{\mA}}{1 - 2\alpha \inner{\vv_2 \vv_2^\top}{\mA^{\frac12}}} < \frac{\inner{\vv_1 \vv_1^\top}{\mA}}{1 - 2\alpha \inner{\vv_1 \vv_1^\top}{\mA^{\frac12}}},
\]
as $\inner{\vv_2 \vv_2^\top}{\mA} = \inner{\vv_1 \vv_1}{\mA}$ and $\inner{\vv_2 \vv_2^\top}{\mA^{\frac12}} = \inner{\vv_1 \vv_1}{\mA^{\frac12}}$.
Also, the gain of adding vector $\vw_1$ (adding $\vw_2$ is similar) to the current solution is
\begin{align*}
    \frac{\inner{\vw_1 \vw_1^\top}{A}}{1 + 2\alpha \inner{\vw_1 \vw_1^\top}{\mA^{\frac12}}} \approx \frac{\frac{\eps^2 N^4}{2 b^3} + \frac1b}{1 + \frac{2\sqrt{2}}{\eps} (\frac{\eps N^6}{\sqrt{2} b^2} + \frac1b )} \leq \frac{\eps^2 }{4 b N^2} + \frac{b}{2 N^6}.
\end{align*}
For appropriately fixed $b$ and $\eps$, this gain is always less than the loss when $N \rightarrow \infty$. 
Therefore, the combinatorial local search algorithm will stop and return the initial solution $\mZ$.
\end{proof}

\section{D/A-Optimal Design with Knapsack Constraints} \label{s:knapsack}

In this section, we propose the following randomized exchange algorithm to solve the D/A-optimal design problems with knapsack constraints.

\begin{framed}
{\bf Randomized Exchange Algorithm}

    Input: $n$ vectors $\vu_1, ..., \vu_n \in \R^d$, an accuracy parameter $\eps \in (0,1)$, and $m$ knapsack constraints $\vc_j \in \R^n_+$ with budgets $b_j \geq \frac{d \norm{\vc_j}_{\infty}}{\eps}$ for all $j \in [m]$.
    \begin{enumerate}
    
    \item Solve the convex programming relaxation~\eqref{eq:convex} for D-design or A-design and obtain an optimal solution $\vx \in [0,1]^n$ with at most $d^2 + m$ fractional entries, i.e.~$|\{i \in [n] \mid 0 < \vx(i) < 1\}| \leq d^2 + m$.
    Let $\mX = \sum_{i=1}^n \vx(i) \cdot \vu_i \vu_i^\top$. 

    \item Preprocessing: 
	Let $\vv_i \leftarrow \mX^{-1/2} \vu_i$ for all $i \in [n]$, 
	so that $\sum_{i=1}^n \vx(i) \cdot \vv_i \vv_i^\top = \mI_d$.
    
    \item Initialization: $t \gets 1$, $S_0 \gets \emptyset$, $\alpha \gets 8\sqrt{d}$, and $k \gets 16d + d^2 + m$.
    
    \item Add $i$ into $S_0$ independently with probability $\vx(i)$ for each $i \in [n]$. Let $\mZ_1 \gets \sum_{i \in S_0} \vv_i \vv_i^\top$.
    
    \item While the termination condition is not satisfied and 
	$t = O\Big( \frac{k}{\eps} \Big)$ do the following,
    where the termination conditions for D-design and A-design are respectively
        \[\det(\mZ_t)^{1/d} \geq 1-10\eps \quad {\rm and} \quad \inner{\mX^{-1}}{\mZ^{-1}_t} \leq (1+\eps) \tr(\mX^{-1}). \]
    \begin{enumerate}
        \item $S_t \gets$ Exchange($S_{t-1}$).
        \item Set $\mZ_{t+1} \gets \sum_{i \in S_t} \vv_i \vv_i^\top$ and $t \gets t+1$.
    \end{enumerate}
    
    \item Return $S_{t-1}$ as the solution. 
\end{enumerate}
\end{framed}

The exchange subroutine is described as follows.

\begin{framed}
{\bf Exchange Subroutine}

    \begin{enumerate}
        \item Compute the action matrix $\mA_t := (\alpha \mZ_t - l_t \mI)^{-2}$, where $\mZ_t = \sum_{i \in S_{t-1}} \vv_i \vv_i^\top$ and $l_t$ is the unique scalar such that $\mA_t \succ 0$ and $\tr(\mA_t) = 1$.
        \item Define $S'_t := \{ i \in S_{t-1} \mid 2\alpha \inner{\vv_i \vv_i^\top}{\mA_t^{1/2}} \leq \frac12 \}$.
        \item Sample $i_t \in S'_{t-1}$ from the following probability distribution 
        \begin{align*}
            & \Pr(i_t = i) =\frac{1-\vx(i)}{k} \cdot \left(1-2\alpha \inner{\vv_i \vv_i^\top}{\mA_t^{1/2}}\right), \text{ for } i \in S'_{t-1} \text{ and} \\
            & \Pr(i_t = \emptyset) = 1- \sum_{i \in S'_{t-1}} \frac{1-\vx(i)}{k} \cdot \left(1-2\alpha \inner{\vv_i \vv_i^\top}{\mA_t^{1/2}}\right).
        \end{align*}
        
        \item Sample $j_t \in [n] \backslash S_{t-1}$ from the following probability distribution
        \begin{align*}
            & \Pr(j_t = j) =\frac{\vx(j)}{k} \cdot \left(1+ 2\alpha \inner{\vv_j \vv_j^\top}{\mA_t^{1/2}}\right), \text{ for } j  \in [n] \backslash  S_{t-1} \text{ and} \\ 
            & \Pr(j_t = \emptyset) = 1- \sum_{j \in [n] \backslash S_{t-1}} \frac{\vx(j)}{k} \cdot \left(1+ 2\alpha \inner{\vv_j \vv_j^\top}{\mA_t^{1/2}}\right).
        \end{align*}
        
        \item Return $S_t := S_{t-1} \cup \{ j_t\} \backslash \{i_t\}$.
    \end{enumerate}
\end{framed}

\begin{remark} \label{r:similar}
The randomized exchange algorithm is almost the same as the iterative randomized rounding algorithm in~\cite{LZ20}.
There are only two differences.
One is that $\alpha \gets 8\sqrt{d}$ instead of $\alpha \gets \frac{\sqrt{d}}{\eps}$ in~\cite{LZ20}.
The other is that the termination condition for E-design, which is $\lambda_{\min}(\mZ_t) \geq 1-2\eps$, is replaced by the termination condition for D-design or the termination condition for A-design.

The parameter $\alpha$ is used to control the approximation guarantee of the algorithm for E-design.
If the termination condition is $\lambda_{\min}(\mZ_t) \geq \frac{3}{4}$, then 
it was proved in Theorem~\ref{t:iterative_rounding_terminate} (Theorem 3.8 of~\cite{LZ20}) that the algorithm will terminate successfully in $O(k)$ steps with high probability.
\end{remark}

\subsubsection*{Intuition and Proof Ideas}

Given the analysis of Fedorov's exchange method for D-design and A-design in Section~\ref{s:comb}, the most natural algorithm is to use the same distributions there for the rounding algorithm as well.
We use D-design to illustrate the difficulty of analyzing this natural algorithm and to motivate the modifications made in the randomized exchange algorithm.
By applying Lemma~\ref{l:Dlower} repeatedly, for any $\tau \geq 1$,
\begin{align*} %\label{eq:det_bound}
    \det(\mZ_{\tau+1}) & \geq \det(\mZ_1) \cdot \prod_{t=1}^\tau \left( 1 - \vv_{i_t}^\top \mZ_t^{-1} \vv_{i_t} \right) \left( 1 + \vv_{j_t}^\top \mZ_t^{-1} \vv_{j_t} \right).
\end{align*}
Using the distributions $\Pr(i_t = i) \propto 1-x_i$ and $\Pr(j_t=j) \propto x_j$ as in Section~\ref{ss:comb-D}, Lemma~\ref{l:gain-D} and Lemma~\ref{l:loss-D} shows that there exist $i_t \in S_{t-1}$ and $j_t \notin S_{t-1}$ such that setting $S_{t} \gets S_{t-1} - i_t + j_t$ will improve the D-design objective.
However, if we randomly sample $i_t$ and $j_t$ from these distributions, we cannot prove that the objective value is consistently improving with good probability.
For D-design, we are analyzing a product of random variables where each random variable could have a large variance, and existing martingale inequalities are not applicable to establish concentration of the product.

To bound the variance, one important observation is that when $\vx$ is an optimal fractional solution, it follows from the optimality condition of the convex programming relaxation that any vector $\vv_i$ with $\vx(i) \in (0,1)$ satisfies $\norm{\vv_i}_2^2 \leq \eps$.
The current algorithm is motivated by the observation that if we can also lower bound the minimum eigenvalue of $\mZ_t$, then we can upper bound $\vv^\top \mZ_t^{-1} \vv$ and this would allow us to establish concentration of the objective value.
So our idea is to use the same algorithm in~\cite{LZ20} for E-design to ensure that the minimum eigenvalue of $\mZ_t$ is at least $\Omega(1)$ as mentioned in Remark~\ref{r:similar}.
Surprisingly, we prove that sampling from the distributions for E-design can also improve the objective values for D-design and A-design, and this is particularly interesting for A-design where the minimum eigenvalue condition is needed to prove so.
Having these in place, we can use Freedman's martingale inequality to prove that the objective values for D-design and A-design will be improving consistently after the minimum eigenvalue of the current solution is at least $\Omega(1)$.

\subsubsection*{Proof Outline and Organization}

In the analysis of the randomized exchange algorithm, 
we conceptually divide the algorithm into two phases.
In the first phase, we show that the minimum eigenvalue of the current solution will reach $\frac34$ in $O(k)$ iterations with high probability.
In the second phase, we prove that the objective value for D/A-design will be $(1\pm \eps)$-approximation of the optimal in $O\left(\frac{k}{\eps}\right)$ iterations with high probability.
The following is an outline of the proof steps.
\begin{enumerate}
\item In Section~\ref{ss:well-defined}, we first prove that the randomized exchange algorithm is well-defined. In particular, we show that a fractional optimal solution to the convex relaxation~\eqref{eq:convex} with at most $O(d^2+m)$ fractional entries can be found in polynomial time, and the probability distributions in the exchange subroutine are well-defined for $k = O(d^2+m)$.
% This will be used to upper bound the number of iterations and the failure probability.
\item In Section~\ref{ss:minE}, we prove that the minimum eigenvalue will reach $\frac34$ in $O(k)$ iterations with high probability.
Furthermore, the minimum eigenvalue will be at least $\frac{1}{4}$ during the next $\Theta\left(\frac{k}{\eps}\right)$ iterations with good probability, for which we require the assumption that $\eps$ is not too small.
The proofs are based on the regret minimization framework and the iterative randomized rounding algorithm developed in~\cite{AZLSW20,LZ20}.
\item In Section~\ref{ss:D-knapsack} and Section~\ref{ss:A-knapsack}, we prove that the objective value of D-design and A-design will improve consistently with high probability.
These are the more technical parts of the proof.
We use the minimum eigenvalue condition in multiple places, both in the martingale concentration arguments for D/A-design and in the expected improvement of the A-design objective.
We also need the optimality conditions for the martingale concentration arguments.
\item In Section~\ref{ss:main-proof}, we prove the main approximation results including Theorem~\ref{t:DA-rounding} for experimental design, by combining the previous steps and using the concentration inequality for the knapsack constraints proved in~\cite{LZ20}.
As a corollary, we slightly improve the previous results of D/A-design with a single cardinality constraint in~\cite{SX18,NST19}.
We also prove Corollary~\ref{c:Reff-rounding} as an application of the main result.
\end{enumerate}

\subsection{Analysis of the Common Algorithm} \label{ss:common}

The algorithm is identical for D-design and A-design except the termination condition.
In this subsection, we will present the proofs of the common parts and the main results, and then present the specific proofs for D-design and A-design in Section~\ref{ss:D-knapsack} and Section~\ref{ss:A-knapsack} respectively.

\subsubsection{Sparse Optimal Solution and Probability Distributions in the Exchange Subroutine} \label{ss:well-defined}

In this subsection, we first show that we can find an optimal fractional solution to the convex programming relaxation~\eqref{eq:convex} with sparse support in polynomial time.
The sparsity of an optimal solution to the convex program~\eqref{eq:convex} was proved and used in~\cite{WYS17, MNST20} for experimental design problems.
The following lemma is proved using similar ideas.

\begin{lemma} \label{l:sparse}
Given any feasible fractional solution $\hat{\vx}$ to the convex relaxation~\eqref{eq:convex}, there exists another feasible fractional solution $\vx$ with $|\{i \in [n] \mid 0< \vx(i) < 1\}| \leq d^2 + m$ such that
\begin{align*}
& \det\left( \sum_{i=1}^n \vx(i) \cdot \vu_i \vu_i^\top \right) = \det\left(\sum_{i=1}^n \hat{\vx}(i) \cdot \vu_i \vu_i^\top \right)\quad \text{for D-Design, or} \\
& \tr\left( \bigg(\sum_{i=1}^n \vx(i) \cdot \vu_i \vu_i^\top \bigg)^{-1} \right) = \tr\left( \bigg( \sum_{i=1}^n \hat{\vx}(i) \cdot \vu_i \vu_i^\top \bigg)^{-1} \right) \quad \text{for A-Design}.
\end{align*}
Furthermore, the solution $\vx$ can be found in polynomial time.
\end{lemma}
\begin{proof}
Given the feasible fractional solution $\hat{\vx}$, 
we compute an extreme point solution $\vx$ to the following polytope, 
which can be done in polynomial time.
\[
\left\{ 
\begin{aligned} 
& \sum_{i=1}^n \vx(i) \cdot \vu_i \vu_i^\top = \sum_{i=1}^n \hat{\vx}(i) \cdot \vu_i \vu_i^\top, \\
& \inner{\vc_j}{\vx} \leq b_j, \quad \text{~for $1 \leq j \leq m$}, \\
& 0 \leq \vx(i) \leq 1, \quad \text{for $1 \leq i \leq n$}.
\end{aligned} \right.
\]
In the extreme point solution $\vx$, the number of variables is equal to the number of linearly independent tight constraints attained by $\vx$.
Clearly, the number of integral variables in $\vx$ is equal to the number of linear independent tight constraints in $0 \leq \vx(i) \leq 1 \text{~for $1 \leq i \leq n$}$ attained by $\vx$.
So, the number of fractional variables in $\vx$ is equal to the number of linear independent tight constraints in $\sum_{i=1}^n \vx(i) \cdot \vu_i \vu_i^\top = \sum_{i=1}^n \hat{\vx}(i) \cdot \vu_i \vu_i^\top$ and $\inner{\vc_j}{\vx} \leq b_j \text{~for $1 \leq j \leq m$}$ attained by $\vx$. 
As there are only $d^2 + m$ such constraints in the above linear program,
there are at most $d^2 + m$ fractional entries in $\vx$. 
Due to the first matrix equality constraint of the polytope, $\vx$ and $\hat{\vx}$ have the same objective value.
\end{proof}

Then, we make a simple observation of the randomized exchange algorithm,
that only vectors with fractional entries will be exchanged, 
as those vectors with $\vx(i)=1$ will always be in the solution and vectors with $\vx(i)=0$ will always not be in the solution.

\begin{observation} \label{obs:prob}
For any $t \geq 0$, it holds that $i \in S_t$ for all $i$ with $\vx(i) = 1$ and $j \in [n] \backslash S_t$ for all $j$ with $\vx(j) = 0$. 
This further implies that $\Pr(i_t = i) = 0$ for all $i$ with $\vx(i) \in \{0,1\}$ and $\Pr(j_t = j) = 0$ for all $j$ with $\vx(j) \in \{0,1\}$.
\end{observation}
\begin{proof}
The observation follows as all vectors with $\vx(i)=1$ are selected and all vectors with $\vx(j)=0$ are not selected initially. 
In each iteration, the probability distributions in the exchange subroutine guarantee that vectors with $\vx(i)=1$ have zero probability to be removed from the solution set, and vectors with $\vx(j)=0$ have zero probability to be added into the solution set. Therefore, the exchange subroutine of the algorithm would only exchange those vectors with fractional entries $\vx(i)$'s.
\end{proof}

Finally, we are ready to show that the probability distributions in the exchange subroutine are well-defined for $k = O(d^2+m)$,
which will be used to upper bound the number of iterations and the failure probability of the algorithm.

\begin{claim} \label{cl:validM}
The probability distributions at any $t$-th iteration of the randomized exchange algorithm are well-defined for $k=16d+d^2+m$.
\end{claim}
\begin{proof}
First, we verify that the probability distribution for sampling $i_t$ is well-defined. 
We need to show that $\Pr(i_t = i) \geq 0$ for $i \in S_{t-1}'$ and $\sum_{i \in S_{t-1}'} \Pr(i_t = i) \leq 1$.
Since $\mA_t \succ 0$ and $x_i \in [0,1]$ and $2\alpha \inner{\vv_i \vv_i^\top}{\mA_t^{\frac12}}\leq 1/2$ for $i \in S_{t-1}'$, it holds for $i \in S_{t-1}'$ that
\[
0 \leq \Pr(i_t = i) = \frac{1}{k} (1-\vx(i))(1-2\alpha \inner{\vv_i \vv_i^\top}{\mA_t^{\frac12}}) \leq \frac{1}{k}.
\]
Thus, $\sum_{i \in S_{t-1}'} \Pr(i_t = i) \leq \frac{1}{k} |\{i \in [n] \mid 0 < \vx(i) < 1\}| < 1$, where the first inequality follows by \Cref{obs:prob}, and the second inequality follows by the the choice of $k=16d + d^2 + m$ and \Cref{l:sparse}.

Next, we verify that the probability distribution for sampling $j_t$ is well-defined. 
It is clear that $\Pr(j_t = j) \geq 0$ as $\mA_t \succ 0$ and $\vx(j) \in [0,1]$.
Then, we consider
\begin{align*}
\sum_{j \in [n] \setminus S_{t-1}} \Pr(j_t = j)
= \frac{1}{k} \sum_{j \in [n] \setminus S_{t-1}} \vx(j) \cdot \left(1+2\alpha \inner{\vv_j \vv_j^\top}{\mA_t^{\frac12}}\right)  
\leq \frac{1}{k} \Bigg( \sum_{j \in [n] \backslash S_{t-1}} \vx(j) + 2\alpha \tr\Big(\mA_t^{\frac12}\Big)\Bigg),
\end{align*}
where the inequality is by $\sum_{j=1}^n \vx(j) \cdot \vv_j \vv_j^\top = \mI_d$.
Notice that $\sum_{j \in [n] \backslash S_{t-1}} \vx(j) \leq |\{i \in [n] \mid 0 < \vx(i) < 1\}| \leq d^2 + m$ by \Cref{obs:prob} and \Cref{l:sparse}. Thus,
\[
\sum_{j \in [n] \setminus S_{t-1}} \Pr(j_t = j) \leq \frac{1}{k} \left( d^2 + m  + 2\alpha \tr(\mA_t^{\frac12}) \right) \leq \frac1k (d^2 + m + 16d) \leq 1,
\]
where the second last inequality is by $\alpha = 8\sqrt{d}$ and $\tr\big( \mA_t^{\frac12} \big) \leq \sqrt{d}$ from Lemma~\ref{l:tr},
and the last inequality is by the choice of $k$.
\end{proof}

Combining \Cref{l:sparse} and \Cref{cl:validM}, we have shown that the randomized exchange algorithm is well-defined.

\subsubsection{Lower Bounding Minimum Eigenvalue} \label{ss:minE}

As discussed above, the minimum eigenvalue of $\mZ_t$ plays a key role in our analysis of the algorithm. 
We conceptually divide the execution of the randomized exchange algorithm into two phases.
In the first phase, we show that the minimum eigenvalue of the current solution will reach $\frac{3}{4}$ in $O(k)$ iterations with high probability.

\begin{proposition} \label{p:terminate1}
The probability that the randomized exchange algorithm has terminated successfully within $16k$ iterations or three exists $\tau_1 \leq 16k$ with $\lambda_{\min}(\mZ_{\tau_1}) \geq \frac{3}{4}$ is at least $1-\exp(-\Omega(\sqrt{d}))$. 
\end{proposition}
\begin{proof}
As noted in Remark~\ref{r:similar}, except for the termination condition, the randomized exchange algorithm is exactly the same as the algorithm in~\cite{LZ20} with $\alpha = 8\sqrt{d}$. 
So, the proposition follows from Theorem~\ref{t:iterative_rounding_terminate} with $\gamma = \frac18$ and $q = 2$.
\end{proof}

In the second phase, we prove that the minimum eigenvalue of $\mZ_t$ is at least $\frac{1}{4}$ in the next $\Theta\left(\frac{k}{\eps}\right)$ iterations with good probability.

\begin{proposition} \label{p:bounded_min}
Suppose $\lambda_{\min}(\mZ_{\tau_1}) \geq \frac34$ for some $\tau_1$.
In the randomized exchange algorithm, the probability that $\lambda_{\min}(\mZ_t) \geq \frac14$ for all $\tau_1 \leq t \leq \tau_1 + \frac{2k}{\eps}$ is at least $1 -\frac{4k^2}{\eps^2} \cdot e^{- \Omega(\sqrt{d})}$. 
\end{proposition}
\begin{proof}
Consider the bad event that there exists a time $t \in [\tau_1, \tau_1 + \frac{2k}{\eps}]$ with $\lambda_{\min}(\mZ_t) < \frac{1}{4}$.
As the initial solution $\mZ_{\tau_1}$ satisfies $\lambda_{\min}(\mZ_{\tau_1}) \geq \frac34$, 
there must exist a time period $[t_0, t_1] \subseteq [\tau_1, \tau_1 + \frac{2k}{\eps})$ such that $\lambda_{\min}(\mZ_{t_0}) \geq \frac34$, $\lambda_{\min}(\mZ_{t_1+1}) < \frac14$, and $\lambda_{\min}(\mZ_t) \in [\frac14, \frac34)$ for all $t \in [t_0+1, t_1]$.

We show that the decrease of the minimum eigenvalue from $t_0$ to $t_1$ implies that the sum of $\Delta_t$ defined in~\eqref{eq:Delta} has decreased significantly. 
Let $\mF_{t_0} = \mZ_{t_0}$ and $\mF_t = \vv_{j_t} \vv_{j_t}^\top - \vv_{i_t} \vv_{i_t}^\top$ for all $t \in [t_0+1, t_1]$. 
Note that the exchange subroutine ensures that $\alpha \inner{\vv_{i_t} \vv_{i_t}^\top}{\mA^{\frac12}_t} \leq \frac14$ for any $t$.  
So, it follows from Corollary~\ref{cor:lambda-min-rank-two} with $\alpha=8\sqrt{d}$ that
\begin{align*} %\label{eq:delta_upper}
\frac14 > \lambda_{\min}(\mZ_{t_1+1}) \geq \sum_{t=t_0+1}^{t_1} \Delta_t
- \frac{2\sqrt{d}}{\alpha} + \lambda_{\min}(\mZ_{t_0}) \geq \sum_{t=t_0+1}^{t_1} \Delta_t - \frac14 + \frac34 \quad \Longrightarrow \quad \sum_{t=t_0+1}^{t_1} \Delta_t < -\frac14.
\end{align*}
On the other hand, $\Delta_t$ is expected to be positive when $\lambda_{\min}(\mZ_t) < \frac{3}{4}$. 
The expectation bound in Lemma~\ref{l:spec-exp-whp-t} with $\tau'=t_0$, $\tau=t_1$, $\lambda \leq \frac{3}{4}$ and $\gamma = \frac18$ implies that
\begin{align*} %\label{eq:exp_lower}
\sum_{t=t_0+1}^{t_1} \E[\Delta_t | S_{t-1}]
\geq \frac{t_1 - t_0}{8k}.
\end{align*}
So, the sum of $\Delta_t$'s has a large deviation from the expectation, i.e.
\begin{align*}
\sum_{t=t_0+1}^{t_1} \Delta_t \leq \sum_{t=t_0+1}^{t_1} \E[\Delta_t | S_{t-1}] - \left( \frac14 + \frac{t_1 - t_0}{8k} \right).
\end{align*}
We can apply the concentration bound in Lemma~\ref{l:spec-exp-whp-t} with $\gamma = \frac18$, $\lambda \leq \frac34$ and $\eta = \frac14 + \frac{t_1 - t_0}{8k}$ to upper bound this probability by
\begin{eqnarray*}
& & \Pr \left[ \sum_{t=t_0+1}^{t_1} \Delta_t \leq \sum_{t=t_0+1}^{t_1} \E[\Delta_t \mid S_{t-1}] - \Big( \frac14 + \frac{t_1 - t_0}{8k} \Big) \right]
\\ 
& \leq & \exp\Bigg( -\frac{4\Big(\frac14 + \frac{t_1 - t_0}{8k}\Big)^2 k\sqrt{d}}{(t_1 - t_0)(1+\frac34+\frac18) + \Big(\frac14 + \frac{t_1 - t_0}{8k}\Big) k/3} \Bigg) ~\leq~ \exp\left( - \Omega(\sqrt{d}) \right),
\end{eqnarray*}
where the last inequality follows as the denominator is in the order of $\Theta(k + t_1 - t_0)$, and the numerator is in the order of $\Omega\Big(1 + \frac{t_1 - t_0}{k} + \frac{(t_1 - t_0)^2}{k^2}\Big) \cdot k \sqrt{d} = \Omega(k + t_1 - t_0) \cdot \sqrt{d}$.
The proposition follows by applying the union bound over the at most $\frac{4k^2}{\eps^2}$ possible pairs of $t_0$ and $t_1$ from time $\tau_1$ to $\tau_1+\frac{2k}{\eps}$.
\end{proof}

\subsubsection{Main Approximation Results} \label{ss:main-proof}

In this subsection, we prove the main approximation results for experimental design, including Theorem~\ref{t:DA-rounding}.
We will do so by first assuming the following theorem about the improvement of the objective value in the second phase, which will be proved in Section~\ref{ss:D-knapsack} for D-design and in Section~\ref{ss:A-knapsack} for A-design.

\begin{theorem} \label{t:termination}
Suppose that $\lambda_{\min}(\mZ_{\tau_1}) \geq \frac{3}{4}$ and $\lambda_{\min}(\mZ_t) \geq \frac{1}{4}$ for $t \geq \tau_1$.
For both D-design and A-design, if $b_j \geq \frac{d \norm{\vc_j}_{\infty}}{\eps}$ for all $j \in [m]$ for some $\eps \leq \frac{1}{100}$, 
then the probability that the randomized exchange algorithm has not terminated by time $\tau_1 + \frac{2k}{\eps}$ is at most $e^{-\Omega(\sqrt{d})}$.
\end{theorem}

First, we prove the following bicriteria approximation result for D/A-design with knapsack constraints,
by combining the previous steps and using the concentration inequality for the knapsack constraints proved in~\cite{LZ20}.

\begin{theorem} \label{t:main-knapsack}
Given $\eps \leq \frac{1}{100}$, 
if $b_j \geq \frac{d \norm{\vc_j}_{\infty}}{\eps}$ for all $j \in [m]$, 
then the randomized exchange algorithm returns a solution set $S$ within $16k + \frac{2k}{\eps}$ iterations such that
\[
\det\bigg(\sum_{i \in S} \vu_i \vu_i^\top\bigg)^{\frac{1}{d}} \geq (1-10\eps) \cdot \det\left( \mX \right)^{\frac{1}{d}}
\quad {\rm or} \quad
\tr\bigg(\bigg( \sum_{i \in S} \vu_i \vu_i^\top\bigg)^{-1}\bigg) \leq (1+\eps) \cdot \tr\left(\mX^{-1}\right)
\]
for D-design and A-design respectively with probability at least $1-O\Big(\frac{k^2}{\eps^2} \cdot e^{-\Omega\left(\sqrt{d}\right)} \Big)$, 
where $\mX$ is an optimal fractional solution.
Moreover, for each $j \in [m]$, the solution set $S$ satisfies
\[
\vc_j(S) \leq (1+\eps)b_j + 120 d \norm{\vc_j}_{\infty} \leq \big(1+O(\eps)\big)b_j
\]
with probability at least $1-e^{-\Omega(\eps d)}$. 
\end{theorem}

\begin{proof}
We start with defining some bad events for the randomized exchange algorithm.
\vspace{-2mm}
\begin{itemize}
\setlength\itemsep{-1pt}
    \item $B_1$: the algorithm has not terminated successfully within $16k$ iterations and $\tau_1 > 16k$ where $\tau_1$ is the first time such that $\lambda_{\min}(\mZ_{\tau_1}) \geq \frac34$.
    \item $B_2$: there exists some $\tau_1 \leq t \leq \tau_1 + \frac{2k}{\eps}$ such that $\lambda_{\min}(\mZ_t) < 1/4$.
    \item $B_3$: the termination condition for D/A-design is not satisfied for all $\tau_1 \leq t \leq \tau_1 + \frac{2k}{\eps}$.
\end{itemize}
\vspace{-2mm}
If none of the bad events happens, then either the algorithm has terminated successfully within $16k$ iterations or the termination condition for D/A-design will be satisfied at some time $t \leq \tau_1 + \frac{2k}{\eps} \leq 16k + \frac{2k}{\eps}$.
So, the probability that the randomized exchange algorithm has not satisfied the termination condition within $16k + \frac{2k}{\eps}$ iterations is upper bounded by
\begin{align*}
\Pr[B_1 \cup B_2 \cup B_3] 
& \leq \Pr[B_1] 
+ \Pr[B_2 \cap \neg B_1] 
+ \Pr[B_3 \cap \neg B_2 \cap \neg B_1]
\\
& \leq O\left(e^{-\Omega(\sqrt{d})} \right)
+ O\left( \frac{k^2}{\eps^2} \cdot e^{-\Omega(\sqrt{d})}  \right)
+ O\left( e^{-\Omega(\sqrt{d})}  \right)
\\
& \leq O\left( \frac{k^2}{\eps^2} \cdot e^{-\Omega(\sqrt{d})}  \right),
\end{align*}
where $\Pr[B_1]$ is bounded in Proposition~\ref{p:terminate1},
$\Pr[B_2 \cap \neg B_1]$ is bounded in Proposition~\ref{p:bounded_min},
and $\Pr[B_3 \cap \neg B_2 \cap \neg B_1]$ is bounded in Theorem~\ref{t:termination}.

For D-design, since $\vv_i = \mX^{-\frac{1}{2}}\vu_i$, 
the termination condition implies the approximation guarantee as 
\[
\det\bigg( \sum_{i \in S} \vv_i \vv_i^\top \bigg)^{\frac{1}{d}} > 1-10\eps 
\quad \Longrightarrow \quad 
\det\bigg( \sum_{i \in S} \vu_i \vu_i^\top \bigg)^{\frac{1}{d}} \geq  (1- 10\eps) \cdot \det(\mX)^{\frac{1}{d}}.
\]
For A-design, note that
\begin{equation} \label{eq:A-obj-transf} 
\begin{aligned}
\bigg\langle \mX^{-1}, \bigg( \sum_{i \in S} \vv_i \vv_i^\top \bigg)^{-1} \bigg\rangle
& =
\bigg\langle \mX^{-1}, \bigg( \sum_{i \in S} \mX^{-\frac12} \vu_i \vu_i^\top \mX^{-\frac12} \bigg)^{-1} \bigg\rangle \\
& = \bigg\langle \mI, \bigg( \sum_{i \in S} \vu_i \vu_i^\top \bigg)^{-1} \bigg\rangle 
= \tr\bigg( \bigg( \sum_{i \in S} \vu_i \vu_i^\top \bigg)^{-1} \bigg),
\end{aligned}
\end{equation}
and so the termination condition also implies the approximation guarantee as
\[
 \bigg\langle \mX^{-1}, \bigg(\sum_{i \in S} \vv_i \vv_i^\top\bigg)^{-1} \bigg\rangle \leq (1+\eps) \tr(\mX^{-1}) 
\quad \Longrightarrow \quad 
\tr\bigg( \bigg(\sum_{i \in S} \vu_i \vu_i^\top \bigg)^{-1}\bigg) \leq (1+\eps) \tr(\mX^{-1}).
\]
Finally, we consider the knapsack constraints.
Note that the termination conditions of both D/A-design imply $\lambda_{\min}(\mZ_t) < 1$ before the algorithm terminates.
So, we can apply Theorem~\ref{t:cost} with $\gamma=\frac18$ to conclude that the returned solution $S$ satisfies 
\[
\vc_j(S) \leq (1+\eps)\inner{\vc_j}{\vx} + 120d\norm{\vc_j}_{\infty}
\leq (1+\eps)b_j + 120d\norm{\vc_j}_{\infty} 
\leq (1+O(\eps))b_j  
\]
with probability at least $1-\exp(-\Omega(\eps d))$, where the last inequality follows from $b_j \geq \frac{d\norm{\vc_j}_{\infty}}{\eps}$.
\end{proof}

We are ready to prove main theorem in this section by turning the above bicriteria approximation result to a true approximation result using a simple scaling argument.

\DARounding*

\begin{proof}
Let $b_1, \ldots, b_m$ be the input budgets for the $m$ knapsack constraints. 
We scale down the budget to $\tilde{b}_j = \frac{b_j}{1+100 \eps}$ for each $j \in [m]$.
Since $\eps \leq \frac{1}{200}$ and $b_j \geq \frac{2d \|\vc_j\|_\infty}{\eps}$ by the assumption, 
the rescaled budget $\tilde{b}_j \geq \frac{d \|\vc_j\|_\infty}{\eps}$. 
Therefore, the budget assumptions in Theorem~\ref{t:main-knapsack} are satisfied by all $\tilde{b}_1, ..., \tilde{b}_m$. 
In the following, we prove the theorem for D-design only, as the proof for A-design follows by the same argument.

Let $\tilde{\vx} \in [0,1]^n$ be an optimal fractional solution of~\eqref{eq:convex} with budget $\tilde{b}_j$ for $j \in \{1,...,m\}$. Let $\tilde{\mX} := \sum_{i=1}^n \tilde{\vx}(i) \cdot \vv_i \vv_i^\top$ and $\mX = \sum_{i=1}^n \vx(i) \cdot \vv_i \vv_i^\top$.
We run the randomized exchange algorithm with budgets $\tilde{b}_1, ..., \tilde{b}_m$. 
By Theorem~\ref{t:main-knapsack}, with probability at least $1 - O\Big( \frac{k^2}{\eps^2} \cdot e^{-\Omega(\sqrt{d})} \Big)$, the algorithm returns a solution set $S$ within $O(\frac{k}{\eps})$ iterations such that
\[
\det\bigg( \sum_{i \in S} \vu_i \vu_i^\top \bigg)^{\frac{1}{d}} 
\geq (1-10\eps) \cdot \det(\tilde{\mX})^{\frac{1}{d}} 
\geq \frac{1-10\eps}{1+100\eps} \cdot \det(\mX)^{\frac{1}{d}} 
= \big(1-O(\eps)\big) \cdot \det(\mX)^{\frac{1}{d}},
\]
where the second inequality holds as $\frac{1}{1+100\eps} \cdot \mX$ is a feasible solution to~\eqref{eq:convex} with budget $\tilde{b}_1, ..., \tilde{b}_m$.
Furthermore, for each knapsack constraint $j \in [m]$, it follows from Theorem~\ref{t:main-knapsack} that
\[
\vc_j(S) \leq (1+\eps) \tilde{b}_j + 120 d \|\vc_j\|_\infty 
\leq \frac{1+\eps}{1+100\eps} \cdot b_j + 60 \eps b_j \leq b_j,
\]
with probability at least $1-\exp(-\Omega(\eps d))$, where the second inequality follows by the assumption $b_j \geq \frac{2 d\|\vc_j\|_\infty}{\eps}$ and the last inequality follows as $\eps \leq \frac{1}{200}$.
\end{proof}

{\bf Unweighted D/A-Design:} Using the main result, we improve the previous result on D/A-design with a single cardinality constraint by replacing the assumption $b \geq \Omega\left(\frac{d}{\eps} + \frac{1}{\eps^2} \log \left(\frac{1}{\eps}\right) \right)$ in~\cite{SX18,NST19} with $b \geq \frac{2d}{\eps}$, although there is a mild assumption on the range of $\eps$.

\begin{corollary} \label{c:DA}
For any $\frac{1}{200} \geq \eps \geq e^{-\delta \sqrt{d}}$ for a small enough constant $\delta$,
if $b \geq \frac{2d}{\eps}$,
then there is a randomized polynomial time algorithm that returns a $\big(1+O(\eps)\big)$-approximate solution for D/A-design with constant probability.
\end{corollary}
\begin{proof}
We apply Theorem~\ref{t:DA-rounding} on the input.
The probability that the output is a $\big(1+O(\eps)\big)$-approximate solution and satisfies the cardinality constraint is at least $1 - e^{-\Omega(\eps d)} - e^{-\Omega(\sqrt{d})}$ as $k=O(d^2)$.
When $\eps d = \Omega(1)$, this success probability is at least a constant for large enough $d$.
Otherwise, this success probability can be lower bounded by
\[
1 - e^{-\Omega(\eps d)} - e^{-\Omega(\sqrt{d})}
\geq \Omega(\eps d) - e^{-\Omega(\sqrt{d})}
\geq \max\Big\{e^{-\Omega(\sqrt{d})}, \Omega\Big(\frac{d^2}{n}  \Big)  \Big\} \geq \Omega\Big(\frac{d^2}{n}\Big), 
\]
where the first inequality is by $e^{-\Omega(\eps d)} \leq 1 - \Omega(\eps d)$ for $\eps d = o(1)$, and the second inequality is by the assumption $\eps \geq \exp(-\delta \sqrt{d})$ for a small enough $\delta$ and  the fact that we can assume $\eps \geq \frac{2d}{n}$ without loss of generality.
Therefore, we can amplify the success probability to be a constant by applying Theorem~\ref{t:DA-rounding} at most $O\left( \frac{n}{d^2} \right)$ times,
and the total time complexity is still polynomial in $n$ and $d$. 
\end{proof}

{\bf Minimizing Total Effective Resistance:}
We present an application of the main result to the total effective resistance minimization problem.
In this problem, we are given a graph $G=(V,E)$ with Laplacian matrix $\mL_G = \sum_{e \in E} \vb_e \vb_e^\top$ and a cost vector $\vc \in \R_+^m$ on the edges, 
and the goal is to find a subgraph $H$ with cost at most $b$ to minimize the sum of all pairs effective resistances
$\sum_{u,v} \Reff_H(u,v) = n \cdot \tr(\mL_H^{\dagger})$.

\RoundReff*

\begin{proof}
This is an A-design problem by using the vectors $\vb_e$ projecting onto the $(n-1)$-dimensional subspace orthogonal to the all-one vector.
Let $\vx^* \in [0,1]^m$ be an optimal fractional solution to the problem and let $\mL_{\vx^*} := \sum_{e \in E} \vx^*(e) \cdot \vb_e \vb_e^\top$. 
Since $b \geq \Omega\Big( \frac{n\|\vc\|_\infty}{\eps}\Big)$, by Theorem~\ref{t:DA-rounding}, 
there is a randomized algorithm that returns a subgraph $H$ with $\tr\big(\mL_H^{\dagger}\big) \leq \big(1+O(\eps)\big) \cdot \tr\big(\mL_{\vx^*}^\dagger\big)$ within $O\big( \frac{k}{\eps} \big)$ iterations with probability at least $1 - O\big( \frac{k^2}{\eps^2} \cdot e^{-\Omega(\sqrt{n})} \big)$ where $k = O(n^2)$. 
Moreover, the cost constraint is satisfied with probability at least $1 - e^{-\Omega(\eps n)}$.
Since the number of edges $m = O(n^2)$, we can assume $\eps \geq \frac{n}{m} = \Omega(\frac1n)$ without loss of generality. The running time is polynomial in the graph size and the failure probability of the algorithm is at most
$O\big( \frac{k^2}{\eps^2} \cdot e^{-\Omega(\sqrt{n})}  \big) + e^{-\Omega(\eps n)} \leq e^{-\Omega(\sqrt{n})} + e^{-\Omega(1)}$, a constant bounded away from 1, when $n$ is large enough.
\end{proof}

\subsection{Analysis of the D-Design Objective} \label{ss:D-knapsack}

We will prove Theorem~\ref{t:termination} for D-design in this subsection.
Let $\tau_1$ be the start time of the second phase.
For ease of notation, we simply reset $\tau_1=1$,
as the first time step in the second phase.
By assumption, $\lambda_{\min}(\mZ_1) \geq \frac{3}{4}$ and $\lambda_{\min}(\mZ_t) \geq \frac{1}{4}$ for all $t \geq 1$,
which will be crucial in the analysis.

To analyze the objective value for D-design, our plan is to transform the product of random variables in Lemma~\ref{l:Dlower} into a sum of random variables in the exponent as follows,
\begin{eqnarray}
\det(\mZ_{\tau+1}) & \geq & 
\det(\mZ_1) \cdot \prod_{t=1}^\tau \left( 1 - \inner{\vv_{i_t} \vv_{i_t}^\top}{\mZ_t^{-1}} \right) \left( 1 + \inner{\vv_{j_t} \vv_{j_t}^\top}{\mZ_t^{-1}} \right)
\nonumber
\\
& \geq & \det(\mZ_1) \cdot \exp\bigg( \sum_{t=1}^\tau \Big( (1-4\eps) \underbrace{\inner{\vv_{j_t} \vv_{j_t}^\top}{\mZ_t^{-1}}}_{\text{gain $g_t$}} - (1+5\eps) \underbrace{\inner{\vv_{i_t} \vv_{i_t}^\top}{\mZ_t^{-1}}}_{\text{loss $l_t$}} \Big) \bigg),
\label{eq:det-expo-bound}
\end{eqnarray}
where the inequalities $1-x \geq e^{(1-4\eps)x}$ and $1-x \geq e^{-(1+5\eps)x}$ only hold when $x \in [0,4\eps]$ and $\eps$ is small enough such as $\eps \leq \frac{1}{50}$.

So, for our plan to work, we need to bound the gain term $\inner{\vv_{j_t} \vv_{j_t}^\top}{\mZ_t^{-1}}$ and the loss term $\inner{\vv_{i_t} \vv_{i_t}^\top}{\mZ_t^{-1}}$.
To do so, we prove in Lemma~\ref{l:optimality-D} that in an optimal fractional solution $\vx$, 
every vector $\vv_i$ with $0 < \vx(i) < 1$ satisfies the condition that $\norm{\vv_i}_2^2 \leq \eps$.
Recall that, \Cref{obs:prob} implies $0 < \vx(i_t), \vx(j_t) < 1$ for all $t \geq 1$.
Therefore, Lemma~\ref{l:optimality-D} implies that $\norm{\vv_{i_t}}_2^2 \leq \eps$ and $\norm{\vv_{j_t}}_2^2 \leq \eps$ for all $t \geq 1$.
Together with the assumption that $\mZ_t \succeq \frac{1}{4}\mI$ for all $t \geq 1$, we can ensure that $\inner{\vv_{j_t} \vv_{j_t}^\top}{\mZ_t^{-1}} \leq 4\eps$ and $\inner{\vv_{i_t} \vv_{i_t}^\top}{\mZ_t^{-1}} \leq 4\eps$ for all $t \geq 1$,
and hence~\eqref{eq:det-expo-bound} holds.   

Once this transformation is done and \eqref{eq:det-expo-bound} is established, we can apply Freedman's martingale inequality to prove concentration of the exponent.
In the following, we define the gain $g_t$, loss $l_t$ and progress $\Gamma_t$ in the $t$-th iteration as 
\[
g_t := \inner{\vv_{j_t} \vv_{j_t}^\top}{\mZ_t^{-1}}, 
\qquad l_t := \inner{\vv_{i_t} \vv_{i_t}^\top}{\mZ_t^{-1}}, 
\qquad \text{and} \qquad \Gamma_t := (1-4\eps)g_t - (1+5\eps) l_t.
\]
In Section~\ref{ss:exp-progress-D}, we will prove that the expected progress is large if the current solution is far from optimal. 
Then, in Section~\ref{ss:progress-whp-D}, we will prove that the total progress is concentrated around its expectation, where the minimum eigenvalue assumption is crucial in the martingale concentration argument. 
Finally, we finish the proof of Theorem~\ref{t:termination} for D-design in Section~\ref{ss:D-terminate}, and present the proof of Lemma~\ref{l:optimality-D} in Section~\ref{ss:optimality-D}.

\subsubsection{Expected Improvement of the D-Design Objective} \label{ss:exp-progress-D}

Here we bound the conditional expectation of progress $\Gamma_t$, and show that $\E[\Gamma_t]$ is large if the current objective value $\det(\mZ_t)^{\frac{1}{d}}$ is small. 

\begin{lemma} \label{l:D-exp-progress}
Let $\gamma \geq 1$.
Let $S_{t-1}$ be the solution set at time $t$ and $\mZ_t = \sum_{i \in S_{t-1}} \vv_i \vv_i^\top$ for $1 \leq t \leq \tau$.
Suppose $\det(\mZ_t)^{\frac{1}{d}} \leq \lambda$ for $1 \leq t \leq \tau$.
Then 
\[
\sum_{t = 1}^\tau \E[\Gamma_t \mid S_{t-1}] \geq \left( \frac{1-4\eps}{\lambda} - (1+5\eps) \right) \cdot \frac{d \tau}{k}.
\]
\end{lemma}
\begin{proof}
Let $t \in [1, \tau]$.
Using the probability distribution for sampling $\vv_{j_t}$ in the randomized exchange algorithm, 
the expected gain of adding vector $\vv_{j_t}$ is
\begin{eqnarray*}
\E[g_t \mid S_{t-1}]  
& = & \sum_{j \in [n] \backslash S_{t-1}} \frac{\vx(j)}{k} \cdot 
\left(1+ 2\alpha \inner{\vv_j \vv_j^\top}{\mA_t^{\frac12}}\right) \cdot 
\inner{\vv_j \vv_j^\top}{\mZ_t^{-1}} 
\\ 
& \geq & 
\sum_{j \in [n] \backslash S_{t-1}} \frac{\vx(j)}{k} \cdot \inner{\vv_j \vv_j^\top}{\mZ_t^{-1}} 
\\
& = & \frac{1}{k} \bigg( \tr(\mZ_t^{-1}) - \sum_{i \in S_{t-1}} \vx(i) \cdot \inner{\vv_i \vv_i^\top}{\mZ^{-1}_t} \bigg),
\end{eqnarray*}
where the last equality uses $\sum_{j=1}^n \vx(j) \cdot \vv_j \vv_j^\top = \mI$.

Using the probability distribution for sampling $\vv_{i_t}$ in the randomized exchange algorithm, 
the expected loss of removing vector $\vv_{i_t}$ is
\begin{equation} \label{eq:loss_exp}
\begin{aligned}
    \E[l_t \mid S_{t-1}] & = \sum_{i \in S'_{t-1}} \frac{1-\vx(i)}{k} \cdot \left(1-2\alpha \inner{\vv_i \vv_i^\top}{\mA_t^{\frac12}} \right) \cdot \inner{\vv_i \vv_i^\top}{\mZ_t^{-1}} \\
    & \leq \frac{1}{k} \sum_{i \in S'_{t-1}} \big(1-\vx(i)\big) \cdot \inner{\vv_i \vv_i^\top}{\mZ_t^{-1}} \\
    & \leq \frac{1}{k} \sum_{i \in S_{t-1}} \big(1-\vx(i)\big) \cdot \inner{\vv_i \vv_i^\top}{\mZ_t^{-1}} \\
    & = \frac{1}{k}\bigg( d - \sum_{i \in S_{t-1}} \vx(i) \cdot \inner{\vv_i \vv_i^\top}{\mZ^{-1}_t}  \bigg),
\end{aligned}
\end{equation}
where the two inequalities hold as $1-2\alpha \inner{\vv_i \vv_i^\top}{\mA_t^{\frac12}} \leq 1$ and $(1-x_i) \cdot \inner{\vv_i \vv_i^\top}{\mZ_t^{-1}} \geq 0$ for all $i \in [n]$, and the last equality holds as $\sum_{i \in S_{t-1}} \vv_i \vv_i^\top = \mZ_t$.

Therefore, the expected progress is
\begin{align*}
\E[\Gamma_t \mid S_{t-1}] 
&= \E[(1-4\eps) g_t - (1+5\eps) l_t \mid S_{t-1}] 
\\
& \geq \frac{1-4\eps}{k}\bigg( \tr(\mZ_t^{-1}) - \sum_{i \in S_{t-1}} \vx(i) \cdot \inner{\vv_i \vv_i^\top}{\mZ^{-1}_t} \bigg) - \frac{1+5\eps}{k} \bigg( d - \sum_{i \in S_{t-1}} \vx(i) \cdot \inner{\vv_i \vv_i^\top}{\mZ^{-1}_t} \bigg) 
\\
& \geq \frac{1}{k} \left( (1-4\eps) \cdot \tr(\mZ_t^{-1}) - (1+5\eps) \cdot d \right) 
\\
& \geq \frac{1}{k} \bigg( (1-4\eps) \cdot \frac{d}{\det(\mZ_t)^{\frac{1}{d}}} - (1+5\eps) \cdot d \bigg) 
\\
& \geq \bigg(\frac{1-4\eps}{\lambda} - (1+5\eps) \bigg) \cdot \frac{d}{k},
\end{align*}
where the second last inequality follows from Lemma~\ref{l:det}, 
and the last inequality is by the assumption that 
$\max_t \det(\mZ_t)^{\frac{1}{d}} \leq \lambda$. 
The lemma follows by summing over all $1 \leq t \leq \tau$.
\end{proof}

\subsubsection{Martingale Concentration Argument} \label{ss:progress-whp-D}

Here we show that the total progress is concentrated around the expectation.
The proof uses the minimum eigenvalue assumption 
and the short vector condition from Lemma~\ref{l:optimality-D} 
to bound the variance of the random process.

\begin{lemma} \label{l:progress-whp-D}
Suppose $\mZ_t \succeq \frac14 \mI$ and $\norm{\vv_{i_t}}_2^2 \leq \eps$ and $\norm{\vv_{j_t}}_2^2 \leq \eps$ for $\eps \leq \frac{1}{100}$ for all $1 \leq t \leq \tau$. 
Then, for any $\eta > 0$,
\[
\Pr \left[ \sum_{t=1}^\tau \Gamma_t \leq \sum_{t=1}^\tau \E[\Gamma_t \mid S_{t-1}] - \eta \right] 
\leq \exp\Big( -\Omega\Big(\frac{\eta^2 k}{\eps \tau d^{1.5} + \eps \eta k} \Big)\Big).
\]
\end{lemma}
\begin{proof}
We define two sequences of random variables $\{X_t\}_t$ and $\{Y_t\}_t$, where $X_t := \E[\Gamma_t \mid S_{t-1}] - \Gamma_t$ and $Y_t := \sum_{l=1}^t X_l$.
It is easy to check that $\{Y_t\}_t$ is a martingale with respect to $\{S_t\}_t$.
We will use Freedman's inequality to bound $\Pr(Y_\tau \geq \eta)$.

To apply Freedman's inequality, we need to upper bound $X_t$ and $E[X_t^2 \mid S_{t-1}]$.
Note that
\[
0 \leq g_t = \inner{\vv_{j_t} \vv_{j_t}^\top}{\mZ^{-1}_t} \leq 4\eps \quad \text{and} \quad 0 \leq l_t = \inner{\vv_{i_t} \vv_{i_t}^\top}{\mZ_t^{-1}} \leq 4\eps
\]
by our assumptions that $\mZ_t \succcurlyeq \frac14 \mI$ and $\norm{\vv_{i_t}}_2^2 \leq \eps$ and $\norm{\vv_{j_t}}_2^2 \leq \eps$ for $1 \leq t \leq \tau$.
These imply that
\[
X_t = \E[\Gamma_t \mid S_{t-1}] - \Gamma_t
\leq (1-4\eps) \cdot \E[g_t \mid S_{t-1}] + (1+5\eps) \cdot l_t \leq (2+\eps) \cdot 4\eps \leq 10\eps,
\]
where the last inequality holds for $\eps \leq \frac{1}{2}$.

To upper bound $\E[X_t^2 \mid S_{t-1}]$, 
we first upper bound $\E[g_t \mid S_{t-1}]$ and $\E[l_t \mid S_{t-1}]$.
Note that
\begin{eqnarray*}
    \E[g_t \mid S_{t-1}] 
& = & \sum_{j \in [n] \backslash S_{t-1}} \frac{\vx(j)}{k} \cdot \left(1+ 2\alpha \inner{\vv_j \vv_j^\top}{\mA_t^{\frac12}}\right) \cdot \inner{\vv_j \vv_j^\top}{\mZ_t^{-1}} 
\\
& \leq & 
\frac{1+16\eps \sqrt{d}}{k} \cdot \sum_{j \in [n] \backslash S_{t-1}} \vx(j) \cdot \inner{\vv_j \vv_j^\top}{\mZ_t^{-1}} 
\\
& \leq & \frac{1+16\eps \sqrt{d}}{k} \cdot \tr(\mZ_t^{-1})
\\
& \leq & \frac{4d + 64\eps d^{1.5}}{k},
\end{eqnarray*}
where the first inequality holds as $\alpha = 8\sqrt{d}$, $\mA_t \preccurlyeq \mI$ and $\norm{\vv_j}_2^2 \leq \eps$ for $j \in [n]\setminus S_{t-1}$ with $\vx(j)>0$, 
the second inequality follows as $\sum_{i=1}^n \vx(i) \cdot \vv_i \vv_i^\top = \mI$, 
and the last inequality follows from the assumption that $\mZ_t \succcurlyeq \frac14 \mI$. 
Note also that $\E[l_t \mid S_{t-1}] \leq \frac{d}{k}$ from~\eqref{eq:loss_exp} in Lemma~\ref{l:D-exp-progress}. 
So, we can upper bound $\E[X_t^2 \mid S_{t-1}]$ by
\begin{align*} 
\E[X^2_t \mid S_{t-1}] 
\leq 10\eps  \E[|X_t| \mid S_{t-1}]
\leq 20\eps  \Big((1-4\eps) \E[g_t \mid S_{t-1}] + (1+5\eps) \E[l_t \mid S_{t-1}]\Big)
\leq O\left( \frac{\eps d^{1.5}}{k} \right),
\end{align*}
where the first inequality is by the upper bound on $X_t$, 
and the last inequality is by the loose bound that $\E[g_t \mid S_{t-1}] \leq O\left( \frac{d^{1.5}}{k} \right)$.
Therefore, $\sum_{t=1}^\tau \E[X^2_t \mid S_{t-1}] \leq  O\left(\frac{\eps \tau d^{1.5}}{k}\right)$.

Finally, we can apply Freedman's martingale inequality (Theorem~\ref{t:Freedman}) with $R=10\eps$ and $\sigma^2 = O\left(\frac{\eps \tau d^{1.5}}{k}\right)$ to 
conclude that
\[
\Pr(Y_\tau \geq \eta)
\leq \exp\Big(-\frac{\eta^2/2}{\sigma^2 + R\eta/3}\Big)
= \exp\Big( -\Omega\Big(\frac{\eta^2 k}{\eps \tau d^{1.5} + \eps \eta k} \Big)\Big).
\]
The lemma follows by noting that $Y_\tau \geq \eta$ is equivalent to 
$\sum_{t=1}^\tau \Gamma_t \leq \sum_{t=1}^\tau \E[\Gamma_t \mid S_{t-1}] - \eta$. 
\end{proof}

\subsubsection{Proof of Theorem~\ref{t:termination} for D-design} \label{ss:D-terminate}

We are ready to prove Theorem~\ref{t:termination} for D-design.
Let $\tau = \frac{2k}{\eps}$. 
Suppose the second phase of the algorithm has not terminated by time $\tau$. 
Then $\lambda = \max_{1 \leq t \leq \tau+1} \det(Z_t)^{\frac{1}{d}} < 1-10\eps$.
Thus, Lemma~\ref{l:D-exp-progress} implies that 
\begin{align*} %\label{eq:exp-progress-D}
\sum_{t=1}^{\tau} \E \left[ \Gamma_t \mid S_{t-1}\right] \geq \Big(\frac{1-4\eps}{\lambda} - (1+5\eps) \Big) \cdot \frac{d\tau}{k}
\geq \frac{\eps d \tau}{k}
= 2d.
\end{align*}
On the other hand, the initial solution of the second phase satisfies $\mZ_1 \succcurlyeq \frac34 \mI$, which implies that $\det(\mZ_1) \geq \left(\frac34\right)^{d}$. 
As the knapsack constraints satisfy $b_j \geq \frac{d\norm{\vc_j}_{\infty}}{\eps}$ for $j \in [m]$, we know from Lemma~\ref{l:optimality-D} that $\norm{\vv_i}_2^2 \leq \eps$ for each $i$ with $0 < \vx(i) < 1$.
Note that, in the randomized exchange algorithm, all $i_t$ and $j_t$ satisfy $0 < \vx(i_t) < 1$ and $0 < \vx(j_t) < 1$ by \Cref{obs:prob}.
Together with the assumption that $\mZ_t \succeq \frac{1}{4}\mI$ for all $1 \leq t \leq \tau$, we have $\inner{\vv_{j_t} \vv_{j_t}^\top}{\mZ_t^{-1}} \leq 4\eps$ and $\inner{\vv_{i_t} \vv_{i_t}^\top}{\mZ_t^{-1}} \leq 4\eps$ for all $1 \leq t \leq \tau$. 
Hence, we can apply~\eqref{eq:det-expo-bound} to deduce that
\[
1 > \det(\mZ_{\tau+1}) \geq \det(\mZ_1) \cdot \exp\bigg( \sum_{t=1}^{\tau} \Gamma_t \bigg) \geq \Big(\frac34\Big)^d \exp\bigg( \sum_{t=1}^{\tau} \Gamma_t \bigg) \quad \Longrightarrow \quad \sum_{t=1}^{\tau} \Gamma_t \leq d \cdot \ln \frac43 \leq d.
\]
Therefore, we can apply Lemma~\ref{l:progress-whp-D} with $\eta=d$ and $\tau=\frac{2k}{\eps}$ to conclude that
\begin{equation*}
\begin{aligned}
\Pr\left[ \max_{1 \leq t \leq \tau+1} \det(\mZ_t)^{\frac{1}{d}} < 1-10\eps \right]
& \leq \Pr\left[ \sum_{t=1}^{\tau} \Gamma_t < \left( \sum_{t=1}^{\tau} \E \left[ \Gamma_t \mid S_{t-1}\right] \right) - d \right]
\\
& \leq \exp\Bigg( -\Omega\Bigg(\frac{d^2 k}{\eps \Big(\frac{2k}{\eps}\Big) d^{1.5} + \eps d k} \Bigg)\Bigg)
\\
& \leq \exp(-\Omega(\sqrt{d})).
\end{aligned}
\end{equation*}

\subsubsection{Optimality Condition of the Convex Program for D-Design} \label{ss:optimality-D}

The following lemma uses the assumption about the budgets to prove that all vectors with fractional value are short.

\begin{lemma} \label{l:optimality-D}
Let $\vx \in [0,1]^n$ be an optimal fractional solution of the convex programming relaxation~\eqref{eq:convex} for D-design.
Let $\mX = \sum_{i=1}^n \vx(i) \cdot \vu_i \vu_i^\top$,
and $\vv_i = \mX^{-\frac12} \vu_i$ for $1 \leq i \leq n$.
Suppose $b_j \geq \frac{d\norm{c_j}_{\infty}}{\eps}$ for $1 \leq j \leq m$.
Then $\norm{\vv_i}_2^2 \leq \eps$ for each $1 \leq i \leq n$ with $0 < \vx(i) < 1$. 
\end{lemma}
\begin{proof}
We rewrite the convex programming relaxation~\eqref{eq:convex} for D-design as follows:
\begin{equation*}
    \begin{aligned}
        & \underset{\vx \in \R^d, \mX \in \R^{d \times d}}{\rm maximize} & & \log \det\left( \mX \right) \\
        & \text{\rm subject to} & & \mX = \sum_{i=1}^n \vx(i) \cdot \vu_i \vu_i^\top, \\
        & & & \inner{\vc_j}{\vx} \leq b_j, \quad ~\forall j \in [m], \\
        & & & 0 \leq \vx(i) \leq 1, \quad \forall i \in [n].
    \end{aligned} 
\end{equation*}
We will use a dual characterization to investigate the length of the vectors. 
We introduce a dual variable $\mY$ for the first equality constraint, 
a dual variable $\mu_j \geq 0$ for each of the budget constraint $b_j - \inner{\vc_j}{\vx} \geq 0$, 
a dual variable $\beta^-_i\geq 0$ for each non-negative constraint $\vx(i)\geq 0$, and a dual variable $\beta^+_i \geq 0$ for each capacity constraint $1 - \vx(i) \geq 0$.

The Lagrange function $L(\vx, \mX, \mY, \mu, \beta^+, \beta^-)$ is defined as
\begin{eqnarray*}
& & \log \det(\mX) + \bigg\langle \mY, \sum_{i=1}^n \vx(i) \cdot \vu_i \vu_i^\top - \mX \bigg\rangle  + \sum_{j=1}^m \mu_j \Big(b_j - \inner{\vc_j}{\vx}\Big) + \sum_{i=1}^n \beta^-_i \vx(i) + \sum_{i=1}^n \beta^+_i (1-\vx(i))
\\
& = & \log \det(\mX) - \inner{\mY}{\mX} + \sum_{j=1}^m \mu_j b_j  + \sum_{i=1}^n \beta^+_i 
 + \sum_{i = 1}^n \vx(i) \cdot \bigg(\inner{\mY}{\vu_i \vu_i^\top} - \sum_{j=1}^m \mu_j \vc_j(i) + \beta^-_i - \beta^+_i \bigg).
\end{eqnarray*}
The Lagrangian dual program is $\min_{\mY, \mu \geq 0, \beta^+ \geq 0, \beta^- \geq 0} \max_{\vx, \mX} L(\vx,\mX,\mY,\mu,\beta^+,\beta^-)$.
Note that we can assume that $\mY \succeq 0$, as otherwise the inner maximization problem is unbounded above.
Given $\mY \succeq 0, \mu \geq 0, \beta^+ \geq 0, \beta^- \geq 0$, 
the maximizers $\vx,\mX$ of the Lagrange function satisfy the optimality conditions that
\[
\nabla_{\mX} L = \mX^{-1} - \mY = 0 \quad \text{and} \quad \nabla_{\vx(i)} L = \inner{\mY}{\vu_i \vu_i^\top} - \sum_{j=1}^m \mu_j \vc_j(i) + \beta^-_i - \beta^+_i = 0.
\]
Therefore, the Lagrangian dual program can be written as
\begin{equation*}
    \begin{aligned}
        & \underset{\substack{\mY \succeq 0, \mu \geq 0, \\ \beta^+ \geq 0, \beta^- \geq 0}}{\rm minimize} & & \log \det\left( \mY^{-1} \right) - d + \sum_{j=1}^m \mu_j b_j + \sum_{i=1}^n \beta^+_i\\
        & \text{\rm subject to} & & \inner{\mY}{\vu_i \vu_i^\top} = \sum_{j=1}^m \mu_j \vc_j(i) - \beta^-_i + \beta^+_i, \quad \forall i \in [n].
    \end{aligned} 
\end{equation*}
It is easy to verify that $\vx = \delta \vec{1}$ is a strictly feasible solution of the primal program for a small enough $\delta$. 
So, Slater's condition implies that strong duality holds. 
Let $\vx,\mX$ be an optimal solution for the primal program,
and $\mY,\mu,\beta^+,\beta^-$ be an optimal solution for the dual program.
The Lagrangian optimality condition implies that $\mY = \mX^{-1}$,
and it follows that
\begin{align*}
\log \det(\mX) = \log \det(\mX) - d + \sum_{j=1}^m \mu_j b_j + \sum_{i=1}^n \beta^+_i ~\implies~ \sum_{j=1}^m \mu_j b_j \leq d 
 ~\implies~ \sum_{j=1}^m \mu_j \norm{\vc_j}_{\infty} \leq \eps,
\end{align*}
where the last implication follows by the assumption $b_j \geq \frac{d \norm{\vc_j}_{\infty}}{\eps}$ for each $j\in [m]$.

Finally, by the complementary slackness conditions, we have $\beta^-_i \cdot \vx(i) = 0$ and $\beta^+_i \cdot (1-\vx(i)) = 0$ for each $i \in [n]$. 
Therefore, for each $i$ with $0 < \vx(i) < 1$, we must have $\beta^+_i = \beta^-_i = 0$, which implies that
\[
\sum_{j=1}^m \mu_j \vc_j(i) = \inner{\mY}{\vu_i \vu_i^\top} = \inner{\mX^{-1}}{\vu_i \vu_i^\top} = \|\vv_i\|_2^2 \qquad \Longrightarrow \qquad \|\vv_i\|_2^2 \leq \sum_{j=1}^m \mu_j \norm{\vc_j}_{\infty} \leq \eps. \qedhere
\]
\end{proof}

\subsection{Analysis of the A-Design Objective} \label{ss:A-knapsack}

We will prove Theorem~\ref{t:termination} for A-design in this subsection.
Let $\tau_1$ be the start time of the second phase.
For ease of notation, we simply reset $\tau_1=1$, as the first time step in the second phase.
By assumption, $\lambda_{\min}(\mZ_1) \geq \frac{3}{4}$ and $\lambda_{\min}(\mZ_t) \geq \frac{1}{4}$ for all $t \geq 1$,
which will be crucial in the analysis.

To analyze the A-design objective $\tr\big( \big(\sum_{i \in S_{t-1}} \vu_i \vu_i^\top \big)^{-1} \big)$, we analyze the equivalent quantity $\inner{\mX^{-1}}{\mZ_t^{-1}}$ after the linear transformation $\vv_i = \mX^{-\frac12} \vu_i$ as shown in~\eqref{eq:A-obj-transf}. 
By Lemma~\ref{l:trace-rank2}, if $2\inner{\vv_{i_t} \vv_{i_t}^\top}{\mZ_t^{-1}} \leq \frac12$, then the change of the objective value is bounded by
\begin{align*}
\inner{\mX^{-1}}{\mZ_{t+1}^{-1}} & = \inner{\mX^{-1}}{(\mZ_t - \vv_{i_t} \vv_{i_t}^\top + \vv_{j_t} \vv_{j_t}^\top)^{-1}} \\
& \leq \inner{\mX^{-1}}{\mZ_t^{-1}} + \frac{\inner{\mX^{-1}}{\mZ_t^{-1} \vv_{i_t} \vv_{i_t}^\top \mZ_t^{-1}}}{1 - 2\inner{\vv_{i_t} \vv_{i_t}^\top}{\mZ^{-1}_t}} 
- \frac{\inner{\mX^{-1}}{\mZ_t^{-1} \vv_{j_t} \vv_{j_t}^\top \mZ_t^{-1}}}{1 + 2\inner{\vv_{j_t} \vv_{j_t}^\top}{\mZ^{-1}_t}}.
\end{align*}
In Section~\ref{ss:comb-A}, we show in Lemma~\ref{l:gain-A} and Lemma~\ref{l:loss-A} that if we sample $i_t$ and $j_t$ from the distributions 
\[
\Pr[ i_t = i ] \propto \big(1-\vx(i)\big) \cdot \left(1-2\inner{\vv_i \vv_i^\top}{\mZ_t^{-1}}\right) 
\quad {\rm and} \quad 
\Pr[j_t = j] \propto \vx(j) \cdot \left(1+2\inner{\vv_j \vv_j^\top}{\mZ_t^{-1}}\right),
\] 
then the objective value will improve in expectation when the current objective value is far from optimal.
In the randomized exchange algorithm, however, we sample $i_t$ and $j_t$ from the E-design distributions.
An important observation is that the quantities in these two distributions can be related when the minimum eigenvalue assumption holds.
The following lemma will be proved in Section~\ref{ss:exp-progress-A}.

\begin{lemma} \label{l:E-A}
If $\mZ_t \succcurlyeq \frac14 \mI$, then
$\inner{\vv_i \vv_i^\top}{\mZ_t^{-1}} \leq \alpha \cdot \inner{\vv_i \vv_i^\top}{\mA_t^{\frac12}} \leq \alpha \lambda_{\min}(\mZ_t) \cdot \inner{\vv_i \vv_i^\top}{\mZ_t^{-1}}$ for $1 \leq i \leq n$.
\end{lemma}

In the exchange subroutine of the randomized exchange algorithm, only those $i_t$ with $2\alpha \inner{\vv_{i_t} \vv_{i_t}^\top}{\mA_t^{\frac12}} \!\leq\! \frac{1}{2}$ are sampled.
So, when the minimum eigenvalue assumption holds,
Lemma~\ref{l:E-A} implies that the randomized exchange algorithm only samples $i_t$ that satisfies $2\inner{\vv_{i_t} \vv_{i_t}^\top}{\mZ_t^{-1}} \leq \frac12$.
Therefore, we can apply Lemma~\ref{l:trace-rank2} repeatedly to obtain that for any $\tau \geq 1$,
\begin{align} \label{eq:trace-bound}
    \inner{\mX^{-1}}{\mZ_{\tau+1}^{-1}} \leq \inner{\mX^{-1}}{\mZ_1^{-1}} - \sum_{t = 1}^\tau \bigg( \frac{\inner{\mX^{-1}}{\mZ_t^{-1} \vv_{j_t} \vv_{j_t}^\top \mZ_t^{-1}}}{1 + 2\inner{\vv_{j_t} \vv_{j_t}^\top}{\mZ^{-1}_t}} - \frac{\inner{\mX^{-1}}{\mZ_t^{-1} \vv_{i_t} \vv_{i_t}^\top \mZ_t^{-1}}}{1 - 2\inner{\vv_{i_t} \vv_{i_t}^\top}{\mZ^{-1}_t}} \bigg).
\end{align}

As in Section~\ref{ss:D-knapsack}, we define gain $g_t$, loss $l_t$ and progress $\Gamma_t$ in the $t$-th iteration as follows
\[
g_t := \frac{\inner{\mX^{-1}}{\mZ_t^{-1} \vv_{j_t} \vv_{j_t}^\top \mZ_t^{-1}}}{1 + 2\inner{\vv_{j_t} \vv_{j_t}^\top}{\mZ^{-1}_t}}, \quad l_t := \frac{\inner{\mX^{-1}}{\mZ_t^{-1} \vv_{i_t} \vv_{i_t}^\top \mZ_t^{-1}}}{1 - 2\inner{\vv_{i_t} \vv_{i_t}^\top}{\mZ^{-1}_t}}, \quad \text{and} \quad \Gamma_t := g_t - l_t.
\]
In Section~\ref{ss:exp-progress-A}, we will prove Lemma~\ref{l:E-A}, and use it to prove that the expected progress is large if the current objective value is far from optimal. 
Then, in Section~\ref{ss:progress-whp-A}, we will prove that the total progress is concentrated around its expectation, while the minimum eigenvalue assumption and the optimality condition of the convex programming relaxation are crucial in the martingale concentration argument. 
Finally, we complete the proof of Theorem~\ref{t:termination} for A-design in Section~\ref{ss:A-terminate}, 
and present the proof of the optimality condition in Section~\ref{ss:optimality-A}.

\subsubsection{Expected Improvement of the A-Design Objective} \label{ss:exp-progress-A}

We first prove Lemma~\ref{l:E-A}, which will also be needed in bounding the expectation.

\begin{proofof}{Lemma~\ref{l:E-A}}
Recall that $\mA_t = (\alpha \mZ_t - l_t \mI)^{-2}$ where $l_t$ is the unique value such that $\mA_t \succ 0$ and $\tr(\mA_t) = 1$. 
Since $\mZ_t \succcurlyeq \lambda_{\min}(\mZ_t) \cdot \mI$, it follows that
\[
1 = \tr(\mA_t) \leq \left(\alpha \lambda_{\min}(\mZ_t) - l_t\right)^{-2} \cdot \tr(\mI) 
\quad \Longrightarrow \quad 
\alpha \lambda_{\min}(\mZ_t) - l_t \leq \sqrt{d} 
\quad \Longrightarrow \quad 
l_t \geq 0,
\]
where the last implication holds as $\alpha = 8\sqrt{d}$ and $\lambda_{\min}(\mZ_t) \geq \frac14$.
This implies that $\mA^{\frac12}_t = (\alpha \mZ_t - l_t \mI)^{-1} \succcurlyeq \alpha^{-1} \mZ^{-1}_t$, proving the first inequality.

For the second inequality, consider the eigen-decomposition of $\mZ_t = \sum_{j =1}^d \lambda_j \vw_j \vw_j^\top$, 
where $0 < \lambda_1 \leq \ldots \leq \lambda_d$ are the eigenvalues and $\{\vw_j\}$ are the corresponding orthonormal eigenvectors.
Then,
\begin{align*}
    \frac{\inner{\vv_i \vv_i^\top}{\mA_t^{\frac12}}}{\inner{\vv_i \vv_i^\top}{\mZ_t^{-1}}} & = \frac{\sum_{j=1}^d \frac{\inner{\vv_i}{\vw_j}^2}{\alpha \lambda_j - l_t}}{\sum_{j=1}^d \frac{\inner{\vv_i}{\vw_j}^2}{\lambda_j}} \leq \max_{j \in [d]} \frac{\lambda_j}{\alpha \lambda_j - l_t} \leq \frac{\lambda_1}{\alpha \lambda_1 - l_t} \leq \lambda_1,
\end{align*}
where the first inequality holds since $\alpha \lambda_j - l_t > 0$ as $\mA_t \succ 0$, 
the second inequality holds as $l_t \geq 0$ and the function $f(x) = \frac{x}{\alpha x - l_t}$ is decreasing for $x \geq \frac{l_t}{\alpha}$ when $l_t \geq 0$,
and the last inequality follows as $1=\tr(\mA_t) \geq (\alpha \lambda_1 - l_t)^{-2}$ which implies $\alpha \lambda_1 - l_t \geq 1$. 
\end{proofof}

The following lemma shows that the expected progress is large if the current objective value is far from optimal.
Note that, in contrast to Section~\ref{ss:D-knapsack} for D-design, the minimum eigenvalue assumption is needed in the proof.

\begin{lemma} \label{l:A-exp-progress}
Let $\lambda \geq 1$.
Let $S_{t-1}$ be the solution set at time $t$ and $\mZ_t = \sum_{i \in S_{t-1}} \vv_i \vv_i^\top$ for $1 \leq t \leq \tau$.
Suppose $\mZ_t \succcurlyeq \frac14 \mI$ and $\inner{\mX^{-1}}{\mZ^{-1}_t} \geq \lambda \cdot \tr\left( \mX^{-1} \right)$ for $1 \leq t \leq \tau$.
Then
\[
\sum_{t=1}^\tau \E[\Gamma_t \mid S_{t-1}] \geq \frac{(\lambda - 1) \tau}{k} \cdot \tr(\mX^{-1}).
\]
\end{lemma}

\begin{proof}
The expected gain of adding vector $\vv_{j_t}$ is
\begin{eqnarray*}
\E[g_t | S_{t-1}] 
& = & \sum_{j \in [n] \backslash S_{t-1}} \frac{\vx(j)}{k} 
\cdot \Big(1+ 2\alpha \inner{\vv_j \vv_j^\top}{\mA_t^{\frac12}}\Big) 
\cdot \frac{\inner{\mX^{-1}}{\mZ_t^{-1} \vv_j \vv_j^\top \mZ_t^{-1}}}{1+2\inner{\vv_j \vv_j^\top}{\mZ_t^{-1}}} 
\\
& \geq & \sum_{j \in [n] \backslash S_{t-1}} \frac{\vx(j)}{k} \cdot \inner{\mX^{-1}}{\mZ_t^{-1} \vv_j \vv_j^\top \mZ_t^{-1}} 
\\
& = & \frac{1}{k} \bigg( \inner{\mX^{-1}}{\mZ_t^{-2}} - \bigg\langle \mX^{-1}, \mZ_t^{-1}\bigg( \sum_{i \in S_{t-1}} \vx(i) \cdot \vv_i \vv_i^\top \bigg) \mZ_t^{-1} \bigg\rangle \bigg),
\end{eqnarray*}
where the inequality follows from Lemma~\ref{l:E-A} and the last equality follows as $\sum_{j=1}^n \vx(j) \cdot \vv_j \vv_j^\top = \mI$.

The expected loss of removing vector $\vv_{i_t}$ is
\begin{eqnarray} 
\E[l_t \mid S_{t-1}] 
& = & \sum_{i \in S'_{t-1}} \frac{1-\vx(i)}{k} 
\cdot \Big(1-2\alpha \inner{\vv_i \vv_i^\top}{\mA_t^{\frac12}}\Big) 
\cdot \frac{\inner{\mX^{-1}}{\mZ_t^{-1} \vv_i \vv_i^\top \mZ_t^{-1}}}{1-2\inner{\vv_i \vv_i^\top}{\mZ_t^{-1}}}  \nonumber
\\
& \leq & \frac{1}{k} \sum_{i \in S'_{t-1}} \big(1-\vx(i)\big) \cdot \inner{\mX^{-1}}{\mZ_t^{-1} \vv_i \vv_i^\top \mZ_t^{-1}} \nonumber
\\
& \leq & \frac{1}{k} \sum_{i \in S_{t-1}} \big(1-\vx(i)\big) \cdot \inner{\mX^{-1}}{\mZ_t^{-1} \vv_i \vv_i^\top \mZ_t^{-1}} \nonumber
\\
& = & \frac{1}{k}\bigg( \inner{\mX^{-1}}{\mZ_t^{-1}} - 
\bigg\langle \mX^{-1}, \mZ_t^{-1} \bigg( \sum_{i \in S_{t-1}} \vx(i) \cdot \vv_i \vv_i^\top \bigg) \mZ_t^{-1} \bigg\rangle \bigg),
\label{eq:loss_exp-A} 
\end{eqnarray}
where the first inequality follows from Lemma~\ref{l:E-A} and $2\alpha \inner{\vv_i \vv_i^\top}{\mA_t^{\frac12}} \leq \frac{1}{2}$ by the definition of $S'_{t-1}$, 
and the last equality holds as $\sum_{i \in S_{t-1}} \vv_i \vv_i^\top = \mZ_t$.

Therefore, the expected progress is
\begin{align*}
\E[\Gamma_t \mid S_{t-1}] = \E[g_t \mid S_{t-1}] - \E[l_t \mid S_{t-1}]& \geq \frac{1}{k}\left( \inner{\mX^{-1}}{\mZ_t^{-2}} - \inner{\mX^{-1}}{\mZ_t^{-1}} \right).
\end{align*}
The term $\inner{\mX^{-1}}{\mZ_t^{-2}}$ can be lower bounded by
\[
\inner{\mX^{-1}}{\mZ_t^{-2}} \geq \frac{\inner{\mX^{-1}}{\mZ_t^{-1}}^2}{\tr(\mX^{-1})} 
\geq \lambda \cdot \inner{\mX^{-1}}{\mZ_t^{-1}},
\]
where the first inequality follows from~\eqref{eq:inner_upper} in Lemma~\ref{l:trace}, and the second inequality follows from our assumption.
This implies that 
\[
\E[\Gamma_t \mid S_{t-1}]  
\geq \frac{\lambda - 1}{k} \cdot \inner{\mX^{-1}}{\mZ_t^{-1}} 
= \frac{\lambda - 1}{k} \cdot \tr\bigg( \bigg(\sum_{i \in S_{t-1}} \vu_i \vu_i^\top \bigg)^{-1} \bigg)  
\geq \frac{\lambda - 1}{k} \cdot \tr(\mX^{-1}),
\]
where the equality is from~\eqref{eq:A-obj-transf} and the last inequality is because $\mX$ is an optimal solution. 
The lemmas follows by summing over $\tau$.
\end{proof}

\subsubsection{Martingale Concentration Argument} \label{ss:progress-whp-A}

Here we prove that the total progress is concentrated around the expectation.
The proof uses the minimum eigenvalue assumption and the optimality condition in Lemma~\ref{l:optimality-A} in Section~\ref{ss:optimality-A} to bound the variance of the random process.

\begin{lemma} \label{l:progress-whp-A}
Suppose $\mZ_t \succeq \frac{1}{4} \mI$ and $\inner{\mX^{-1}}{\vv_{i_t}\vv_{i_t}^\top} \leq \frac{\eps}{d} \cdot \tr(\mX^{-1})$ and $\inner{\mX^{-1}}{\vv_{j_t}\vv_{j_t}^\top} \leq \frac{\eps}{d} \cdot \tr(\mX^{-1})$ for all $1 \leq t \leq \tau$. 
Then, for any $\eta > 0$, 
\[
\Pr \left[ \sum_{t=1}^\tau \Gamma_t \leq \sum_{t=1}^\tau \E[\Gamma_t \mid S_{t-1}] - \eta \right] 
\leq \exp\bigg( -\Omega\bigg(\frac{\eta^2 k d}{\eps \tau \sqrt{d} \cdot \tr(\mX^{-1})^2 + \eps \eta k \cdot \tr(\mX^{-1})} \bigg)\bigg).
\]
\end{lemma}

\begin{proof}
We define two sequences of random variables $\{X_t\}_t$ and $\{Y_t\}_t$, where $X_t := \E[\Gamma_t \mid S_{t-1}] - \Gamma_t$ and $Y_t := \sum_{l=1}^t X_l$.
It is easy to check that $\{Y_t\}_t$ is a martingale with respect to $\{S_t\}_t$.
We will use Freedman's inequality to bound $\Pr(Y_\tau \geq \eta)$.

To apply Freedman's inequality, we need to upper bound $X_t$ and $E[X_t^2 \mid S_{t-1}]$.
To upper bound $X_t$, we first prove an upper bound on $g_t$ and $l_t$.
Note that
\begin{align*}
\inner{\mX^{-1}}{\mZ^{-1}_t \vv_{i_t} \vv_{i_t}^\top \mZ_t^{-1}} & = \inner{\mZ_t^{-1} \mX^{-1} \mZ_t^{-1}}{\vv_{i_t} \vv_{i_t}^\top} = \bigg\langle \mX^{\frac12} \bigg(\sum_{j \in S_{t-1}} \vu_j \vu_j^\top\bigg)^{-2} \mX^{\frac12}, \vv_{i_t} \vv_{i_t}^\top \bigg\rangle \\
& = \inner{\mX^{-\frac12} \mZ_t^{-2} \mX^{-\frac12}}{\vv_{i_t} \vv_{i_t}^\top} \leq 16 \inner{\mX^{-1}}{\vv_{i_t} \vv_{i_t}^\top} \leq \frac{16\eps }{d} \cdot \tr(\mX^{-1}),
\end{align*}
where the second equality uses the fact that $\mZ_t = \mX^{-\frac12} \left( \sum_{j \in S_{t-1}} \vu_j \vu_j^\top \right) \mX^{-\frac12}$, the first inequality uses the assumption $\mZ_t \succeq \frac14 \mI$, and the last inequality follows from the assumption that $\inner{\mX^{-1}}{\vv_{i_t}\vv_{i_t}^\top} \leq \frac{\eps}{d} \cdot \tr(\mX^{-1})$. 
This implies that
\[
g_t = \frac{\inner{\mX^{-1}}{\mZ_t^{-1} \vv_{j_t} \vv_{j_t}^\top \mZ_t^{-1}}}{1+2\inner{\vv_{j_t} \vv_{j_t}^\top}{\mZ_t^{-1}}} \leq \frac{16\eps}{d} \cdot \tr(\mX^{-1}) 
\quad \text{and} \quad 
l_t = \frac{\inner{\mX^{-1}}{\mZ_t^{-1} \vv_{i_t} \vv_{i_t}^\top \mZ_t^{-1}}}{1-2\inner{\vv_{i_t} \vv_{i_t}^\top}{\mZ_t^{-1}}} \leq \frac{32 \eps}{d} \cdot \tr(\mX^{-1}),
\]
where the second inequality holds as $2\inner{\vv_{i_t} \vv_{i_t}^\top}{\mZ_t^{-1}} \leq 2\alpha \inner{\vv_{i_t} \vv_{i_t}^\top}{\mA_t^{\frac12}} \leq \frac12$ 
by Lemma~\ref{l:E-A} and the definition that $i_t \in S'_{t-1}$ in the exchange subroutine. 
Therefore,
\[
X_t = \E[\Gamma_t \mid S_{t-1}] - \Gamma_t
\leq \E[g_t \mid S_{t-1}] + l_t \leq \frac{48\eps}{d} \cdot \tr(\mX^{-1}).
\]
Next, we upper bound $\E[X_t^2 \mid S_{t-1}]$ by
\begin{equation*} 
\begin{aligned} 
\E[X^2_t \mid S_{t-1}] 
\leq \frac{48\eps}{d} \cdot \tr(\mX^{-1}) \cdot \E[|X_t| \mid S_{t-1}]
\leq \frac{96\eps}{d} \cdot \tr(\mX^{-1}) \cdot  \Big(\E[g_t \mid S_{t-1}] + \E[l_t \mid S_{t-1}]\Big).
\end{aligned}
\end{equation*}
Using~\eqref{eq:loss_exp-A}, we bound the expected loss term by
\[
\E[l_t \mid S_{t-1}] 
\leq \frac{1}{k} \cdot \inner{\mX^{-1}}{\mZ_t^{-1}} 
\leq \frac{4}{k} \cdot \tr(\mX^{-1}),
\]
where the last inequality follows by the assumption that $\mZ_t \succcurlyeq \frac14 \mI$. 
Then, we bound the expected gain term by
\begin{eqnarray*}
\E[g_t \mid S_{t-1}] 
& = & \sum_{j \in [n] \backslash S_{t-1}} \frac{\vx(j)}{k} 
\cdot \Big( 1+ 2\alpha \inner{\vv_j \vv_j^\top}{\mA_t^{\frac12}} \Big) 
\cdot \frac{\inner{\mX^{-1}}{\mZ_t^{-1} \vv_j \vv_j^\top \mZ_t^{-1}}}{1+2\inner{\vv_j \vv_j^\top}{\mZ_t^{-1}}} 
\\
& \leq & \frac{1}{k} \cdot \max_{j \in [n]} 
\bigg\{ \frac{\alpha \inner{\vv_j \vv_j^\top}{\mA_t^{\frac12}}}{\inner{\vv_j \vv_j^\top}{\mZ_t^{-1}}} \bigg\} 
\cdot \sum_{j =1}^n \vx(j) \cdot \inner{\mX^{-1}}{\mZ_t^{-1} \vv_j \vv_j^\top \mZ_t^{-1}} 
\\
& \leq & \frac{1}{k} \cdot \alpha \lambda_{\min}(\mZ_t) \cdot \inner{\mX^{-1}}{\mZ_t^{-2}} 
\\
& \leq & \frac{32\sqrt{d}}{k} \cdot \tr\big(\mX^{-1}\big),
\end{eqnarray*}
where the first inequality follows from the first inequality in Lemma~\ref{l:E-A},
and the second inequality follows from the second inequality in Lemma~\ref{l:E-A} and $\sum_{j=1}^n \vx(j) \cdot \vv_j \vv_j^\top = \mI$, 
and the last inequality holds as $\alpha = 8\sqrt{d}$, $\mZ_t^{-2} \preccurlyeq \lambda_{\min}(\mZ_t)^{-2} \mI$, 
and $\lambda_{\min}(\mZ_t) \geq \frac14$ by our assumption.
Therefore, 
\[
\E[X^2_t \mid S_{t-1}] \leq O\Big( \frac{\eps}{k\sqrt{d}} \Big) \cdot \tr\big(\mX^{-1}\big)^2
\quad \implies \quad
\sum_{t=1}^\tau \E[X^2_t \mid S_{t-1}] \leq  O\Big( \frac{\eps \tau}{k\sqrt{d}}\Big) \cdot \tr\big(\mX^{-1}\big)^2.
\]
Finally, we can apply Freedman's martingale inequality Theorem~\ref{t:Freedman} with $R=\frac{48\eps}{d} \cdot \tr(\mX^{-1})$ and $\sigma^2 = O\big( \frac{\eps \tau}{k\sqrt{d}}\big) \cdot \tr\big(\mX^{-1}\big)^2$ to conclude that
\[
\Pr(Y_\tau \geq \eta)
\leq \exp\left(-\frac{\eta^2/2}{\sigma^2 + R\eta/3}\right)
= \exp\left( -\Omega\left(\frac{\eta^2 k d}{\eps \tau \sqrt{d} \tr(\mX^{-1})^2 + \eps \eta k \tr(\mX^{-1})} \right)\right).
\]
The lemma follows by noting that $Y_\tau \geq \eta$ is equivalent to 
$\sum_{t=1}^\tau \Gamma_t \leq \sum_{t=1}^\tau \E[\Gamma_t \mid S_{t-1}] - \eta$. 
\end{proof}

\subsubsection{Proof of Theorem~\ref{t:termination} for A-Design} \label{ss:A-terminate}

We are ready to prove Theorem~\ref{t:termination} for A-design.
Let $\tau = \frac{2k}{\eps}$. 
Suppose the second phase of the algorithm has not terminated by time $\tau$. 
Then $\lambda = \min_{1 \leq t \leq \tau+1} \frac{\inner{\mX^{-1}}{\mZ_t^{-1}}}{\tr\big(\mX^{-1} \big)} > (1+\eps)$.
Thus, Lemma~\ref{l:A-exp-progress} implies that 
\begin{align*} 
\sum_{t=1}^{\tau} \E \left[ \Gamma_t \mid S_{t-1}\right] \geq \frac{(\lambda - 1) \tau}{k} \cdot \tr\big(\mX^{-1}\big)
> 2 \tr\left(\mX^{-1}\right).
\end{align*}
On the other hand, the initial solution of the second phase satisfies $\mZ_1 \succcurlyeq \frac34 \mI$, which implies that $\inner{\mX^{-1}}{\mZ_1^{-1}} \leq \frac43 \tr(\mX^{-1})$.
By the minimum eigenvalue assumption, we know from Lemma~\ref{l:E-A} that 
$2\inner{\vv_{i_t} \vv_{i_t}^\top}{\mZ_t^{-1}} \leq 2\alpha \cdot \inner{\vv_{i_t} \vv_{i_t}^\top}{\mA_t^{\frac12}} \leq \frac{1}{2}$,
and so we can apply~\eqref{eq:trace-bound} to deduce that
\[
\tr(\mX^{-1}) \leq \inner{\mX^{-1}}{\mZ_{\tau+1}^{-1}} \leq \inner{\mX^{-1}}{\mZ_1^{-1}} - \sum_{t=1}^\tau \Gamma_t \leq \frac43 \cdot \tr\big(\mX^{-1}\big) - \sum_{t=1}^\tau \Gamma_t \quad \Longrightarrow \quad \sum_{t=1}^\tau \Gamma_t \leq \frac13 \cdot \tr\big(\mX^{-1}\big). 
\]
As the knapsack constraints satisfy $b_j \geq \frac{d \norm{\vc_j}_{\infty}}{\eps}$ for $j \in [m]$, 
the optimality conditions in Lemma~\ref{l:optimality-A} imply that  
$\inner{\mX^{-1}}{\vv_i\vv_i^\top} \leq \frac{\eps}{d} \cdot \tr(\mX^{-1})$
for each $i$ with $0 < \vx(i) < 1$.
Note that, in the randomized exchange algorithm, all $i_t$ and $j_t$ satisfy $0 < \vx(i_t) < 1$ and $0 < \vx(j_t) < 1$ by \Cref{obs:prob}.
Therefore, we can apply Lemma~\ref{l:progress-whp-A} with $\eta=\frac53 \cdot \tr(\mX^{-1})$ and $\tau=\frac{2k}{\eps}$ to conclude that
\begin{equation*}
\begin{aligned}
\Pr\left[\min_{1 \leq t \leq \tau+1} \inner{\mX^{-1}}{\mZ_t^{-1}} > (1+\eps) \tr(\mX^{-1})  \right] & \leq \Pr\left[ \sum_{t=1}^{\tau} \Gamma_t <  \sum_{t=1}^{\tau} \E \left[ \Gamma_t \mid S_{t-1}\right] - \frac{5}{3} \cdot \tr(\mX^{-1}) \right]
\\
& \leq \exp\bigg( -\Omega\bigg(\frac{\tr(\mX^{-1})^2 \cdot k d }{\eps \left(\frac{2k}{\eps}\right) \sqrt{d} \cdot \tr\left(\mX^{-1}\right)^2 + \eps k \cdot \tr\left(\mX^{-1}\right)^2} \bigg)\bigg)
\\
& \leq \exp\left(-\Omega(\sqrt{d})\right).
\end{aligned}
\end{equation*}

\subsubsection{Optimality Condition for the Convex Program of A-Design} \label{ss:optimality-A}

This lemma follows from the optimality condition of the convex programming relaxation and the assumption about the budgets.

\begin{lemma} \label{l:optimality-A}
Let $\vx \in [0,1]^n$ be an optimal fractional solution of the convex programming relaxation~\eqref{eq:convex} for A-design.
Let $\mX = \sum_{i=1}^n \vx(i) \cdot \vu_i \vu_i^\top$,
and $\vv_i = \mX^{-\frac12} \vu_i$ for $1 \leq i \leq n$.
Suppose $b_j \geq \frac{d\norm{c_j}_{\infty}}{\eps}$ for $1 \leq j \leq m$.
Then, for each $1 \leq i \leq n$ with $0 < \vx(i) < 1$,  
\[
\inner{\mX^{-1}}{\vv_i \vv_i^\top} \leq \frac{\eps}{d} \cdot \tr(\mX^{-1}).
\]
\end{lemma}

\begin{proof}
We rewrite the convex relaxation~\eqref{eq:convex} for A-design as follows.
\begin{equation*}
    \begin{aligned}
        & \underset{\vx \in \R^d, \mX \in \R^{d \times d}}{\rm minimize} & & \tr\big(\mX^{-1}\big) \\
        & \text{\rm subject to} & & \mX = \sum_{i=1}^n \vx(i) \cdot \vu_i \vu_i^\top\\
        & & & \inner{\vc_j}{\vx} \leq b_j, ~~\quad \forall j \in [m], \\
        & & & 0 \leq \vx(i) \leq 1, \quad \forall i \in [n].
    \end{aligned} 
\end{equation*}
We will use a dual characterization to prove the lemma.
We introduce a dual variable $\mY$ for the first equality constraint, 
a dual variable $\mu_j \geq 0$ for each of the budget constraint $\inner{\vc_j}{\vx} - b_j \leq 0$, 
a dual variable $\beta^-_i \geq 0$ for each non-negative constraint $-\vx(i) \leq 0$,
and a dual variable $\beta^+_i \geq 0$ for each capacity constraint $\vx(i) - 1 \leq 0$.

The Lagrange function $L(\vx, \mX, \mY, \mu, \beta^+, \beta^-)$ is defined as
\begin{eqnarray*}
& & \tr\big(\mX^{-1}\big) 
+ \bigg\langle \mY, \mX - \sum_{i=1}^n \vx(i) \cdot \vu_i \vu_i^\top \bigg\rangle 
+ \sum_{j=1}^m \mu_j \Big(\inner{\vc_j}{\vx} - b_j\Big)
 - \sum_{i=1}^n \beta^-_i \vx(i) + \sum_{i=1}^n \beta^+_i (\vx(i)-1) 
\\
& = & \tr\big(\mX^{-1}\big) + \inner{\mY}{\mX} - \sum_{j=1}^m \mu_j b_j - \sum_{i=1}^n \beta^+_i
- \sum_{i=1}^n \vx(i) \cdot \bigg(\inner{\mY}{\vu_i \vu_i^\top} - \sum_{j=1}^m \mu_j \vc_j(i) + \beta^-_i - \beta^+_i \bigg).
\end{eqnarray*}
%HERE
The Lagrangian dual program is $\max_{\mY, \mu \geq 0, \beta^+ \geq 0, \beta^- \geq 0} \min_{\vx, \mX} L(\vx,\mX,\mY,\mu,\beta^+,\beta^-)$.
Note that we can assume that $\mY \succeq 0$, as otherwise the inner minimization problem is unbounded below.
Given $\mY \succcurlyeq 0, \mu, \beta^+, \beta^- \geq 0$, the minimizers $\vx,\mX$ of the Lagrange function satisfy the optimality conditions that
\[
\nabla_{\mX} L = -\mX^{-2} + \mY = 0 \quad \text{and} \quad \nabla_{\vx(i)} L = -\inner{\mY}{\vu_i \vu_i^\top} + \sum_{j=1}^m \mu_j \vc_j(i) - \beta^-_i + \beta^+_i = 0.
\]
Therefore, the Lagrangian dual program can be written as
\begin{equation*}
    \begin{aligned}
        & \underset{\substack{\mY \succcurlyeq 0, \mu \geq 0, \\ \beta^+ \geq 0, \beta^- \geq 0}}{\rm maximize} & & 2\tr(\mY^{\frac12}) - \sum_{j=1}^m \mu_j b_j - \sum_{i=1}^n \beta^+_i\\
        & \text{\rm subject to} & & \inner{\mY}{\vu_i \vu_i^\top} = \sum_{j=1}^m \mu_j \vc_j(i) - \beta^-_i + \beta^+_i, \quad \forall i \in [n].
    \end{aligned} 
\end{equation*}
It is easy to verify that $\vx = \delta \vec{1}$ is a strictly feasible solution of the primal program for a small enough $\delta$. 
So, Slater's condition implies that strong duality holds. 
Let $\vx,\mX$ be an optimal solution for the primal program,
and $\mY,\mu,\beta^+,\beta^-$ be an optimal solution for the dual program.
The Lagrangian optimality condition implies that $\mY = \mX^{-2}$,
and it follows that
\begin{align*}
\tr(\mX^{-1}) = 2\tr(\mX^{-1}) - \sum_{j=1}^m \mu_j b_j - \sum_{i=1}^n \beta^+_i 
\quad & \Longrightarrow \quad 
\sum_{j=1}^m \mu_j b_j  + \sum_{i=1}^n \beta^+_i = \tr(\mX^{-1}) 
\\
& \Longrightarrow \quad 
\sum_{j=1}^m \mu_j \norm{\vc_j}_{\infty} \leq \frac{\eps}{d} \cdot \tr(\mX^{-1}),
\end{align*}
where the last implication follows from $\beta^+_i \geq 0$ for all $i \in [n]$ and the assumption $b_j \geq \frac{d \norm{\vc_j}_{\infty}}{\eps}$ for all $j \in [m]$.

Finally, by the complementary slackness conditions, we have $\beta^-_i \cdot \vx(i) = 0$ and $\beta^+_i \cdot (1-\vx(i)) = 0$ for all $i \in [n]$. 
Therefore, for each $i \in [n]$ with $0 < \vx(i) < 1$, we must have $\beta^+_i = \beta^-_i = 0$, which implies that
\[
\sum_{j=1}^m \mu_j \vc_j(i) = \inner{\mX^{-2}}{\vu_i \vu_i^\top} = \inner{\mX^{-1}}{\vv_i \vv_i^\top} 
\quad \Longrightarrow \quad 
\inner{\mX^{-1}}{\vv_i \vv_i^\top} \leq \sum_{j=1}^m \mu_j \norm{\vc_j}_{\infty} \leq \frac{\eps}{d} \tr(\mX^{-1}). \qedhere
\]
\end{proof}

\section*{Acknowledgment}

We would like to thank Tsz Chiu Kwok, Akshay Ramachandran, and Mohit Singh for many useful discussions and helpful comments on the writing.

\bibliographystyle{siamplain}

\end{document}